\documentclass[12pt]{llncs}

\usepackage{a4wide,times}

\usepackage{pgf,tikz}
\pgfrealjobname{thesis}
\usetikzlibrary{arrows}
\usetikzlibrary{positioning}
\usetikzlibrary{shapes.geometric}
\usetikzlibrary{shapes.misc}
\usetikzlibrary{fit}
\usetikzlibrary{patterns}
\usetikzlibrary{decorations.pathreplacing}

\makeatletter
\renewenvironment{table}
               {\setlength\abovecaptionskip{10\p@}%
                \setlength\belowcaptionskip{10\p@}%
                \@float{table}}
               {\end@float}
\renewenvironment{table*}
               {\setlength\abovecaptionskip{10\p@}%
                \setlength\belowcaptionskip{10\p@}%
                \@dblfloat{table}}
               {\end@dblfloat}
\makeatother

\usepackage{rotating,listings}
\usepackage{longtable,paralist}
\usepackage{multicol}
\usepackage[hang,small,bf]{caption}

\usepackage[FIGBOTCAP,TABBOTCAP]{subfigure}

\renewcommand{\thesubtable}{\thetable.\arabic{subtable}}
\renewcommand{\thesubfigure}{\thefigure.\arabic{subfigure}}
\makeatletter
\renewcommand{\p@subtable}{}
\renewcommand{\p@subfigure}{}
\renewcommand{\@thesubtable}{\textbf{Table \thesubtable:}\hskip\subfiglabelskip}
\renewcommand{\@thesubfigure}{\textbf{Fig. \thesubfigure:}\hskip\subfiglabelskip}
\makeatother

\usepackage{needspace}
\usepackage{ifthen}

\usepackage{booktabs} 

\usepackage{graphicx} 
\setkeys{Gin}{width=\linewidth,totalheight=\textheight,keepaspectratio}
\graphicspath{{graphics/}}

\usepackage{xspace} 
\usepackage{units} 

\newcommand{\tp}{digram\xspace}
\newcommand{\tps}{digrams\xspace}
\newcommand{\ltp}{Digram\xspace}
\newcommand{\ltps}{Digrams\xspace}

\newcommand{\trp}{\mbox{TreeRePair}\xspace}

\newcommand{\hairsp}{\hspace{1pt}}

\newcommand{\ie}{\mbox{\textit{i.\hairsp{}e.}}\xspace}

\newcommand{\cf}{\textit{cf.}~}

\usepackage{amsmath,amssymb}

\begin{document}

\title{Tree structure compression with RePair}
\author{Markus Lohrey\inst{1,}\thanks{The first author 
is supported by the DFG  research project ALKODA.} 
\and Sebastian Maneth\inst{2} \and  Roy Mennicke\inst{1}
\institute{
Universit\"at
Leipzig, Institut f\"ur Informatik, Germany
\and
NICTA and University of New South Wales, Australia
\\
\email{lohrey@informatik.uni-leipzig.de,
  sebastian.maneth@nicta.com.au,
  roy@mennicke.info}}}

\maketitle

\begin{abstract}
In this work we introduce a new linear time compression algorithm,
called "Re-pair for Trees", which compresses ranked ordered trees
using linear straight-line context-free tree grammars. Such grammars
generalize straight-line context-free string grammars and allow basic
tree operations, like traversal along edges, to be executed without
prior decompression. Our algorithm can be considered as a
generalization of the "Re-pair" algorithm developed by N. Jesper
Larsson and Alistair Moffat in 2000. The latter 
algorithm is a dictionary-based compression algorithm for strings. 
We also introduce a succinct coding which is specialized in further
compressing the grammars generated by our algorithm. This is
accomplished without loosing the ability do directly execute queries
on this compressed representation of the input tree. Finally, we
compare the grammars and output files generated by a prototype of the
Re-pair for Trees algorithm with those of similar compression
algorithms. The obtained results show that that 
our algorithm outperforms its competitors 
in terms of compression ratio, runtime and memory usage.
\end{abstract}


\section{Introduction}

\subsection{Motivation}

Trees are nowadays a common data structure used in computer science to
represent data hierarchically. This is, for instance, evidenced by XML
documents which are widely used after their introduction in 1996. They
are sequential representations of ordered unranked trees.	When
processing trees it is often convenient to hold the tree structure in
memory in order to retain fast and random access to its
nodes. However, this often leads to a heavy resource consumption in
terms of memory usage due to the necessary pointer structure which
represents the tree structure. The space needed to load an entire XML
document into main memory in order to access it through a DOM proxy is
usually 3--8 times larger than the size of the document itself
\cite{Wang07space}. Therefore, it is essential for very large tree
structures to use a memory 
efficient representation. 

In \cite{Frick03query,Buneman03path} directed acyclic graphs (DAGs) were proposed to overcome 
this problem. By sharing common subtrees one is able to reduce the
size of the in-memory representation by a factor of about 10
\cite{Buneman03path}. One of the most appealing properties of this
representation is that queries like the ones of the XPath language can
be directly executed on the compressed representation, \ie, it is not
necessary to completely unfold the DAG.

Later, in \cite{Busatto08efficient} so called linear straight-line
context-free tree grammars were proposed as a more succinct
representation of an input tree. These grammars represent exactly one
tree and generalize the concept of sharing common subtrees to the
sharing of repeating tree patterns. Most important, this new
representation is still queryable, \ie, queries can be evaluated
without prior decompression. At the same time, the complexity of
querying, e.g., using XQuery, stays the same as for 
DAGs \cite{Lohrey2006complexity}.

However, finding the smallest linear straight-line context-free tree
grammar generating a given tree is NP-hard. Already finding the
smallest context-free string grammar for a given string is NP-complete
\cite{Charikar05smallest}. In \cite{Busatto08efficient} an algorithm
called BPLEX was introduced which generates a small linear
straight-line context-free tree grammar for a given input tree. On
average, the resulting grammar is 3.5--4 times smaller than the
minimal DAG (in terms of the number of edges). 
An implementation of this algorithm, which 
processes the underlying tree structure of XML documents, was also provided.

\subsection{Main Contribution}

Our main contribution is a compression algorithm, called "Re-pair for
Trees", which is based on linear straight-line context-free tree
grammars. Our investigations show that, regarding our test data, the
grammars generated by Re-pair for Trees are always smaller than the
grammars produced by the BPLEX algorithm. In addition, our algorithm

outperforms BPLEX 
in terms of runtime and memory usage. Note that especially 
runtime was a huge drawback of the BPLEX implementation.

The Re-pair for Trees algorithm is a generalization 
of the "Re-pair" algorithm which was developed by \textsc{Larsson} and
\textsc{Moffat} in \cite{larsson2000off}. The latter algorithm is an
offline dictionary-based compression method for strings consisting of
a simple but powerful phrase derivation method and a compact
dictionary encoding. A \emph{dictionary-based} compression algorithm
is an algorithm where the input message is parsed into a sequence of
phrases selected from a dictionary. Since the reference to a phrase in
the dictionary is more compact than the phrase itself often a
considerable compression can be achieved. Re-pair's dictionary is
inferred \emph{offline} since it is generated by considering the whole 
input message and since it is written out as a 
part of the compressed data so that it is available to the decoder.
\begin{figure}[t]
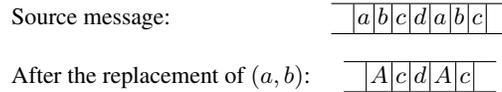

	\centering
	\renewcommand{\arraystretch}{2.0}
	\begin{tabular}{lcc}
		Source message:&&{\renewcommand{\arraystretch}{1.0}
			$\begin{array}{c|c|c|c|c|c|c|c|c}\hline
			\hspace{0.2cm}&a&b&c&d&a&b&c&\hspace{0.2cm}\\\hline
			\end{array}$}\\
		After the replacement of $(a,b)$:&&{\renewcommand{\arraystretch}{1.0}
			$\begin{array}{c|c|c|c|c|c|c}\hline
				\hspace{0.2cm}&A&c&d&A&c&\hspace{0.2cm}\\\hline
			\end{array}$}
	\end{tabular}
	\caption{The pair $(a,b)$ is replaced by the new symbol $A$.}
	\label{fig:repalg}
\end{figure}
The name Re-pair stands for "recursive pairing" and describes the idea
of the algorithm. The latter is to count the frequencies of all pairs
formed by two adjacent symbols of the source message, replacing the
most frequent pair by a new symbol (see Fig.~\ref{fig:repalg}),
updating the frequency counters of all involved pairs and repeating
this process until there are no pairs occurring twice in 
the source message. This compression technique allows 
searching the compressed data without prior decompression.

\subsection{Organization of this Work}

In Sect.~\ref{ch:repairAlgorithm} we explain in detail the two steps
of which the Re-pair for Trees algorithm consists. We also present a
complete example of a run of our algorithm and consider the
compressibility of special types of trees depending on the maximal
rank allowed for a nonterminal. Sect.~\ref{ch:implementationDetails}
explains some of the implementation details for the Re-pair for Trees
algorithm which is called \trp. In particular, we elaborate on its
linear runtime, the internal data structures used and its efficient
in-memory representation of the input tree. Moreover, in
Sect.~\ref{ch:succinctCoding} we present a succinct coding which is
specialized in further compressing the grammars generated by the
Re-pair for Trees algorithm without loosing the ability to directly
execute queries on this compressed representation of the input
tree. By using a combination of multiple Huffman codings, a run-length
coding and a fixed-length coding the resulting file sizes are always
smaller than the sizes of the files generated by competing compression
algorithms when executed on our test data. In
Sect.~\ref{ch:experimentalResults} we compare the compression results
of our implementation of the Re-pair for Trees algorithm with several
other compression algorithms. In particular, we consider BPLEX and
"Extended-Repair". The latter algorithm is also based on the Re-pair
for strings algorithm and was independently developed at the
University of Paderborn, Germany \cite{Krislin08repair,Boettcher10clustering}.

\section{Preliminaries}

In the following, $\mathbb{N}_{>0}=\mathbb{N}\setminus\{0\}$ denotes
the set of non-zero natural numbers. For a set $X$ we denote by 
$X^\ast$ the set of all finite words over $X$. 
For $w=x_1x_2\ldots x_n\in X^\ast$ we define $|w|=n$. 
The empty word is denoted by $\varepsilon$.

We sometimes surround an element of $\mathbb{N}$ by square brackets in
order to emphasize that we currently consider it a character instead
of a number. For instance, for the sequence of integers 
$222221$ we shortly write $[2]^5[1]$ instead of $2^51$ to clarify that we are not dealing with the fifth power of $2$.

\subsection{Labeled Ordered Tree}

A \emph{ranked alphabet} is a tuple
$(\mathcal{F},\mathsf{rank})$, where $\mathcal{F}$ is a finite set of
\emph{function symbols} and the function
$\mathsf{rank}:\mathcal{F}\rightarrow\mathbb{N}$ assigns to each
$\alpha\in\mathcal{F}$ its \emph{rank}. Furthermore, we define
$\mathcal{F}_i=\{a\in\mathcal{F}\mid\mathsf{rank}(\alpha)=i\}$.
We fix a ranked alphabet $(\mathcal{F},\mathsf{rank})$ in the following.
An \emph{$\mathcal{F}$-labeled ordered tree} 
is a pair $t=(\mathsf{dom}_t,\lambda_t)$, where
\begin{enumerate}[(1)]
	\item $\mathsf{dom}_t\subseteq\mathbb{N}_{>0}^\ast$ is a finite set of \emph{nodes},
	\item $\lambda_t:\mathsf{dom}_t\to\mathcal{F}$,
	\item if $w=vv'\in\mathsf{dom}_t$, then also $v\in\mathsf{dom}_t$, and
	\item if $v\in\mathsf{dom}_t$ and $\lambda_t(v)\in\mathcal{F}_n$, then $vi\in\mathsf{dom}_t$ if and only if $1\leq i\leq n$.
\end{enumerate}
The node $\varepsilon\in\mathsf{dom}_t$ is called the \emph{root} of
$t$. By $\mathsf{index}(w)$, where
$w=vi\in\mathsf{dom}_t\setminus\{\varepsilon\}$ and
$i\in\mathbb{N}_{>0}$, we denote the \emph{index} $i$ of the node $w$,
\ie, $w$ is the $i$-th child of its parent node. Furthermore, we
define $\mathsf{parent}(w)=v$. The \emph{size} of $t$ is given by the
number of edges of which it consists, \ie, we have
$|t|=|\mathsf{dom}_t|-1$. The \emph{depth} of the tree $t$ is
$\mathsf{depth}(t)=\max\{|u|\mid u\in\mathsf{dom}_t\}$. We identify an
$\mathcal{F}$-labeled tree $t$ with a term in the usual way: if
$\lambda_t(\varepsilon)=\alpha\in\mathcal{F}_i$, then this term is
$\alpha(t_1,\ldots,t_i)$, where $t_j$ is the term associated with the
subtree of $t$ rooted at node $j$, where $j\in\{1,\ldots,i\}$. The set
of all 
$\mathcal{F}$-labeled trees is $T(\mathcal{F})$.
\begin{example}\label{ex:tree}
	In Fig.~\ref{fig:tree} an $\mathcal{F}$-labeled ordered tree $t$ is shown. We have
	\[\mathsf{dom}_t=\{\varepsilon,1,2,3,11,12,21,22,31,111,112,121,122,211,212,221,222\}\enspace\mbox{.}\]
\end{example}
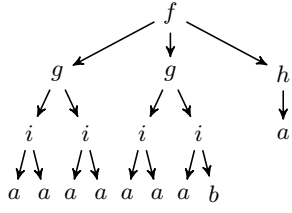
\begin{figure}[t]\centering
	\centering
	\begin{tikzpicture}[->,>=stealth',semithick]
		\tikzstyle{level 1}=[level distance=0.8cm, sibling distance=1.5cm]
		\tikzstyle{level 2}=[level distance=0.8cm, sibling distance=0.75cm]
		\tikzstyle{level 3}=[level distance=0.8cm, sibling distance=0.4cm]

		\node {$f$}
			child {node {$g$}
				child {node {$i$}
					child {node {$a$}}
					child {node {$a$}}
				}
				child {node {$i$}
					child {node {$a$}}
					child {node {$a$}}
				}
			}
			child {node {$g$}
				child {node {$i$}
					child {node {$a$}}
					child {node {$a$}}
				}
				child {node {$i$}
					child {node {$a$}}
					child {node {$b$}}
				}
			}
			child {node {$h$}
				child {node {$a$}}
			}
		;
	\end{tikzpicture}
	\caption{$\mathcal{F}$-labeled ordered tree $t$}\label{fig:tree}
\end{figure}
We fix a countable set $\mathcal{Y}=\{y_1,y_2,\ldots\}$ with
$\mathcal{Y}\cap\mathcal{F}=\emptyset$ of \emph{(formal context-)
  parameters} (below we also use a distinguished parameter
${z\notin\mathcal{Y}}$). The set of all $\mathcal{F}$-labeled trees
with parameters from $Y\subseteq\mathcal{Y}$ is denoted by
$T(\mathcal{F},Y)$. Formally, we consider 
parameters as function symbols of rank $0$ and define $T(\mathcal{F},Y)=T(\mathcal{F}\cup Y)$. 
The tree $t\in T(\mathcal{F},Y)$ is said to be \emph{linear} if every
parameter $y\in Y$ occurs at most once in $t$. By
$t[y_1/t_1,\ldots,y_n/t_n]$ we denote the tree that is obtained by
replacing in $t$ for every $i\in\{1,2,\ldots,n\}$ every 
$y_i$-labeled leaf with $t_i$, where 
$t\in T(\mathcal{F},\{y_1,\ldots,y_n\})$ and $t_1,\ldots,t_n\in T(\mathcal{F},Y)$.
A \emph{context} is a tree $C\in T(\mathcal{F},\mathcal{Y}\cup\{z\})$ 
in which the distinguished parameter $z$ appears exactly once. 
Instead of $C[z/t]$ we write briefly $C[t]$.
Let $t=(\mathsf{dom}_t,\lambda_t)\in
T(\mathcal{F},\{y_1,\ldots,y_n\})$ such that for every $y_i$ there
exists a node $v\in\mathsf{dom}_t$ with $\lambda_t(v)=y_i$. We say
that \emph{$t$ is a tree pattern occurring in $t'\in
  T(\mathcal{F},\mathcal{Y})$} if there exist a context $C\in
T(\mathcal{F},\mathcal{Y}\cup\{z\})$ and trees $t_1,\ldots,t_n\in
T(\mathcal{F},\mathcal{Y})$ 
such that
\[C\big[t[y_1/t_1,y_2/t_2,\ldots,y_n/t_n]\big]=t'\enspace\mbox{.}\]

\subsection{SLCF Tree Grammar}

For further consideration, let us fix a countable infinite set
$\mathcal{N}_i$ of symbols of rank $i\in\mathbb{N}$ with
$\mathcal{F}_i\cap\mathcal{N}_i=\emptyset$ and
$\mathcal{Y}\cap\mathcal{N}_0=\emptyset$. 
Hence, every finite subset 
$N\subseteq\bigcup_{i\geq0}\mathcal{N}_i$ is a ranked alphabet.
A \emph{context-free tree grammar (over the 
ranked alphabet $\mathcal{F}$)} or short \emph{CF tree grammar} is a triple ${\cal G}=(N,P,S)$, where
\begin{enumerate}[(1)]
	\item $N\subseteq\bigcup_{i\geq0}\mathcal{N}_i$ is a finite set of \emph{nonterminals},
	\item $P$ (the set of \emph{productions}) is a finite set of
          pairs $(A\to t)$, where $A\in N$, $t\in T(\mathcal{F}\cup
          N,\{y_1,\ldots,y_{\mathsf{rank}(A)}\})$,
          $t\notin\mathcal{Y}$, each of the parameters
          $y_1,\ldots,y_{\mathsf{rank}(A)}$ appears in $t$,
          and\footnote{In contrast to
            \cite{Lohrey09parameterreduction}, our definition of a
            context-free tree grammar inherits productivity, \ie,
            $t\notin\mathcal{Y}$ and each parameter
            $y_1,\ldots,y_{\mathsf{rank}(A)}$ appears in $t$ for every
            $(A\to t)\in P$. This is justified by the fact that the 
            grammars generated by the Re-pair for Trees algorithm are always productive.}
	\item $S\in N$ is 
the \emph{start nonterminal} of rank $0$.
\end{enumerate}
We assume that every nonterminal $B\in N\setminus\{S\}$ 
as well as every terminal symbol from $\mathcal{F}$ 
occurs in the right-hand side $t$ of some production $(A\to t)\in P$. 

Let us define the derivation relation $\Rightarrow_{\cal G}$ on
$T(\mathcal{F}\cup N,\mathcal{Y})$ as follows: $s\Rightarrow_{\cal G}
s'$ iff there exists a production $(A\rightarrow t)\in P$ with
$\mathsf{rank}(A)=n$, a context $C\in T(\mathcal{F}\cup
N,\mathcal{Y}\cup\{z\})$, and trees $t_1,\ldots,t_n\in
T(\mathcal{F}\cup N,\mathcal{Y})$ such that 
$s=C[A(t_1,\ldots,t_n)]$ and $s'=C[t[y_1/t_1\cdots y_n/t_n]]$. Let
\[L({\cal G})=\{t\in T(\mathcal{F})\mid S\Rightarrow_{\cal G}^\ast t\}\subseteq T(\mathcal{F})\enspace\mbox{.}\]
The \emph{size} $|{\cal G}|$ of the CF tree grammar ${\cal G}$ is defined by
\[|{\cal G}|=\sum_{(A\to t)\in P}|t|\enspace\mbox{.}\]
That means that $|\mathcal{G}|$ equals the sum of the numbers of edges of the right-hand sides of $P$'s productions.
We consider the following restrictions on context-free tree grammars:
\begin{itemize}
	\item $\mathcal{G}$ is \emph{$k$-bounded} (for $k\in\mathbb{N}$) if $\mathsf{rank}(A)\leq k$ for every $A\in N$.
	\item $\mathcal{G}$ is \emph{monadic} if it is $1$-bounded.
	\item $\mathcal{G}$ is \emph{linear} if for every $(A\to t)\in P$ the term $t$ is linear.
\end{itemize}
Let $\mathcal{G}=(N,P,S)$ be a CF tree grammar. 
We denote the set of all nodes in the right-hand 
sides of $\mathcal{G}$'s productions which are labeled by the nonterminal $A\in N$ by $\mathsf{ref}_{\mathcal{G}}(A)$, \ie,
\[\mathsf{ref}_{\mathcal{G}}(A)=\{(t,v)\mid\exists(B\to t)\in P:v\in\mathsf{dom}_t\wedge\lambda_t(v)=A\}\enspace\mbox{.}\]
Furthermore, let us define the following relation:
\[\leadsto_{\mathcal{G}}\,=\{(A,B)\in N\times N\mid (B\to t)\in P\wedge A\mbox{ occurs in }t\}\]
A \emph{straight-line context-free tree grammar (SLCF tree grammar)} 
is a CF tree grammar ${\cal G}=(N,P,S)$, where
\begin{enumerate}[(1)]
	\item\label{slcf_grammar_cond1} for every $A\in N$ there is \emph{exactly one} production $(A\rightarrow t)\in P$ with left-hand side $A$, and
	\item\label{slcf_grammar_cond2} the relation $\leadsto_{\mathcal{G}}$ is acyclic.
\end{enumerate}
The conditions (\ref{slcf_grammar_cond1}) and 
(\ref{slcf_grammar_cond2}) ensure that $L({\cal G})$ contains exactly one tree, which we denote by $\mathsf{val}({\cal G})$.
Let $\mathcal{G}$ be an SLCF tree grammar. 
We call the reflexive transitive closure of 
$\leadsto_{\mathcal{G}}$ the \emph{hierarchical order} 
of ${\cal G}$ and denote it by $\leadsto_{\mathcal{G}}^\ast$.

\begin{example}
	Consider the (linear and monadic) SLCF tree grammar $\mathcal{G}=(N,P,S)$ given by the following productions:
	\begin{align*}
		S&\to f\big(A(a),A(b),B\big)\\
		A(y_1)&\to g\big(i(a,a),i(a,y_1)\big)\\
		B&\to h(a)
	\end{align*}
	We have $\mathsf{val}(\mathcal{G})=t$, where $t\in T(\mathcal{F})$ is the tree from Example \ref{ex:tree} on page \pageref{ex:tree}.
\end{example}
SLCF tree grammars can be considered as a generalization of the
well-known DAGs (see, for instance, \cite{Lohrey2006complexity} for a
common definition). Whereas the latter is a structure preserving
compression of a tree by sharing common subtrees (see
Fig.~\ref{fig:subtree} for a depiction), SLCF tree grammars broaden
this concept to the sharing of repeated 
tree patterns in a tree (see Fig.~\ref{fig:subpattern}). 
Actually, a DAG can be considered as a $0$-bounded SLCF tree grammar.
\begin{figure}[t]
        \centering
	\subfigure[{A tree $t$ containing two occurrences of the very same subtree $t'$.}]{
		\begin{tikzpicture}[semithick,scale=.09]
		
			\draw (0,0) -- +(3,0) -- +(12,16) -- +(21,0) -- + (23,0) -- +(32,16) -- +(41,0) -- +(50,0) -- +(25,45) -- +(0,0);
			\draw +(22,28) node {\large $t$};
			
			\foreach \x in {3,23} {
			\begin{scope}[xshift=\x cm, yshift=-4cm]
				\filldraw (9,16) circle (10pt);
				\filldraw (9,20) circle (10pt);
				\draw[densely dotted] (9,16) -- +(0,4);
				\filldraw[fill=black!10, draw=black] (0,0) -- +(9,16) -- +(18,0) -- +(0,0);
				\draw (9,6) node {$t'$};
			\end{scope}
			}
		\end{tikzpicture}
		\label{fig:subtree}
	}
	\hspace{2cm}
	\subfigure[{A tree $t$ containing two occurrences of the tree pattern $p$.}]{
		\begin{tikzpicture}[semithick,scale=.09]
		
			\draw (0,0) -- +(50,0) -- +(25,45) -- +(0,0);
			\draw +(22,28) node {\large $t$};
	
			\foreach \xs/\ys/\lr in {10/5/2, 25/17/-2} {
				\begin{scope}[xshift=\xs cm, yshift=\ys cm]
					\filldraw[fill=black!10, draw=black] (0,0) -- +(12,0) -- +(6,10) -- +(0,0);
					\filldraw (0,0) circle (10pt);
					\draw[thin,densely dotted] (0,0) -- +(-1,-3) -- +(0,0) -- +(1,-3);
					\filldraw +(4,0) circle (10pt);
					\draw[thin,densely dotted] +(4,0) -- +(3,-3) -- +(4,0) -- +(5,-3);
					\filldraw +(8,0) circle (10pt);
					\draw[thin,densely dotted] +(8,0) -- +(7,-3) -- +(8,0) -- +(9,-3);
					\filldraw +(12,0) circle (10pt);
					\draw[thin,densely dotted] +(12,0) -- +(11,-3) -- +(12,0) -- +(13,-3);
					\draw +(6,3) node {$p$};
					\filldraw (6,10) circle (10pt);
					\draw[thin,densely dotted] (6,10) -- +(\lr,3);
				\end{scope}
			}
		\end{tikzpicture}
		\label{fig:subpattern}
	}
\end{figure}

Let $\mathcal{G}=(N,P,S)$ be a linear SLCF tree grammar. We define the function
\begin{equation}
\label{def:savValue}
\mathsf{sav}_{\mathcal{G}}(A)=|\mathsf{ref}_{\mathcal{G}}(A)|\cdot(|t|-\mathsf{rank}(A))-|t| ,
\end{equation}
which computes for every production 
$(A\to t)\in P$ its contribution to a small 
representation of the tree $\mathsf{val}(\mathcal{G})$ by the linear SLCF tree grammar $\mathcal{G}$.
The value $\mathsf{sav}_{\mathcal{G}}(A)$ specifies the number of
edges by which the production with left-hand side $A$ reduces the size
of the grammar $\mathcal{G}$. However, $\mathsf{sav}_{\mathcal{G}}$ is
not restricted to positive values. In particular, for a production
$(A\to t)\in P$ with $|\mathsf{ref}_{\mathcal{G}}(A)|=1$ we have 
$\mathsf{sav}_{\mathcal{G}}(A)=-\mathsf{rank}(A)$. 
Thus, a production which is only referenced once can 
be safely removed from the grammar without increasing the size of $\mathcal{G}$. 

Context-free tree grammars \cite{tata2007} and especially SLCF tree
grammars have been thoroughly studied recently. 
In theory, SLCF tree grammars in theory can be exponentially 
more succinct than DAGs \cite{Lohrey2006complexity}, which already can achieve exponential compression ratios.
Furthermore, in \cite{Lohrey2006complexity} various membership and
complexity problems were considered. It was shown that in many cases
the same complexity bounds hold as for DAGs. In particular, it was
pinpointed that for a given nondeterministic tree automaton ${\cal A}$
and a linear, $k$-bounded SLCF tree grammar ${\cal G}$ it can be
checked in polynomial time if $\mathsf{val}({\cal G})$ is accepted by
${\cal A}$ -- provided that $k$ is a constant. This is a worth
mentioning result since in the context of XML, for instance, 
tree automata are used to type check XML documents against an XML schema (\cf \cite{Murata05taxonomy,Neven02automata}).
Moreover, this result was further improved in
\cite{Lohrey09parameterreduction}, where it was shown 
that every linear SLCF tree grammar can be transformed in polynomial time into a monadic
(and linear) one. Together with the above mentioned result from
\cite{Lohrey2006complexity}, a polynomial time
algorithm for testing if a given nondeterministic tree automaton
accepts a tree given by a linear SLCF tree grammar 
(of arbitrary maximal rank for the nonterminals) can be obtained.

In \cite{Busatto08efficient} the so called BPLEX algorithm was
presented. It produces for a given $0$-bounded SLCF tree grammar
$\mathcal{G}_1$, \ie, $\mathcal{G}_1$ represents a DAG, in time
$O(|\mathcal{G}_1|)$ an equivalent linear SLCF tree grammar
$\mathcal{G}_2$, where
$\mathsf{val}(\mathcal{G}_2)=\mathsf{val}(\mathcal{G}_1)$ and
$\mathcal{G}_2$ is $k$-bounded 
($k$ is an input parameter). Experiments have shown that $|\mathcal{G}_2|$ is approximately 2--3 times smaller than $|\mathcal{G}_1|$.

Moreover, in \cite{Lohrey09parameterreduction} it was proved that the
evaluation problem for core XPath (the navigational part of XPath)
over SLCF tree grammars is PSPACE-complete just as this was proved
earlier for DAG-compressed trees by \textsc{Frick}, \textsc{Grohe} and
\textsc{Koch} in \cite{Frick03query}. The evaluation problem for XPath
asks whether a given node in a given tree is selected by a given XPath
expression. This result is remarkable since with SLCF tree grammars
one achieves better compression 
ratios than with DAGs.

\subsection{XML Terminology}

Regarding XML documents, we use the official terminology introduced in
\cite{Yergeau08}. Thus an XML document contains one or more
\emph{elements} which are either delimited by \emph{start-tags} and
\emph{end-tags} or by an \emph{empty-element tag}. The text between
the start-tag and the end-tag of an element is called the element's
\emph{content}. An element with no content is said to be 
\emph{empty}. There is exactly one element, called \emph{root}, which does not appear in the content of any other element.
\begin{figure}[t]\small
	\[
		\begin{array}{l}
			\begin{array}{l}
				\mbox{\ttfamily <books>}
			\end{array}\\
			\left.\begin{array}{l}
				\mbox{\ttfamily\quad <book>}\\
				\mbox{\ttfamily\qquad <author\,/><title\,/><isbn\,/>}\\
				\mbox{\ttfamily\quad </book>}\\
				\mbox{\ttfamily\quad ...}\\
				\mbox{\ttfamily\quad <book>}\\
				\mbox{\ttfamily\qquad <author\,/><title\,/><isbn\,/>}\\
				\mbox{\ttfamily\quad </book>}
			\end{array}\right\}\enspace\begin{rotate}{90}\mbox{\hskip-3ex 5 times}\end{rotate}\\
			\begin{array}{l}
				\mbox{\ttfamily </books>}
			\end{array}
		\end{array}
	\]
	\caption{An simplified XML document.\label{fig:XmlDocument}}
\end{figure}
\begin{example}\label{ex:XmlDocument}
The simplified XML document from Fig.~\ref{fig:XmlDocument} consists
of 21 elements of the five types \texttt{books}, \texttt{book},
\texttt{author}, \texttt{title} and \texttt{isbn}. The elements of
type \texttt{books} and \texttt{book} are delimited by start- and
end-tags and exhibit element content. The remaining elements are empty
elements delimited by 
empty-element tags. The root of the XML document is the element of type \texttt{books}.
\end{example}
The \emph{name} in the start- and end-tags of an 
element give the element's \emph{type}. Elements can 
specify \emph{attributes} by using name-value pairs. Consider for instance the element
\begin{verbatim}
      <phone prefix="012">3456</phone>
\end{verbatim}
exhibiting one attribute specification with attribute name \texttt{prefix} and attribute value \texttt{012}.

In addition to these terms we denote by \emph{XML document tree} 
the nested structure of elements which is left after 
removing all character data and attribute specifications from an XML document.

\subsection{Binary Tree Model}\label{sec:binaryTreeModel}

An XML document tree can be considered as an unranked tree, \ie, nodes
with the same label possibly have a varying number of
children. Figure~\ref{fig:XmlDocumentTree} shows the XML document tree
of the XML document from Example \ref{ex:XmlDocument}. In our case,
the XML document tree is a ranked tree, \ie, all nodes with the same
label exhibit the same number of children. However, the XML document
might as well have contained an element of type \verb|book| exhibiting
a second \verb|author| child element. In this case, 
we would have not obtained a ranked tree.
\begin{figure*}[t]
	\centering
	\begin{tikzpicture}[->,>=stealth',semithick]
		\tikzstyle{level 1}=[level distance=1.2cm, sibling distance=3cm]
		\tikzstyle{level 2}=[level distance=1.2cm, sibling distance=0.9cm]

		\node {$\mathsf{books}$}
			child {node {$\mathsf{book}$}
				child {node {$\mathsf{author}$}}
				child {node {$\mathsf{title}$}}
				child {node {$\mathsf{isbn}$}}
			}
			child {node {$\mathsf{book}$}
				child {node {$\mathsf{author}$}}
				child {node {$\mathsf{title}$}}
				child {node {$\mathsf{isbn}$}}
			}
			child {node {$\mathsf{book}$}
				child {node {$\mathsf{author}$}}
				child {node {$\mathsf{title}$}}
				child {node {$\mathsf{isbn}$}}
			}
			child {node {$\mathsf{book}$}
				child {node {$\mathsf{author}$}}
				child {node {$\mathsf{title}$}}
				child {node {$\mathsf{isbn}$}}
			}
			child {node {$\mathsf{book}$}
				child {node {$\mathsf{author}$}}
				child {node {$\mathsf{title}$}}
				child {node {$\mathsf{isbn}$}}
			}
		;
	\end{tikzpicture}
	\caption{XML document tree of the XML document listed in Fig.~\ref{fig:XmlDocument}}\label{fig:XmlDocumentTree}
\end{figure*}
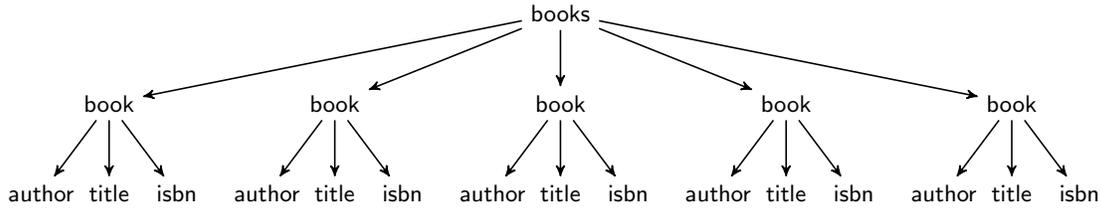

In the next section we will learn that our Re-pair for Trees algorithm
operates on ranked trees only. Therefore, in general, a transformation
of an XML document tree becomes necessary. A common way of modeling
such a tree in a ranked way is to transform it into a binary
$\mathcal{F}$-labeled ordered tree $t$ by encoding first-child and
next-sibling relations. 
In fact, 
\begin{itemize}
	\item the first child element of an XML element becomes the left child of the node representing its parent element and 
	\item the right sibling element of another element becomes the right-child of the node representing its left sibling (\cf Fig.\ \ref{fig:binaryRepresentation}).
\end{itemize}
\begin{figure}[t]
	\centering
	\begin{tikzpicture}
		[bend angle=25, node distance=0.25cm,text height=1.5ex,
		place/.style={},
		pre/.style={<-,>=stealth'},
		post/.style={->,>=stealth'}]
		\node[place] (books) {$\mathsf{books}^{10}$};
		
		\node[place] (book1) [below=of books] {$\mathsf{book}^{11}$}
			edge [pre] (books);
		\node[place] (author1) [below=of book1] {$\mathsf{author}^{01}$}
			edge [pre] (book1);
		\node[place] (title1) [below=of author1] {$\mathsf{title}^{01}$}
			edge [pre] (author1);
		\node[place] (isbn1) [below=of title1] {$\mathsf{isbn}^{00}$}
			edge [pre] (title1);
			
		\node[place] (book2) [right=of book1] {$\cdots$}
			edge [pre] (book1);
			
		\node[place] (book4) [right=of book2] {$\mathsf{book}^{11}$}
			edge [pre] (book2);
		\node[place] (author4) [below=of book4] {$\mathsf{author}^{01}$}
			edge [pre] (book4);
		\node[place] (title4) [below=of author4] {$\mathsf{title}^{01}$}
			edge [pre] (author4);
		\node[place] (isbn4) [below=of title4] {$\mathsf{isbn}^{00}$}
			edge [pre] (title4);
			
		\draw[decorate,decoration={brace,mirror}] ([xshift=-5pt]isbn1.south) -- ([xshift=5pt]isbn4.south) node[midway,sloped,below=2pt] {$4$ times};

		\node[place] (book5) [right=of book4] {$\mathsf{book}^{10}$}
			edge [pre] (book4);
		\node[place] (author5) [below=of book5] {$\mathsf{author}^{01}$}
			edge [pre] (book5);
		\node[place] (title5) [below=of author5] {$\mathsf{title}^{01}$}
			edge [pre] (author5);
		\node[place] (isbn5) [below=of title5] {$\mathsf{isbn}^{00}$}
			edge [pre] (title5);
	\end{tikzpicture}
	\caption{Binary tree representation of the XML 
            document tree from Fig.~\ref{fig:XmlDocumentTree}.\label{fig:binaryRepresentation}}
\end{figure}
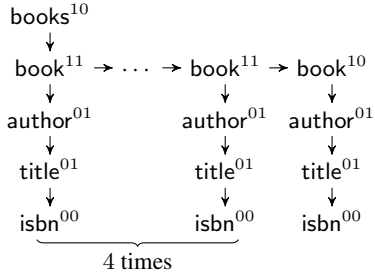
Note that a node representing a leaf (resp.~a last sibling) of the XML
document has no left (resp.~no right) child in the binary tree model
representation. Therefore $\mathcal{F}$ does not consist of the
element types of the XML document but of special versions of the
element types indicating that the left, the right, both or no children
are missing. In Fig.~\ref{fig:binaryRepresentation} this is denoted by
superscripts at the end of the element types. These superscripts are
listed in Table \ref{tab:superscriptsAndTheirMeanings} 
together with their meanings.
\begin{table}[h]
	\centering
	\begin{tabular*}{3.6cm}{c@{\extracolsep{\fill}}l}
		\toprule
		Superscript&Meaning\\
		\midrule
		00&no children\\
		10&no right child\\
		01&no left child\\
		11&two children\\
		\bottomrule
	\end{tabular*}
	\caption{The superscripts and their meanings.}\label{tab:superscriptsAndTheirMeanings}
\end{table}
Let us point out that another way of preserving the rankedness along
with circumventing the introduction of special labels with a lower
rank is the introduction of placeholder nodes. These can be used to
indicate missing left or right children. However, our experiments
showed that our implementation of Re-pair for Trees achieves slightly
less competitive 
compression results in this setting.

In \cite{Busatto08efficient} it was stated that the binary tree model
allows access to the next-in-preorder and previous-in-preorder node in
$O(\mathit{depth})$, where $\mathit{depth}$ refers to the longest path
from the root of the XML document to one of its leaves. Furthermore,
in \cite{Milo03typechecking} it was demonstrated that XML query
languages can be readily evaluated on the binary tree model.

\section{Re-Pair for Trees}\label{ch:repairAlgorithm}

In this section we study the Re-pair for Trees algorithm in detail.
It consists of two steps, namely, a \emph{replacement step} and a \emph{pruning step}. Furthermore, a detailed example of a run of our algorithm is presented. Finally, we investigate the impact of a possible restriction on the maximal rank allowed for nonterminals.

\subsection{Digrams}

In order to be able to elaborate on our Re-pair for Trees algorithm we need the following definitions. Recall that we have fixed a ranked alphabet $\mathcal{F}$ of function symbols, a set $\mathcal{N}$ of nonterminals and a set $\mathcal{Y}$ of parameters. We define the set of triples
\[\Pi=\bigcup_{a\in\mathcal{F}\cup\mathcal{N}}\{a\}\times\{1,2,\ldots,\mathsf{rank}(a)\}\times(\mathcal{F}\cup\mathcal{N})\enspace\mbox{.}\]
A \emph{\tp} is a triple $\alpha=(a,i,b)\in\Pi$. The symbol $a$ is called the \emph{parent symbol of the \tp $\alpha$} and $b$ is called the \emph{child symbol of the \tp $\alpha$}, respectively. We define
\begin{align*}
	\mathsf{par}(\alpha)&=\mathsf{rank}(a)+\mathsf{rank}(b)-1\text{ and}\\
	\mathsf{pat}(\alpha)&=a\big(y_1,\ldots,y_{i-1},b(y_i,\ldots,y_{j-1}),y_j,\ldots,y_{\mathsf{par}(\alpha)}\big)\enspace\mbox{,}
\end{align*}
where $j=i+\mathsf{rank}(b)$ and $y_1,y_2,\ldots,y_{\mathsf{par}(\alpha)}\in\mathcal{Y}$. Let $m\in\mathbb{N}\cup\{\infty\}$. We further define the set
\[\Pi_m=\{\alpha\in\Pi\mid\mathsf{par}(\alpha)\leq m\}\enspace\mbox{.}\]
Obviously, it holds that $\Pi_\infty=\Pi$.
We can consider $\mathsf{pat}(\alpha)$ as the tree pattern which is represented by the \tp $\alpha$. We usually denote \tps by possibly indexed lowercase letters $\alpha,\alpha_1,\alpha_2,\ldots,\beta,\ldots$ of the Greek alphabet.
An \emph{occurrence} of the \tp $\alpha\in\Pi$ within the tree 
$t=(\mathsf{dom}_t,\lambda_t)\in T(\mathcal{F}\cup\mathcal{N},\mathcal{Y})$ is a node $v\in\mathsf{dom}_t$ at which a subtree
	\[\mathsf{pat}(\alpha)[y_1/t_1,y_2/t_2,\ldots,y_{\mathsf{par}(\alpha)}/t_{\mathsf{par}(\alpha)}]\enspace\mbox{,}\]
with $t_1,t_2,\ldots,t_{\mathsf{par}(\alpha)}\in T(\mathcal{F}\cup\mathcal{N},\mathcal{Y})$, is rooted. 
The \emph{set of all occurrences of the \tp $\alpha$ in $t$} is denoted by 
$\mathsf{OCC}_t(\alpha)\subseteq\mathsf{dom}_t$.

Let $\alpha=(a,i,b)\in\Pi$ and $t\in T(\mathcal{F}\cup\mathcal{N},\mathcal{Y})$. 
Two occurrences $v, w \in \mathsf{OCC}_t(\alpha)$ are \emph{overlapping} 
if one of the following equations holds: $v=w$, $vi=w$ or $wj=v$.  
Otherwise, \ie, if $v$ and $w$ 
are not overlapping, $v$ and $w$ are said to be \emph{non-overlapping}.
A subset $\sigma\subseteq\mathsf{OCC}_t(\alpha)$ is said to be overlapping 
if there exist overlapping $v,w\in\sigma$, otherwise it is called
\emph{non-overlapping}.
It is easy to see that the set $\mathsf{OCC}_t(\alpha)$ is non-overlapping if $a\neq b$. 
In contrast, if we have $a=b$, the set $\mathsf{OCC}_t(\alpha)$ potentially 
contains overlapping occurrences. Consider the following example:
\begin{example}
Let $t\in T(\mathcal{F})$ be the tree depicted in Fig.~\ref{fig:treeOfOverlappingOccs} and let $\alpha=(f,2,f)$. Hence, $\{\varepsilon,2,22\}\subseteq\mathsf{OCC}_t(\alpha)$, where on the one hand $\varepsilon$ and $2$ and on the other hand $2$ and $22$ are overlapping occurrences of $\alpha$.
\end{example}
\begin{figure}[t]
	\centering
		\begin{tikzpicture}[->,>=stealth',semithick,level distance=0.8cm, sibling distance=1.5cm]
			\tikzstyle{fan}=[anchor=north,isosceles triangle, shape border uses incircle,inner sep=0.5pt,shape border rotate=90,draw]
			
			\node {$f$} [child anchor=north] 
				child {node [fan] {$t_1$}}
				child {node [below=2pt] {$f$}
					child {node [fan] {$t_2$}}
					child {node [below=5pt] {$f$}
						child {node [fan] {$t_3$}}
						child {node [below=2pt] {$f$}
							child {node [fan] {$t_4$}}
							child {node [fan] {$t_5$}}
						}
					}
				}
			;
		\end{tikzpicture}
	\caption{Tree $t\in T(\mathcal{F})$ consisting of nodes labeled by the terminal $f\in\mathcal{F}_2$ and the subtrees $t_1,t_2,\ldots,t_5\in T(\mathcal{F})$. We have to deal with overlapping occurrences of the \tp $(f,2,f)$.}\label{fig:treeOfOverlappingOccs}
\end{figure}
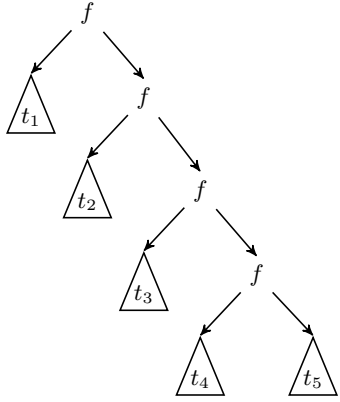
Let $\alpha\in\Pi$ and $t\in T(\mathcal{F}\cup\mathcal{N},\mathcal{Y})$. Let $\sigma\subseteq\mathsf{OCC}_t(\alpha)$ be a non-overlapping set. Furthermore, let us assume that $\sigma\cup\{v\}$ is overlapping for all $v\in\mathsf{OCC}_t(\alpha)\setminus\sigma$, \ie, $\sigma$ is maximal with respect to inclusion among non-overlapping subsets. Then $\sigma$ is not necessarily maximal with respect to cardinality.
\begin{example}\label{ex:maximality}
	Consider the tree $t\in T(\mathcal{F})$ which is depicted in Fig.~\ref{fig:treeOfOverlappingOccs}. Let $\alpha=(f,2,f)\in\Pi$. We have $\mathsf{OCC}_{t}(\alpha)=\{\varepsilon,2,22\}$. Let $\sigma=\{2\}\subseteq\mathsf{OCC}_{t}(\alpha)$. The set $\sigma$ is non-overlapping and $\sigma\cup\{v\}$ is overlapping for all $v\in\mathsf{OCC}_{t}(\alpha)\setminus\sigma$. However, $\sigma$ is not maximal with respect to cardinality. Consider the non-overlapping subset $\sigma'=\{\varepsilon,22\}\subseteq\mathsf{OCC}_{t}(\alpha)$. We have $|\sigma|<|\sigma'|$.
\end{example}
Example \ref{ex:maximality} shows us that we cannot choose an arbitrary subset $\sigma\subseteq\mathsf{OCC}_t(\alpha)$ which is non-overlapping and maximal with respect to inclusion to obtain a set which is maximal with respect to cardinality. 
Let us also point out that the set $\mathsf{OCC}_t(\alpha)$ may contain more than one maximal (with respect to cardinality) non-overlapping subset. 
\begin{example}
Consider the tree $f(f(f(a)))$ over the ranked alphabet $\mathcal{F}$. The sets $\{\varepsilon\}$ and $\{1\}$ are both maximal with respect to cardinality.
\end{example}
\begin{figure}[tb]
\lstset{emph={FUNCTION,ENDFUNC,while,endwhile,true,for,if,else,then,do,endif,endfor,return}, emphstyle=\bfseries,morecomment=[l]{//},commentstyle=\color{gray}}
\begin{lstlisting}
FUNCTION next-in-postorder((*@$t,v$@*)) // let (*@\color{gray}$t=(\mathsf{dom}_t,\lambda_t)$@*)
	if ((*@$v=\varepsilon$@*)) then
		(*@$v:=\;$@*)walk-down((*@$t$@*), (*@$v$@*));
	else	
		(*@$i:=\mathsf{index}(v)+1$@*);
		(*@$v:=\mathsf{parent}(v)$@*);
		
		if ((*@$\mathsf{rank}(\lambda_t(v))\geq i$@*)) then
			(*@$v:=vi$@*);
			(*@$v:=\;$@*)walk-down((*@$t$@*), (*@$v$@*));	
		endif
	endif
	return (*@$v$@*);
ENDFUNC

FUNCTION walk-down((*@$t$@*), (*@$v$@*)) // let (*@$\color{gray}t=(\mathsf{dom}_t,\lambda_t)$@*)
	while (true) do
		if ((*@$\mathsf{rank}(\lambda_t(v))>0$@*)) then
			(*@$v:=v1$@*);
		else
			return (*@$v$@*);
		endif
	endwhile
ENDFUNC
\end{lstlisting}
	\caption{The algorithm which is used to traverse a tree in postorder.}\label{lst:traversalAlgorithm}
\end{figure}
\begin{figure}[tb]
	\lstset{emph={FUNCTION,ENDFUNC,return,for,if,then,do,endif,endfor,while,endwhile}, emphstyle=\bfseries,morecomment=[l]{//},commentstyle=\color{gray}}
	\begin{lstlisting}
FUNCTION retrieve-occurrences((*@$t,\alpha$@*)) // let (*@\color{gray}$\alpha=(a,i,b)$@*)
	(*@$\mathsf{occ}_t(\alpha):=\emptyset$@*); (*@$v:=\varepsilon$@*);
	while (true) do
		(*@$v:=\;$@*)next-in-postorder((*@$t$@*), (*@$v$@*));
		if ((*@$v\in\mathsf{OCC}_t(\alpha)\;\wedge\;vi\notin\mathsf{occ}_t(\alpha)$@*)) then
			(*@$\mathsf{occ}_t(\alpha):=\mathsf{occ}_t(\alpha)\cup\{v\}$@*)
		endif
		if ((*@$v=\varepsilon$@*)) then
			return (*@$\mathsf{occ}_t(\alpha)$@*);
		endif
	endwhile
ENDFUNC
	\end{lstlisting}
	\caption{The function \texttt{retrieve-occurrences} which is used to construct the set $\mathsf{occ}_t(\alpha)$ for a \tp $\alpha\in\Pi$ and a tree $t\in T(\mathcal{F}\cup\mathcal{N})$.}\label{lst:functionRetrieveOccurrences}
\end{figure}
The algorithm \texttt{retrieve-occurrences($t,\alpha$)}
from Fig.~\ref{lst:functionRetrieveOccurrences} computes one non-over\-lapping subset of $\mathsf{OCC}_t(\alpha)$ which we denote by $\mathsf{occ}_t(\alpha)$. Lemma \ref{lemma:occtIsMaximal} ascertains that this subset is maximal with respect to cardinality.
Using the function \verb|next-in-postorder| listed in Fig.~\ref{lst:traversalAlgorithm} we traverse the tree $t$ in postorder. We begin by passing the parameters $t$ and $\varepsilon$ and obtain the first node $u\in\mathsf{dom}_t$ of $t$ in postorder. The second node in post order is obtained by passing the parameters $t$ and $u$. This step can be repeated to traverse the whole tree $t$ in postorder.
For every node $v$ which is encountered during the postorder traversal it is checked if $v$ is an occurrence of $\alpha$ and if it is non-overlapping with all occurrences already contained in the current set $\mathsf{occ}_t(\alpha)$. If both conditions are fulfilled, the node $v$ is added to $\mathsf{occ}_t(\alpha)$. 

Now, let us assume that we have constructed the set $\mathsf{occ}_t(\alpha) \subseteq\mathsf{OCC}_t(\alpha)$ using the function \texttt{retrieve-occurrences}. If $a\neq b$ in $\alpha=(a,i,b)$ we have 
$\mathsf{occ}_t(\alpha)=\mathsf{OCC}_t(\alpha)$. In the following, we show that the subset $\mathsf{occ}_t(\alpha)\subseteq\mathsf{OCC}_t(\alpha)$ is maximal with respect to cardinality.
\begin{lemma}\label{lemma:occtIsMaximal}
	Let $\alpha\in\Pi$ and $t\in T(\mathcal{F}\cup\mathcal{N},\mathcal{Y})$. Let $\sigma\subseteq\mathsf{OCC}_t(\alpha)$ be non-overlapping and maximal with respect to cardinality. Then the equation $|\mathsf{occ}_t(\alpha)|=|\sigma|$ holds.
\end{lemma}
\begin{proof}
	In the following we briefly write "maximal" for "maximal with
        respect to cardinality". Let $\alpha=(a,i,b)\in\Pi$, $t\in
        T(\mathcal{F}\cup\mathcal{N},\mathcal{Y})$ and $\sigma$ as
        above. The graph $(V,E)$ with
$$
		V=\mathsf{OCC}_t(\alpha)\cup\{vi\mid v\in\mathsf{OCC}_t(\alpha)\} \text{ and }
		E=\{(v,vi)\mid v\in\mathsf{OCC}_t(\alpha)\}
$$
is a disjoint union of paths. Maximal non-overlapping subsets of $\mathsf{OCC}_t(\alpha)$ exactly correspond to maximum matchings in $(V,E)$. Clearly, a path with an odd number of edges has a unique maximum matching, whereas a path with an even number of edges has two maximum matchings: one containing the first edge (in direction from the root) and one containing the last edge on the path. Intuitively, the algorithm from Fig.~\ref{lst:functionRetrieveOccurrences} finds the maximum matching in $(V,E)$ which contains for every path with an even number of edges the last edge in direction from the root.
\qed
\end{proof}
Let $t\in T(\mathcal{F}\cup\mathcal{N},\mathcal{Y})$, $\alpha=(a,i,b)\in\Pi$ and $A\in\mathcal{N}_{\mathsf{par}(\alpha)}$. By $t[\alpha/A]$ we denote the tree which is obtained by replacing all occurrences from $\mathsf{occ}_t(\alpha)$ in the tree $t$ by the nonterminal $A$ (in parallel). More precisely, we replace every subtree
\[\mathsf{pat}(\alpha)[y_1/t_1,y_2/t_2,\ldots,y_{\mathsf{par}(\alpha)}/t_{\mathsf{par}(\alpha)}]\mbox{,}\]
where $t_1,t_2,\ldots,t_{\mathsf{par}(\alpha)}\in T(\mathcal{F}\cup\mathcal{N},\mathcal{Y})$, which is rooted at an occurrence $v\in\mathsf{occ}_t(\alpha)$ by a new subtree $A(t_1,t_2,\ldots,t_{\mathsf{par}(\alpha)})$.
\begin{example}
Consider the tree $t\in T(\mathcal{F})$ which is depicted in Fig.~\subref{fig:replacementProcessBefore}. We have $\mathsf{occ}_t(\alpha)=\{\varepsilon,11,12,21,22\}$,	where $\alpha=(f,2,f)$. By replacing the \tp $\alpha$ in $t$ by a nonterminal $A\in\mathcal{N}_3$ we obtain the tree $t[\alpha/A]$ which is depicted in Fig.~\subref{fig:replacementProcessAfter}.
\end{example}
\begin{figure}[tb]
	\begin{center}
	\subfigure[The tree $t\in T(\mathcal{F})$.]{
		\parbox{6.3cm}{
		\beginpgfgraphicnamed{figReplacementProcessBefore}
		\begin{tikzpicture}[->,>=stealth']
			\tikzstyle{level 1}=[level distance=0.8cm, sibling distance=3.2cm]
			\tikzstyle{level 2}=[level distance=0.8cm, sibling distance=1.6cm]
			\tikzstyle{level 3}=[level distance=0.8cm, sibling distance=0.8cm]
			\tikzstyle{level 4}=[level distance=0.8cm, sibling distance=0.4cm]
			\node {$f$}
				child {node {$f$}
					child {node {$f$}
						child {node {$f$}
							child {node {$a$}}
							child {node {$a$}}
						}
						child {node {$f$}
							child {node {$a$}}
							child {node {$a$}}
						}
					}
					child {node {$f$}
						child {node {$f$}
							child {node {$a$}}
							child {node {$a$}}
						}
						child {node {$f$}
							child {node {$a$}}
							child {node {$a$}}
						}
					}
				}
				child {node {$f$}
					child {node {$f$}
						child {node {$f$}
							child {node {$a$}}
							child {node {$a$}}
						}
						child {node {$f$}
							child {node {$a$}}
							child {node {$a$}}
						}
					}
					child {node {$f$}
						child {node {$f$}
							child {node {$a$}}
							child {node {$a$}}
						}
						child {node {$f$}
							child {node {$a$}}
							child {node {$a$}}
						}
					}
				}
			;
		\end{tikzpicture}
		\endpgfgraphicnamed
		}
		\label{fig:replacementProcessBefore}
	}
	\hspace{1cm}
	\subfigure[{The tree $t[\alpha/A]$.}]{
		\parbox{6.3cm}{
		\beginpgfgraphicnamed{figReplacementProcessAfter}
		\begin{tikzpicture}[->,>=stealth']
			\tikzstyle{level 1}=[level distance=0.8cm, sibling distance=3.2cm]
			\tikzstyle{level 2}=[level distance=0.8cm, sibling distance=1.6cm]
			\tikzstyle{level 3}=[level distance=0.8cm, sibling distance=0.8cm]
			\tikzstyle{level 4}=[level distance=0.8cm, sibling distance=0.4cm]
			\node {$A$}
				child[sibling distance=2.7cm] {node {$f$}
					child {node {$A$}
						child[sibling distance=0.4cm] {node {$f$}
							child {node {$a$}}
							child {node {$a$}}
						}
						child {node {$a$}}
						child[sibling distance=0.4cm] {node {$a$}}
					}
					child {node {$A$}
						child[sibling distance=0.4cm] {node {$f$}
							child {node {$a$}}
							child {node {$a$}}
						}
						child {node {$a$}}
						child[sibling distance=0.4cm] {node {$a$}}
					}
				}
				child {node {$A$}
					child[sibling distance=0.4cm] {node {$f$}
						child[sibling distance=0.4cm] {node {$a$}}
						child[sibling distance=0.4cm] {node {$a$}}
					}
					child {node {$a$}}
					child[sibling distance=0.4cm] {node {$a$}}
				}
				child[sibling distance=1.6cm] {node {$A$}
					child[sibling distance=0.4cm] {node {$f$}
						child[sibling distance=0.4cm] {node {$a$}}
						child[sibling distance=0.4cm] {node {$a$}}
					}
					child {node {$a$}}
					child[sibling distance=0.4cm] {node {$a$}}
				}
			;
		\end{tikzpicture}
		\endpgfgraphicnamed
		}
		\label{fig:replacementProcessAfter}
	}
	\end{center}
	\stepcounter{figure}
\end{figure}
For $t\in T(\mathcal{F}\cup\mathcal{N})$ and $m\in\mathbb{N}\cup\{\infty\}$, we define
\[\mathsf{max}_m(t)=\begin{cases}
	\alpha\in\Pi_m&\text{if }\mathsf{occ}_t(\alpha)\neq\emptyset\text{ and }
		\forall\beta\in\Pi_m:|\mathsf{occ}_t(\beta)|\leq|\mathsf{occ}_t(\alpha)|\\
	\mathsf{undefined}&\text{if }\forall\alpha\in\Pi_m:\mathsf{occ}_t(\alpha)=\emptyset
\end{cases}\]
The function $\mathsf{max}_m:T(\mathcal{F}\cup\mathcal{N})\to\Pi$ associates with every tree $t\in T(\mathcal{F}\cup\mathcal{N})$ a \tp $\alpha\in\Pi_m$ which occurs in $t$ most frequently (with respect to all \tps from $\Pi_m$). If there are multiple most frequent \tps, we can choose any of them. In contrast, we have $\mathsf{max}_m(t)=\mathsf{undefined}$ if there is no most frequent \tp. If $m=\infty$ there is no most frequent pair if and only if the tree $t$ consists of exactly one node. Now let us assume that $m\neq\infty$. We have $\mathsf{max}_m(t)=\mathsf{undefined}$ if and only if $t$ consists of exactly one node or if for all \tps $\alpha$ occurring in $t$ it holds that $\alpha\notin\Pi_m$. 

In the sequel, if we do not specify the maximal rank allowed for a nonterminal, we always assume that $m=\infty$. For convenience we write $\mathsf{max}(t)$ instead of $\mathsf{max}_\infty(t)$, \ie, we omit the symbol $\infty$.

\subsection{Replacement of \ltps}\label{sec:replacementStep}

In this section we introduce the first step of our Re-pair for Trees algorithm, namely, the \emph{replacement step}. Let $m\in\mathbb{N}\cup\{\infty\}$ be the maximal rank allowed for a nonterminal\footnote{Regarding our implementation of the Re-pair for Trees algorithm which is described in Sect.~\ref{ch:implementationDetails}, $m$ is a parameter which can be specified by the user.} and let the tree $t=(\mathsf{dom}_t,\lambda_t)\in T(\mathcal{F})$ be the input of our algorithm.

We describe a run of the Re-pair for Trees algorithm by a sequence of $h+1$ linear SLCF tree grammars
$\mathcal{G}_0,\mathcal{G}_1,\ldots,\mathcal{G}_{h}$, where $h\in\mathbb{N}$. For every $i\in\{0,1,\ldots,h\}$ we have $\mathcal{G}_i=(N_i,P_i,S_i)$, $(S_i\to t_i)\in P_i$, $\alpha_i=\mathsf{max}_m(t_i)$ and $\mathsf{val}(\mathcal{G}_i)=t$. 
The grammar $\mathcal{G}_0$ contains solely the start production $(S_0\to t_0)$, where $t_0=t$. We obtain the grammar $\mathcal{G}_{i+1}$ by replacing the \tp $\alpha_i$ in the right-hand side of $\mathcal{G}_i$'s start production $t_i$ by a new nonterminal $A_{i+1}\in\mathcal{N}_{\mathsf{par}(\alpha_i)}\setminus N_i$ ($0\leq i\leq h-1$). We set
\begin{align*}
	N_{i+1}&=(N_i\setminus\{S_i\})\cup\{S_{i+1},A_{i+1}\}\mbox{ and}\\
	P_{i+1}&=(P_i\setminus\{(S_i\to t_i)\})\cup\{\big(A_{i+1}\to\mathsf{pat}(\alpha_i)\big),(S_{i+1}\to t_{i+1})\}\enspace\mbox{,}
\end{align*}
where $t_{i+1}=t_i[\alpha_i/A_{i+1}]$.

The computation stops if there is no \tp $\alpha\in\Pi_m$ occurring at least twice in the start production of the current grammar, \ie, either the equation $|\mathsf{occ}_{t_h}(\mathsf{max}_m(t_h))|=1$ or the equation $\mathsf{max}_m(t_h)=\mathsf{undefined}$ holds. In contrast, for all $0\leq i\leq h-1$ we have $|\mathsf{occ}_{t_i}(\mathsf{max}_m(t_i))|>1$.

Note that the linear SLCF tree grammar $\mathcal{G}_h$ is almost in Chomsky normal form (CNF) as it is defined in \cite{Lohrey09parameterreduction}. By appropriately transforming the right-hand side of $S_h$ (as it is described in the proof of Proposition 5 of \cite{Lohrey09parameterreduction}) and introducing a production with right-hand side $a(y_1,\ldots,y_n)$ for every terminal $a\in\mathcal{F}_n$ ($n\in\mathbb{N}$) we would obtain a linear SLCF tree grammar which perfectly meets the requirements of the CNF.

The linear SLCF tree grammar $\mathcal{G}_h$ can only be considered an intermediate result, since it potentially consists of productions which do not contribute to a compact representation of the input tree $t$. Therefore, we get rid of unprofitable productions by eliminating them during the so-called \emph{pruning step}. The latter, which is described in the next section, is executed directly after the replacement step.

\subsection{Pruning the Grammar}\label{sec:pruningStep}

Let $\mathcal{G}=(N,P,S)$ be a linear SLCF tree grammar. We \emph{eliminate} a production $(A\to t)$ from $P$ as follows: 
	\begin{enumerate}[(1)]
		\item For every reference $(t',v)\in\mathsf{ref}_{\mathcal{G}}(A)$ we replace the subtree $A(t_1,t_2,\ldots,t_n)$ rooted at $v\in\mathsf{dom}_{t'}$ by the tree
			\[t[y_1/t_1,y_2/t_2,\ldots,y_n/t_n]\enspace\mbox{,}\]
			where $t_1,\ldots,t_n\in T(\mathcal{F}\cup\mathcal{N},\mathcal{Y})$ and $n=\mathsf{rank}(A)$. 
		\item We update the set of productions by setting 
			\[P:=P\setminus\{(A\to t)\}\enspace\mbox{.}\]
	\end{enumerate}
Let $\mathcal{G}=(N,P,S)$ be the linear SLCF tree grammar generated in the replacement step of our algorithm, \ie, we have $\mathcal{G}=\mathcal{G}_h$. Let $n=|N|$ and let
\[\omega=B_1,B_2,\ldots,B_{n-1},B_n\]
be a sequence of all nonterminals of $N$ in hierarchical order, \ie, the following conditions hold:
\begin{enumerate}[(i)]
	\item $B_n=S$
	\item $\forall 1\leq i<j\leq n:B_j\not\leadsto_{\mathcal{G}}^\ast B_i$
\end{enumerate}
Let $(B_i\to t_i),(B_j\to t_j)\in P$, where $1\leq i,j<n$ and $i\neq
j$. If we eliminate $B_i$ this may have an impact on the value of
$\mathsf{sav}_{\mathcal{G}}(B_j)$ from \eqref{def:savValue}. We need to differentiate between two cases:
\begin{enumerate}[(1)]
	\item $B_j\to\begin{tikzpicture}[baseline=(X.base)]\node [draw,isosceles triangle,shape border rotate=90,inner sep=1pt] (X) {$B_i$};\draw ([xshift=3pt,yshift=4pt]X.north east) node {$t_j$};\end{tikzpicture}$
	
		If $B_i$ occurs in $t_j$, \ie, $B_i\leadsto_{\mathcal{G}}B_j$, then $|t_j|$ is increased because of the elimination of $B_i$. At the same time, $\mathsf{sav}_{\mathcal{G}}(B_j)$ goes up if we have $|\mathsf{ref}_{\mathcal{G}}(B_j)|>1$. The increase of $|t_j|$ is due to the fact that we can assume that the inequality $|\{v\in\mathsf{dom}_{t_i}\mid\lambda_{t_i}(v)\notin\mathcal{Y}\}|\geq2$ holds. Every production which was introduced in the replacement step represents a \tp and therefore consists of at least two nodes labeled by the parent and child symbol, respectively, of this \tp.
	\item $B_i\to\begin{tikzpicture}[baseline=(X.base)]\node [draw,isosceles triangle,shape border rotate=90,inner sep=1pt] (X) {$B_j$};\draw ([xshift=3pt,yshift=4pt]X.north east) node {$t_i$};\end{tikzpicture}$
	
		If $B_j$ occurs in $t_i$, \ie, $B_j\leadsto_{\mathcal{G}}B_i$, then $|\mathsf{ref}_{\mathcal{G}}(B_j)|$ and therefore $\mathsf{sav}_{\mathcal{G}}(B_j)$ are possibly increased by eliminating $B_i$. In fact, both values go up if $|\mathsf{ref}_\mathcal{G}(B_i)|>1$.
\end{enumerate}

\paragraph*{First phase} In the first phase of the pruning step, we eliminate every production $(A\to t)\in P$ with $|\mathsf{ref}_{\mathcal{G}}(A)|=1$. That way we achieve not only a possible reduction of the size of $\mathcal{G}$ (because we have $\mathsf{sav}_{\mathcal{G}}(A)=-\mathsf{rank}(A)$ for every $A\in N$ referenced only once) but we also decrement the number of nonterminals $|N|$ each time we eliminate such a production.

\paragraph*{Second phase} In the second phase of the pruning step we eliminate all remaining inefficient productions. We consider a production $(A\to t)\in P$ as inefficient if $\mathsf{sav}_{\mathcal{G}}(A)\leq0$. 
Unfortunately, this time we have to deal with a rather complex optimization problem. In contrast to the first phase, the decision whether to eliminate a production $(A\to t)\in P$ or not does now depend on the value $\mathsf{sav}_{\mathcal{G}}(A)$. However, the latter may be increased by eliminating other nonterminals (see the above case distinction). This forces us to use a heuristic to decide what productions to remove next from the grammar. In fact, after completing the first phase, we cycle through the remaining productions in their reverse hierarchical order. For every $(A\to t)\in P$ we check if $\mathsf{sav}_{\mathcal{G}}(A)\leq0$. If this proves to be true, we eliminate $(A\to t)$. That way $|\mathcal{G}|$ and $|N|$ are possibly further reduced.

The following example shows that the size of the final grammar generated by the Re-pair for Trees algorithm may depend on the order in which possible inefficient productions are eliminated.
\begin{example}\label{ex:differentOrdersPruningStep}
	Consider the linear SLCF tree grammar $\mathcal{G}=(N,P,S)$, where $N=\{S,A,B\}$ and $P$ is the following set of productions:
	\begin{align*}
		S&\to f(A(a,a),B(A(a,a)))\\
		A(y_1,y_2)&\to f(B(y_1),y_2)\\
		B(y_1)&\to f(y_1,a)
	\end{align*}
	Let us assume that the grammar $\mathcal{G}$ was generated by the replacement step of our algorithm and that we now want to remove all inefficient productions. We have $\mathsf{sav}_{\mathcal{G}}(A)=-1$ and $\mathsf{sav}_{\mathcal{G}}(B)=0$, \ie, the productions with left-hand sides $A$ and $B$ do not contribute to a small representation of the input tree $\mathsf{val}(\mathcal{G})$. Let us consider the following two cases:
	\begin{enumerate}[(1)]
		\item If we eliminate the production with left-hand side $A$, we obtain the grammar $\mathcal{G}_1=(N_1,P_1,S_1)$, where $N_1=\{S_1,B_1\}$ and $P_1$ is the following set of productions:
			\begin{align*}
				S_1&\to f(f(B_1(a),a),B_1(f(B_1(a),a)))\\
				B_1(y_1)&\to f(y_1,a)
			\end{align*}
			We have $|\mathcal{G}_1|=11$ and $\mathsf{sav}_{\mathcal{G}_1}(B_1)=1$, \ie, the production with left-hand side $B_1$ is not considered inefficient.
		\item In contrast, the elimination of the production with left-hand side $B$ yields the linear SLCF tree grammar $\mathcal{G}_2=(N_2,P_2,S_2)$, where $N_2=\{S_2,A_2\}$ and $P_2$ is the following set of productions:
			\begin{align*}
				S_2&\to f(A_2(a,a),f(A_2(a,a),a))\\
				A_2(y_1,y_2)&\to f(f(y_1,a),y_2)
			\end{align*}
			We also eliminate the production with left-hand side $A_2$ since we have $\mathsf{sav}_{\mathcal{G}_2}(A_2)=0$. This leads to an updated grammar $\mathcal{G}_2=(N_2,P_2,S_2)$, where $N_2=\{S_2\}$ and $P_2$ contains solely the production
			\[S_2\to f(f(f(a,a),a),f(f(f(a,a),a),a))\enspace\mbox{.}\]
			We have $|\mathcal{G}_2|=12$.
	\end{enumerate}
	This case distinction shows that the order in which inefficient productions are eliminated has an influence on the size of the final grammar (since $|\mathcal{G}_1|<|\mathcal{G}_2|$). Let us consider the sequence $A,B,S$ which is the only way to enumerate the nonterminals from $N$ in hierarchical order. Due the fact that the above described heuristic cycles through the productions in their reverse hierarchical order to eliminate inefficient productions we would obtain the larger grammar $\mathcal{G}_2$ if we would execute the pruning step with $\mathcal{G}$ as the input grammar.
\end{example}
Given the above example one might expect better compression results if the inefficient productions are eliminated in the order of their $\mathsf{sav}_{\mathcal{G}}$-values, \ie, if we would proceed as follows: as long as their is a production whose left-hand side has a $\mathsf{sav}_{\mathcal{G}}$-value smaller or equal to $0$ we remove a production whose left-hand side has the smallest occurring $\mathsf{sav}_{\mathcal{G}}$-value. However, our investigations showed that this approach leads to unappealing final grammars --- at least for our set of test input trees. The grammars generated by this approach exhibit nearly the same number of edges but much more nonterminals (about 50\% more) compared to the grammars obtained using the above heuristic.

\bigskip

\noindent Note that it is not possible to already detect \tps leading to inefficient productions during the replacement step. For instance, we would not act wisely if we would ignore \tps occurring only twice and exhibiting a large number of parameters a priori.
\begin{example}
	Imagine an input tree $t\in T(\mathcal{F})$ comprising two instances of a large tree pattern $t'\in T(\mathcal{F},\mathcal{Y})$. Let $\lambda_{t'}(v)\neq\lambda_{t'}(u)$ for all $v,u\in\mathsf{dom}_{t'}$, $u\neq v$. Furthermore, let us assume that all symbols in the tree pattern $t'$ are not occurring outside of this pattern. For every \tp $\alpha$ occurring in the tree pattern $t'$ (whose replacement may firstly lead to a production with a large number of parameters) we would have $|\mathsf{occ}_{t}(\alpha)|=2$. It becomes clear that this great redundancy in the input tree $t$, which can be represented by a production with right-hand side $t'$, would not be detected if we would not carry out these initially anything but efficient seeming \tp replacements.
\end{example}

\subsection{Complete Example}

Let the tree depicted in Fig.~\ref{fig:binaryRepresentation} be our input tree $t_0$ and let there be no restrictions on the maximal rank allowed for a nonterminal. We set $\mathcal{G}_0=(N_0,P_0,S_0)$, where $N_0=\{S_0\}$ and $P_0$ solely contains the production $(S_0\to t_0)$. Table \subref{tab:pairs0iteration} shows every \tp $\alpha$ encountered in $t_0$ along with its number of non-overlapping occurrences $|\mathsf{occ}_{t_0}(\alpha)|$. Furthermore, this table tells us that the two \tps $(\mathsf{title}^{01}, 1, \mathsf{isbn}^{00})$ and $(\mathsf{author}^{01}, 1, \mathsf{title}^{01})$ are the most frequent \tps occuring in $t_0$. We decide to replace the former \tp and therefore have $\mathsf{max}(t_0)=(\mathsf{title}^{01}, 1, \mathsf{isbn}^{00})=:\alpha_0$.

\begin{table}[t]
\centering
	\parbox{14.5cm}{
		\footnotesize
		\subtable[All \tps encountered in the input tree $t_0$ and their number of non-overlapping occurrences.]{
			\begin{tabular}{p{2.1cm}c}
				\toprule
				\tp $\alpha$&$|\mathsf{occ}_{t_0}(\alpha)|$\\
				\midrule
				$(\mathsf{title}^{01}, 1, \mathsf{isbn}^{00})$&$5$\\
				$(\mathsf{author}^{01}, 1, \mathsf{title}^{01})$&$5$\\
				$(\mathsf{book}^{11}, 1, \mathsf{author}^{01})$&$4$\\
				$(\mathsf{book}^{11}, 2, \mathsf{book}^{11})$&$2$\\
				$(\mathsf{book}^{11}, 2, \mathsf{book}^{10})$&$1$\\
				$(\mathsf{book}^{10},1,\mathsf{author}^{01})$&$1$\\
				$(\mathsf{books}^{10}, 1, \mathsf{book}^{11})$&$1$\\
				\bottomrule
			\end{tabular}
			\label{tab:pairs0iteration}
		}
		\hfill
		\subtable[All \tps encountered in the tree $t_1$ and their number of non-overlapping occurrences.]{
			\begin{tabular}{p{2.1cm}c}
				\toprule
				\tp $\alpha$&$|\mathsf{occ}_{t_1}(\alpha)|$\\
				\midrule
				$(\mathsf{author}^{01}, 1, A_1)$&$5$\\
				$(\mathsf{book}^{11}, 1, \mathsf{author}^{01})$&$4$\\
				$(\mathsf{book}^{11}, 2, \mathsf{book}^{11})$&$2$\\
				$(\mathsf{book}^{11}, 2, \mathsf{book}^{10})$&$1$\\
				$(\mathsf{book}^{10},1,\mathsf{author}^{01})$&$1$\\
				$(\mathsf{books}^{10}, 1, \mathsf{book}^{11})$&$1$\\
				&\\
				\bottomrule
			\end{tabular}
			\label{tab:pairs1iteration}
		}
		\hfill
		\subtable[All \tps encountered in the tree $t_2$ and their number of non-overlapping occurrences.]{
			\begin{tabular}{p{2.1cm}c}
				\toprule
				\tp $\alpha$&$|\mathsf{occ}_{t_2}(\alpha)|$\\
				\midrule
				$(\mathsf{book}^{11}, 1, A_2)$&$4$\\
				$(\mathsf{book}^{11}, 2, \mathsf{book}^{11})$&$2$\\
				$(\mathsf{book}^{11}, 2, \mathsf{book}^{10})$&$1$\\
				$(\mathsf{book}^{10},1,A_2)$&$1$\\
				$(\mathsf{books}^{10}, 1, \mathsf{book}^{11})$&$1$\\
				&\\
				&\\
				\bottomrule
			\end{tabular}
			\label{tab:pairs2iteration}
		}
	}
	\stepcounter{table}
\end{table}

\begin{figure}[t]
	\centering
	\begin{tikzpicture}
		[bend angle=25, node distance=0.25cm, text height=1.5ex,
		place/.style={},
		pre/.style={<-,>=stealth'},
		post/.style={->,>=stealth'}]
		\node[place] (books) {$\mathsf{books}^{10}$};
		
		\node[place] (book1) [below=of books] {$\mathsf{book}^{11}$}
			edge [pre] (books);
		\node[place] (author1) [below=of book1] {$\mathsf{author}^{01}$}
			edge [pre] (book1);
		\node[place] (A1) [below=of author1] {$A_1$}
			edge [pre] (author1);
	
		\node[place] (book2) [right=of book1] {$\cdots$}
			edge [pre] (book1);

		\node[place] (book4) [right=of book2] {$\mathsf{book}^{11}$}
			edge [pre] (book2);
		\node[place] (author4) [below=of book4] {$\mathsf{author}^{01}$}
			edge [pre] (book4);
		\node[place] (A4) [below=of author4] {$A_1$}
			edge [pre] (author4);
			
		\draw[decorate,decoration={brace,mirror}] ([xshift=-5pt]A1.south) -- ([xshift=5pt]A4.south) node[midway,sloped,below=2pt] {$4$ times};

		\node[place] (book5) [right=of book4] {$\mathsf{book}^{10}$}
			edge [pre] (book4);
		\node[place] (author5) [below=of book5] {$\mathsf{author}^{01}$}
			edge [pre] (book5);
		\node[place] (A5) [below=of author5] {$A_1$}
			edge [pre] (author5);
	\end{tikzpicture}
	\caption{Tree $t_1$ which evolved from the input tree $t_0$ in the first iteration of our computation.}\label{fig:treeAfterFirstRepl}
\end{figure}
Now, in the first iteration of our computation, we generate a new linear SLCF tree grammar $\mathcal{G}_1=(N_1,P_1,S_1)$ as follows. We introduce a new nonterminal $A_1\in\mathcal{N}_0$ and set $N_1=\{S_1,A_1\}$. After that, we introduce the new production $\big(A_1\to\mathsf{pat}(\alpha_0)\big)$, where $\mathsf{pat}(\alpha_0)=\mathsf{title}^{01}(\mathsf{isbn}^{00})$. Finally, we set $P_1=\{\big(S_1\to t_1\big),\big(A_1\to\mathsf{pat}(\alpha_0)\big)\}$, where we have $t_1=t_0[\alpha_0/A_1]$. The tree $t_1$ is depicted in Fig.~\ref{fig:treeAfterFirstRepl}.

In the second iteration, during which we generate the grammar $\mathcal{G}_2=(N_2,P_2,S_2)$, we have $\mathsf{max}(t_1)=(\mathsf{author}^{01}, 1, A_1)=:\alpha_1$ as it can be seen in Table \subref{tab:pairs1iteration}. Again, we introduce a new nonterminal $A_2\in\mathcal{N}_0$ with right-hand side $\mathsf{pat}(\alpha_1)$, set $N_2=\{S_2,A_1,A_2\}$ and set $P=\{(S_2\to t_2),\big(A_1\to\mathsf{pat}(\alpha_0)\big),\big(A_2\to\mathsf{pat}(\alpha_1)\big)\}$, where $t_2=t_1[\alpha_1/A_2]$ (see Fig.~\ref{fig:treeAfterSecondRepl}).
\begin{figure}[t]
	\centering
	\begin{tikzpicture}
		[bend angle=25, node distance=0.25cm, text height=1.5ex,
		place/.style={},
		pre/.style={<-,>=stealth'},
		post/.style={->,>=stealth'}]
		\node[place] (books) {$\mathsf{books}^{10}$};
		
		\node[place] (book1) [below=of books] {$\mathsf{book}^{11}$}
			edge [pre] (books);
		\node[place] (B1) [below=of book1] {$A_2$}
			edge [pre] (book1);
				
		\node[place] (book2) [right=of book1] {$\cdots$}
			edge [pre] (book1);
				
		\node[place] (book4) [right=of book2] {$\mathsf{book}^{11}$}
			edge [pre] (book2);
		\node[place] (B4) [below=of book4] {$A_2$}
			edge [pre] (book4);
			
		\draw[decorate,decoration={brace,mirror}] ([xshift=-5pt]B1.south) -- ([xshift=5pt]B4.south) node[midway,sloped,below=2pt] {$4$ times};
		
		\node[place] (book5) [right=of book4] {$\mathsf{book}^{10}$}
			edge [pre] (book4);
		\node[place] (B5) [below=of book5] {$A_2$}
			edge [pre] (book5);
	\end{tikzpicture}
	\caption{Tree $t_2$ which evolved from the tree $t_1$ in the second iteration of our computation.}\label{fig:treeAfterSecondRepl}
\end{figure}
We have $\mathsf{max}(t_2)=(\mathsf{book}^{11}, 1, A_2)=:\alpha_2$ (\cf Table \subref{tab:pairs2iteration}) in the third iteration of our algorithm. This time, we need to introduce a new nonterminal $A_3\in\mathcal{N}_1$, \ie, a nonterminal with one parameter, with right-hand side $\mathsf{pat}(\alpha_2)=\mathsf{book}^{11}(A_2,y_1)$. We obtain the grammar $\mathcal{G}_3=(N_3,P_3,S_3)$, where
\begin{align*}
N_3&=\{S_3,A_1,A_2,A_3\}\enspace\mbox{,}\\
P_3&=(P_2\setminus\{(S_2\to t_2)\})\cup\{(S_3\to t_3),(A_3\to\mathsf{pat}(\alpha_2))\}\text{ and}\\
t_3&=\mathsf{books}^{10}(A_3(A_3(A_3(A_3(\mathsf{book}^{10}(A_2))))))
\end{align*}
by replacing the $4$ occurrences of $\alpha_2$.

In the fourth and last iteration the \tp $(A_3, 1, A_3)$ is replaced by a new nonterminal $A_4\in\mathcal{N}_1$. Therefore, we obtain the grammar $\mathcal{G}_4=(N_4,P_4,S_4)$ with $10$ edges and $5$ nonterminals, where we have $N_4=\{S_4,A_1,A_2,A_3,A_4\}$ and $P_4$ is the following set of productions:
\begin{align*}
	S_4&\to\mathsf{books}^{10}(A_4(A_4(\mathsf{book}^{10}(A_2))))\\
	A_4(y_1)&\to A_3(A_3(y_1))\\
	A_3(y_1)&\to\mathsf{book}^{11}(A_2,y_1)\\
	A_2&\to\mathsf{author}^{01}(A_1)\\
	A_1&\to\mathsf{title}^{01}(\mathsf{isbn}^{00})
\end{align*}
Finally, in the pruning step, we begin with merging the right-hand
side of $A_1$ with the right-hand side of $A_2$ since
$|\mathsf{ref}_{\mathcal{G}_4}(A_1)|=1$, \ie, it is only referenced
once. This yields the updated production
$\big(A_2\to\mathsf{author}^{01}(\mathsf{title}^{01}(\mathsf{isbn}^{00}))\big)$. Furthermore,
we roll back the replacement of the \tp $(A_3, 1, A_3)$ due to the
fact that it does not contribute to the reduction of the total number
of edges. Although the production with left-hand side $A_4$ is
referenced twice in the right-hand sides of $\mathcal{G}_4$ and
removes redundancy this gain is neutralized by the necessary edge to
the parameter node. This is indicated by the
$\mathsf{sav}_{\mathcal{G}_4}$ value of $A_4$, see \eqref{def:savValue}:
\begin{align*}
	\mathsf{sav}_{\mathcal{G}_4}(A_4)&=|\mathsf{ref}_{\mathcal{G}_4}(A_4)|\cdot(|A_3(A_3(y_1))|-\mathsf{rank}(A_4))-|A_3(A_3(y_1))|\\
	&=2\cdot(2-1)-2=0
\end{align*}
With these adjustments we obtain the linear SLCF tree grammar $\mathcal{G}=(N,P,S_4)$\label{exampleGrammar}, where $N=\{S_4,A_2,A_3\}$ and $P$ is the following set of productions:
\begin{align*}
	S_4&\to\mathsf{books}^{10}(A_3(A_3(A_3(A_3(\mathsf{book}^{10}(A_2))))))\\
	A_3(y_1)&\to\mathsf{book}^{11}(A_2, y_1)\\
	A_2&\to\mathsf{author}^{01}(\mathsf{title}^{01}(\mathsf{isbn}^{00}))
\end{align*}
Compared to the grammar $\mathcal{G}_4$ it has the same number of edges (namely 10) but nearly half as much nonterminals only.

\subsection{Another Example}

It is very unlikely to be confronted with an XML document tree which, in the binary tree model, is represented by a perfect binary tree\footnote{A \emph{perfect binary tree} is a binary tree in which every node is either of rank $2$ or $0$ and all leaves are at the same level (\ie, the paths to the root are of the same length). In contrast, a \emph{full binary tree} has no restrictions on the level of the leaves, \ie, the only requirement is that every node is either of rank $2$ or $0$.}. Nevertheless we want to investigate the compression performance of our algorithm on this kind of trees since it is an interesting aspect from a theoretical point of view. Last but not least our undertaking is justified by the fact that the actual Re-pair for Trees algorithm is not restricted to applications processing XML files but can be used in other applications as well. The latter, in turn, may exhibit ranked trees similar to full binary trees.

Let $t\in T(\mathcal{F})$ be a sufficiently large perfect binary tree of which each inner node is labeled by a terminal $f\in\mathcal{F}_2$ and each leaf is labeled by a terminal $a\in\mathcal{F}_0$. A run of Re-pair for Trees on $t$ consists of $2\cdot(d-1)$ iterations folding the input tree beginning at its leaves, where $d=\mathsf{depth}(t)$. Thus, in the first two iterations, the \tps formed by the leaf nodes and their parents are replaced. We obtain the productions $A_1(y_1)\to f(y_1,a)$ and $A_2\to A_1(a)$ each occurring $2^{d-1}$ times. Now, we undertake further \tp replacements in a bottom up fashion. 
In the $(2i-1)$-th and $2i$-th iteration we replace two \tps resulting in the productions $A_{2i-1}(y_1)\to f(y_1,A_{2(i-1)})$ and $A_{2i}\to A_{2i-1}(A_{2(i-1)})$, respectively, where $2\leq i\leq d-1$. 

The production with left-hand side $A_{2k-1}$ occurs only once for every $1\leq k\leq d-1$. Therefore, in the pruning step, for every $1\leq k\leq d-1$ the production with left-hand side $A_{2k-1}$ is eliminated by merging its right-hand side with the right-hand side of the production with left-hand side $A_{2k}$. In particular, the production with left-hand side $A_1$ is merged with the production for $A_2$ resulting in a production $A_2\to f(a,a)$.

Finally, we obtain a linear SLCF tree grammar with $d$ nonterminals --- including the left-hand side of the start production $S\to f(A_{2(d-1)},A_{2(d-1)})$ --- and a total of $2\cdot d$ edges. Note that even though some of the intermediate productions exhibit parameters the final grammar consists only of nonterminals of rank $0$. Thus, the generated grammar is a DAG and in this particular case the minimal DAG of the input tree.
\begin{figure}[tb]
	\small
	\centering
	\parbox{14.5cm}{
		\subfigure[Perfect binary tree $t\in T(\mathcal{F})$ of height 4]{
			\parbox[t]{6cm}{
			\beginpgfgraphicnamed{figPerfectBinaryTree}
			\begin{tikzpicture}[->,>=stealth']
				\tikzstyle{level 1}=[level distance=0.8cm, sibling distance=3.2cm]
				\tikzstyle{level 2}=[level distance=0.8cm, sibling distance=1.6cm]
				\tikzstyle{level 3}=[level distance=0.8cm, sibling distance=0.8cm]
				\tikzstyle{level 4}=[level distance=0.8cm, sibling distance=0.4cm]
				\node {$f$}
					child {node {$f$}
						child {node {$f$}
							child {node {$f$}
								child {node {$a$}}
								child {node {$a$}}
							}
							child {node {$f$}
								child {node {$a$}}
								child {node {$a$}}
							}
						}
						child {node {$f$}
							child {node {$f$}
								child {node {$a$}}
								child {node {$a$}}
							}
							child {node {$f$}
								child {node {$a$}}
								child {node {$a$}}
							}
						}
					}
					child {node {$f$}
						child {node {$f$}
							child {node {$f$}
								child {node {$a$}}
								child {node {$a$}}
							}
							child {node {$f$}
								child {node {$a$}}
								child {node {$a$}}
							}
						}
						child {node {$f$}
							child {node {$f$}
								child {node {$a$}}
								child {node {$a$}}
							}
							child {node {$f$}
								child {node {$a$}}
								child {node {$a$}}
							}
						}
					}
				;
			\end{tikzpicture}
			\endpgfgraphicnamed
			}
			\label{fig:perfectBinaryTree}
		}
		\hfill
		\subfigure[Productions before the pruning step without the start production.]{
			\parbox[b]{3.3cm}{\begin{align*}
				A_1(y_1)&\to f(y_1,a)\\
				A_2&\to A_1(a)\\
				A_3(y_1)&\to f(y_1,A_2)\\
				A_4&\to A_3(A_2)\\
				A_5(y_1)&\to f(y_1,A_4)\\
				A_6&\to A_5(A_4)
			\end{align*}}
			\label{fig:prodPerfectBinaryTree}
		}
		\hfill
		\subfigure[Productions after the pruning step.]{
			\parbox[b]{3.3cm}{\begin{align*}
				A_2&\to f(a,a)\\
				A_4&\to f(A_2,A_2)\\
				A_6&\to f(A_4,A_4)\\
				S&\to f(A_6,A_6)
			\end{align*}}
			\label{fig:prodPerfectBinaryTreeAfterPruning}
		}
	}
	\stepcounter{figure}
\end{figure}
\begin{example}
	Let $t\in T(\mathcal{F})$ be the perfect binary tree from Fig.~\subref{fig:perfectBinaryTree} with 30 edges and $\mathsf{depth}(t)=4$. A run of Re-pair for Trees initially generates the 6 productions listed in Fig.~\subref{fig:prodPerfectBinaryTree}. After the pruning step we finally obtain the linear SLCF tree grammar $\mathcal{G}=(N,P,S)$, where $N=\{A_2,A_4,A_6,S\}$ and the set of productions $P$ consists of the productions from Fig.~\subref{fig:prodPerfectBinaryTreeAfterPruning}. The size of $\mathcal{G}$ is $|\mathcal{G}|=8$.
\end{example}

\subsection{Unlimited Maximal Rank}\label{sec:observations}

It seems natural to assume that, in general, trees can be compressed best by the Re-pair for Trees algorithm if there are no restrictions on the maximal rank of a nonterminal. However, it turns out that there are (not so uncommon) types of trees for which the opposite is true. Firstly, in this section, we will construct a set of trees whose compressibility is best if there are no restrictions on the maximal rank of a nonterminal. After that, in the succeeding section, we will present a set of trees whose compressibility is best when restricting the maximal rank to $1$.

Let us consider the infinite set $M=\{t_1,t_2,t_3,\ldots\}\subseteq T(\mathcal{F})$ of trees, where for all $i\in\mathbb{N}_{>0}$ the tree $t_i$ has the following properties:
\begin{itemize}
	\item The tree $t_i$ is a perfect binary tree of depth $2^i$.
	\item Each inner node of $t_i$ is labeled by the terminal $f\in\mathcal{F}_2$.
	\item Each leaf of $t_i$ is labeled by a unique terminal from $\mathcal{F}_0$, \ie, there do not exist two different leaves which are labeled by the same symbol.
\end{itemize}

\begin{figure}[p]
      \centering
	\newcommand{\ldistance}{7}
	\newcommand{\innersize}{0.6}
	\newcommand{\colorOne}{green!50!black}
	\newcommand{\colorTwo}{blue!50!black}
	\newcommand{\colorThree}{red!50!black}
	\subfigure[The tree $t_3\in M$.]{
		\begin{tikzpicture}
			[level distance=\ldistance mm,
			inner/.style={fill=black,draw=black,circle,inner sep=\innersize pt},
			leaf/.style={fill=white,draw=black,circle,inner sep=0.5pt},
			level 1/.style={sibling distance=76.8mm,level distance=5mm},
			level 2/.style={sibling distance=38.4mm,level distance=5mm},
			level 3/.style={sibling distance=19.2mm,level distance=5mm},
			level 4/.style={sibling distance=9.6mm},
			level 5/.style={sibling distance=4.8mm},
			level 6/.style={sibling distance=2.4mm},
			level 7/.style={sibling distance=1.2mm},
			level 8/.style={sibling distance=0.6mm}
			]
			
			\node[inner,\colorOne] {}
				child foreach \i/\color in {1/black,2/\colorOne} {node[inner,draw=\color,fill=\color] (\i) {} edge from parent[draw=\color]
					child foreach \ii/\color in {1/black,2/black} {node[inner,draw=\colorOne,fill=\colorOne] {} edge from parent[draw=\color]
						child foreach \iii/\color in {1/black,2/\colorOne} {node[inner,draw=\color,fill=\color] {} edge from parent[draw=\color]
							child foreach \iiii/\color in {1/black,2/black} {node[inner,draw=\colorOne,fill=\colorOne] {} edge from parent[draw=\color]
								child foreach \iiiii/\color in {1/black,2/\colorOne} {node[inner,draw=\color,fill=\color] {} edge from parent[draw=\color]
									child foreach \iiiii/\color in {1/black,2/black} {node[inner,draw=\colorOne,fill=\colorOne] {} edge from parent[draw=\color]
										child foreach \iiiii/\color in {1/black,2/\colorOne} {node[inner,draw=\color,fill=\color] {} edge from parent[draw=\color]
											child foreach \iiiiii/\color in {1/black,2/black} {node[leaf] {}  edge from parent[draw=\color]}
										}
									}
								}
							}
						}
					}
				}
			;
		\end{tikzpicture}
		\label{fig:unlimitedRank1}
	}
		
	\vspace{0.5cm}
	
	\subfigure[The tree $t_3\in M$ after replacing the \tp $(f,2,f)$.]{
		\begin{tikzpicture}
			[level distance=\ldistance mm,
			inner/.style={fill=\colorOne,draw=\colorOne,circle,inner sep=\innersize pt},
			A/.style={inner}, 
			leaf/.style={fill=white,draw=black,circle,inner sep=0.5pt},
			level 1/.style={sibling distance=76.8mm,level distance=5mm},
			level 2/.style={sibling distance=38.4mm,level distance=5mm},
			level 3/.style={sibling distance=19.2mm,level distance=5mm},
			level 4/.style={sibling distance=9.6mm},
			level 5/.style={sibling distance=4.8mm},
			level 6/.style={sibling distance=2.4mm},
			level 7/.style={sibling distance=1.2mm},
			level 8/.style={sibling distance=0.6mm}
			]
			
			\node[A] {}
				child[sibling distance=63mm] {node[inner] {} edge from parent[draw=\colorOne]
					child[sibling distance=38mm] foreach \j in {1,2} {node[A] {} edge from parent[draw=black]
						child[sibling distance=14.4mm] {node[inner] {} edge from parent[draw=\colorOne]
							child[sibling distance=9mm] foreach \ii in {1,2} {node[A] {} edge from parent[draw=black]
								child[sibling distance=3.6mm] {node[inner] {} edge from parent[draw=\colorOne]
									child[sibling distance=1.8mm] foreach \iiii in {1,2} {node[A] {} edge from parent[draw=black]
										child[sibling distance=0.6mm] {node[inner] {} edge from parent[draw=\colorOne]
											child foreach \iiiiii in {1,2} {node[leaf] {} edge from parent[draw=black]}
										}
										child[sibling distance=0.6mm] foreach \iiiiii in {1,2} {node[leaf] {} edge from parent[draw=black]}
									}
								}
								child[sibling distance=1.8mm] foreach \iiii in {1,2} {node[A] {} edge from parent[draw=black]
									child[sibling distance=0.6mm] {node[inner] {} edge from parent[draw=\colorOne]
										child[sibling distance=0.6mm] foreach \iiiiii in {1,2} {node[leaf] {} edge from parent[draw=black]}
									}
									child[sibling distance=0.6mm] foreach \iiiiii in {1,2} {node[leaf] {} edge from parent[draw=black]}
								}
							}
						}
						child[sibling distance=9mm] foreach \ii in {1,2} {node[A] {} edge from parent[draw=black]
							child[sibling distance=3.6mm] {node[inner] {} edge from parent[draw=\colorOne]
								child[sibling distance=1.8mm] foreach \iiii in {1,2} {node[A] {} edge from parent[draw=black]
									child[sibling distance=0.6mm] {node[inner] {} edge from parent[draw=\colorOne]
										child[sibling distance=0.6mm] foreach \iiiiii in {1,2} {node[leaf] {} edge from parent[draw=black]}
									}
									child[sibling distance=0.6mm] foreach \iiiiii in {1,2} {node[leaf] {} edge from parent[draw=black]}
								}
							}
							child[sibling distance=1.8mm] foreach \iiii in {1,2} {node[A] {} edge from parent[draw=black]
								child[sibling distance=0.6mm] {node[inner] {} edge from parent[draw=\colorOne]
									child[sibling distance=0.6mm] foreach \iiiiii in {1,2} {node[leaf] {} edge from parent[draw=black]}
								}
								child[sibling distance=0.6mm] foreach \iiiiii in {1,2} {node[leaf] {} edge from parent[draw=black]}
							}
						}
					}
				}
				child[sibling distance=38mm] foreach \jj in {1,2} {node[A] {} edge from parent[draw=black]
					child[sibling distance=14.4mm] {node[inner] {} edge from parent[draw=\colorOne]
						child[sibling distance=9mm] foreach \ii in {1,2} {node[A] {} edge from parent[draw=black]
							child[sibling distance=3.6mm] {node[inner] {} edge from parent[draw=\colorOne]
								child[sibling distance=1.8mm] foreach \iiii in {1,2} {node[A] {} edge from parent[draw=black]
									child[sibling distance=0.6mm] {node[inner] {} edge from parent[draw=\colorOne]
										child[sibling distance=0.6mm] foreach \iiiiii in {1,2} {node[leaf] {} edge from parent[draw=black]}
									}
									child[sibling distance=0.6mm] foreach \iiiiii in {1,2} {node[leaf] {} edge from parent[draw=black]}
								}
							}
							child[sibling distance=1.8mm] foreach \iiii in {1,2} {node[A] {} edge from parent[draw=black]
								child[sibling distance=0.6mm] {node[inner] {} edge from parent[draw=\colorOne]
									child[sibling distance=0.6mm] foreach \iiiiii in {1,2} {node[leaf] {} edge from parent[draw=black]}
								}
								child[sibling distance=0.6mm] foreach \iiiiii in {1,2} {node[leaf] {} edge from parent[draw=black]}
							}
						}
					}
					child[sibling distance=9mm] foreach \ii in {1,2} {node[A] {} edge from parent[draw=black]
						child[sibling distance=3.6mm] {node[inner] {} edge from parent[draw=\colorOne]
							child[sibling distance=1.8mm] foreach \iiii in {1,2} {node[A] {} edge from parent[draw=black]
								child[sibling distance=0.6mm] {node[inner] {} edge from parent[draw=\colorOne]
									child[sibling distance=0.6mm] foreach \iiiiii in {1,2} {node[leaf] {} edge from parent[draw=black]}
								}
								child[sibling distance=0.6mm] foreach \iiiiii in {1,2} {node[leaf] {} edge from parent[draw=black]}
							}
						}
						child[sibling distance=1.8mm] foreach \iiii in {1,2} {node[A] {} edge from parent[draw=black]
							child[sibling distance=0.6mm] {node[inner] {} edge from parent[draw=\colorOne]
								child[sibling distance=0.6mm] foreach \iiiii in {1,2} {node[leaf] {} edge from parent[draw=black]}
							}
							child[sibling distance=0.6mm] foreach \iiiiii in {1,2} {node[leaf] {} edge from parent[draw=black] }
						}
					}
				}
			;
		\end{tikzpicture}
		\label{fig:unlimitedRank2}
	}
		
	\vspace{0.5cm}
	
	\subfigure[The tree $t_3\in M$ after the second iteration, \ie, after replacing the \tps $(f,2,f)$ and $(A_1,1,f)$.]{
		\begin{tikzpicture}
			[level distance=\ldistance mm,
			inner/.style={fill=black,draw=black,circle,inner sep=\innersize pt},
			A/.style={inner}, 
			B/.style={inner}, 
			leaf/.style={fill=white,draw=black,circle,inner sep=0.5pt},
			level 1/.style={sibling distance=76.8mm,level distance=5mm},
			level 2/.style={sibling distance=38.4mm,level distance=5mm},
			level 3/.style={sibling distance=19.2mm},
			level 4/.style={sibling distance=9.6mm},
			level 5/.style={sibling distance=4.8mm},
			level 6/.style={sibling distance=2.4mm},
			level 7/.style={sibling distance=1.2mm},
			level 8/.style={sibling distance=0.6mm}
			]
			
			\node[B,draw=\colorOne,fill=\colorOne] {}
				child[sibling distance=38.4mm,draw=\colorOne,fill=\colorOne] foreach \ii/\color in {1/black,2/black,3/black,4/\colorOne} {node[B,draw=\color,fill=\color] {} edge from parent[draw=\color]
					child[sibling distance=9.6mm] foreach \iiii in {1,2,3,4} {node[B,draw=\colorOne,fill=\colorOne] {}  edge from parent[draw=black]
						child[sibling distance=2.4mm] foreach \iiiii/\color in {1/black,2/black,3/black,4/\colorOne} {node[B,draw=\color,fill=\color] {} edge from parent[draw=\color]
							child[sibling distance=0.6mm] foreach \iiiiii in {1,2,3,4} {node[leaf] {}  edge from parent[draw=black]}
						}
					}
				}
			;
		\end{tikzpicture}
		\label{fig:unlimitedRank3}
	}
		
	\vspace{0.5cm}
	
	\subfigure[The tree which remains after replacing the \tp $(A_2,4,A_2)$ in the tree from Fig.~\subref{fig:unlimitedRank3}.]{
		\begin{tikzpicture}
			[level distance=\ldistance mm,
			inner/.style={fill=black,draw=black,circle,inner sep=\innersize pt},
			A/.style={inner},
			B/.style={inner},
			leaf/.style={fill=white,draw=black,circle,inner sep=0.5pt}
			]
			
			\node[B,draw=\colorOne,fill=\colorOne] {}
				child[sibling distance=35mm] {node[B] {}
					child[sibling distance=9.6mm] foreach \iiii in {1,2,3,4} {node[B,draw=\colorOne,fill=\colorOne] {}
						child[sibling distance=2.4mm] foreach \iiiii/\color in {1/black,2/black,3/\colorOne} {node[B,draw=\color,fill=\color] {} edge from parent[draw=\color]
							child[sibling distance=0.6mm] foreach \iiiiii in {1,2,3,4} {node[leaf] {}  edge from parent[draw=black]}
						}
						child[sibling distance=0.6mm] foreach \iiiiii in {1,2,3,4} {node[leaf] {} edge from parent[draw=black]}
					}
				}
				child[sibling distance=33mm] {node[B] {}
					child[sibling distance=9.6mm] foreach \iiii in {1,2,3,4} {node[B,draw=\colorOne,fill=\colorOne] {}
						child[sibling distance=2.4mm] foreach \iiiii/\color in {1/black,2/black,3/\colorOne} {node[B,draw=\color,fill=\color] {} edge from parent[draw=\color]
							child[sibling distance=0.6mm] foreach \iiiiii in {1,2,3,4} {node[leaf] {}  edge from parent[draw=black]}
						}
						child[sibling distance=0.6mm] foreach \iiiiii in {1,2,3,4} {node[leaf] {} edge from parent[draw=black]}
					}
				}
				child[sibling distance=27mm] {node[B,draw=\colorOne,fill=\colorOne] {} edge from parent[draw=\colorOne]
					child[sibling distance=9.6mm] foreach \iiii in {1,2,3,4} {node[B,draw=\colorOne,fill=\colorOne] {} edge from parent[draw=black]
						child[sibling distance=2.4mm] foreach \iiiii/\color in {1/black,2/black,3/\colorOne} {node[B,draw=\color,fill=\color] {} edge from parent[draw=\color]
							child[sibling distance=0.6mm] foreach \iiiiii in {1,2,3,4} {node[leaf] {}  edge from parent[draw=black]}
						}
						child[sibling distance=0.6mm] foreach \iiiiii in {1,2,3,4} {node[leaf] {} edge from parent[draw=black]}
					}
				}
				child[sibling distance=9.6mm] {node[B,draw=\colorOne,fill=\colorOne] {}
					child[sibling distance=2.4mm] foreach \iiiii/\color in {1/black,2/black,3/\colorOne} {node[B,draw=\color,fill=\color] {} edge from parent[draw=\color]
						child[sibling distance=0.6mm] foreach \iiiiii in {1,2,3,4} {node[leaf] {}  edge from parent[draw=black]}
					}
					child[sibling distance=0.6mm] foreach \iiiiii in {1,2,3,4} {node[leaf] {} edge from parent[draw=black]}
				}
				child[sibling distance=9.6mm] foreach \iiii in {1,2,3} {node[B,draw=\colorOne,fill=\colorOne] {}
					child[sibling distance=2.4mm] foreach \iiiii/\color in {1/black,2/black,3/\colorOne} {node[B,draw=\color,fill=\color] {} edge from parent[draw=\color]
						child[sibling distance=0.6mm] foreach \iiiiii in {1,2,3,4} {node[leaf] {}  edge from parent[draw=black]}
					}
					child[sibling distance=0.6mm] foreach \iiiiii in {1,2,3,4} {node[leaf] {} edge from parent[draw=black]}
				}
			;
		\end{tikzpicture}
		\label{fig:unlimitedRank4}
	}
	
	\vspace{0.5cm}
	
	\subfigure[The tree $t_3\in M$ after $6$ iterations of our algorithm. We obtained a $16$-ary tree whose inner nodes are labeled by the nonterminal $A_6$.]{
		\begin{tikzpicture}
			[level distance=\ldistance mm,
			inner/.style={fill=black,draw=black,circle,inner sep=\innersize pt},
			A/.style={inner}, 
			B/.style={inner}, 
			leaf/.style={fill=white,draw=black,circle,inner sep=0.5pt},
			level 1/.style={sibling distance=76.8mm},
			level 2/.style={sibling distance=38.4mm},
			level 3/.style={sibling distance=19.2mm},
			level 4/.style={sibling distance=9.6mm},
			level 5/.style={sibling distance=4.8mm},
			level 6/.style={sibling distance=2.4mm},
			level 7/.style={sibling distance=1.2mm},
			level 8/.style={sibling distance=0.6mm}
			]
			
			\node[B] {}
				child[sibling distance=9.6mm] foreach \ii in {1,2,3,4,5,6,7,8,9,10,11,12,13,14,15,16} {node[B] {}
					child[sibling distance=0.6mm] foreach \iiiiii in {1,2,3,4,5,6,7,8,9,10,11,12,13,14,15,16} {node[leaf] {}}
				}
			;
		\end{tikzpicture}
		\label{fig:unlimitedRank5}
	}
	\stepcounter{figure}
\end{figure}
\begin{example}\label{ex:unlimitedRank1}
	Figure \subref{fig:unlimitedRank1} shows a simplified depiction of the tree $t_3\in M$. The inner nodes labeled by the symbol $f\in\mathcal{F}_2$ are represented by a circle filled with paint. In contrast, the leaves, of which each is labeled by a unique symbol from $\mathcal{F}_0$, are depicted by a circle which is not filled with paint. 
	
	The tree $t_3$ is compressed by a run of our algorithm as follows. The \tps $(f,1,f)$ and $(f,2,f)$ occur equally often in $t_3$. It makes no difference to the size of the final grammar whether we replace the former or the latter. Let us replace the \tp $(f,2,f)$ (whose occurrences are painted in green in Fig.~\subref{fig:unlimitedRank1}) by a nonterminal $A_1\in\mathcal{N}_3$ with right-hand side $f(y_1,f(y_2,y_3))$. We obtain the tree of the form shown in Fig.~\subref{fig:unlimitedRank2}. After that, the \tp $(A_1,1,f)$, which occurs the same number of times as $(f,2,f)$ did, is replaced by the nonterminal $A_2\in\mathcal{N}_4$ with right-hand side $A_1\big(f(y_1,y_2),y_3,y_4\big)$. The occurrences of $(A_1,1,f)$ are marked with green paint in Fig.~\subref{fig:unlimitedRank2}. The right-hand side of the nonterminal $A_1$ is merged with the right-hand side of $A_2$ during the pruning step since $A_1$ is only referenced once. This yields the production with left-hand side $A_2$ and right-hand side $f\big(f(y_1,y_2),f(y_3,y_4)\big)$. 
	
	After the replacement of the above two \tps the right-hand side of the start production is a $4$-ary tree of depth $4$ whose inner nodes are labeled by $A_2$ (see Fig.~\subref{fig:unlimitedRank3}).
	Now, the \tps
	\[(A_2,1,A_2),(A_2,2,A_2),(A_2,3,A_2),(A_2,4,A_2)\] 
	occur equally often. Again, the order of the \tp replacements makes no difference to the final grammar. Assuming that at first we replace the \tp $(A_2,4,A_2)$, which is marked with green paint in Fig.~\subref{fig:unlimitedRank3}, by a new nonterminal $A_3$, we obtain the tree shown in Fig.~\subref{fig:unlimitedRank4}. After that, the \tps $(A_3,3,A_2)$, $(A_4,2,A_2)$ and $(A_5,1,A_2)$ are replaced in three additional iterations. The above four \tp replacements result in a new production
	\[A_6(y_1,\ldots,y_{16})\to A_2\big(A_2(y_1,\ldots,y_4),\ldots,A_2(y_{13},\ldots,y_{16})\big)\]
	after pruning the grammar. The remaining tree is a 16-ary tree of depth $2$ (of the form depicted in Fig.~\subref{fig:unlimitedRank5}) whose inner nodes are labeled by the nonterminal $A_6$. In this tree there is no \tp occurring more than once. Therefore, the execution of our algorithm stops.
\end{example}
\begin{figure}[tb]
	\centering\footnotesize
	\begin{tikzpicture}[semithick]
		\tikzstyle{level 1}=[level distance=1.35cm, sibling distance=3.5cm]
		\tikzstyle{level 2}=[level distance=1.35cm, sibling distance=0.8cm]
		\begin{scope}[->,>=stealth']
			\node {$B_{i-1}$}
				child {node (B1) {$B_{i-1}$}
					child {node {$y_1$}}
					child {node (y11) {$y_2$}}
					child[missing] {node {$y_i$}}
					child {node (y12) {$y_{r}$}}
				}
				child {node (B2) {$B_{i-1}$}
					child {node {$y_{r+1}$}}
					child {node (y21) {$y_{r+2}$}}
					child[missing] {node {$y_i$}}
					child {node (y22) {$y_{2\cdot r}$}}
				}
				child[missing] {node {$B_{i-1}$}}
				child[sibling distance=4cm] {node (Bn) {$B_{i-1}$}
					child[sibling distance=2cm] {node {$y_{(r-1)\cdot r+1}$}}
					child[sibling distance=2cm] {node (y31) {$y_{(r-1)\cdot r+2}$}}
					child[missing] {node {$y_i$}}
					child[sibling distance=1.2cm] {node (y32) {$y_{r^2}$}}
				}
			;
		\end{scope}
		
		\draw[dotted] ([xshift=1cm]B2.center) -- ([xshift=-1cm]Bn.center);
		\draw[dotted] ([xshift=1cm]y11.center) -- ([xshift=-1cm]y12.center);
		\draw[dotted] ([xshift=1cm]y21.center) -- ([xshift=-1cm]y22.center);
		\draw[dotted] ([xshift=1cm]y31.center) -- ([xshift=-0.5cm]y32.center);
	\end{tikzpicture}
	\caption{Right-hand side $s_i$ of the nonterminal $B_i$, where $r=\mathsf{rank}(B_{i-1})$ and $i>1$.}\label{fig:rhsOfB}
\end{figure}

\bigskip

\noindent Now, we want to analyze the behavior of Re-pair for Trees on a tree from $M$ in general. Let $x\in\mathbb{N}_{>0}$ and let $\mathsf{it}:\mathbb{N}_{>0}\to\mathbb{N}_{>0}$ be the following function:
\[\mathsf{it}(x)=\sum_{i=0}^{x-1}2^{2^i}\]
Let $B_1,B_2,B_3,\ldots$ be a sequence of nonterminals where for all $i>0$ the following conditions are fulfilled:
\begin{itemize}
	\item $\mathsf{rank}(B_i)=2^{2^i}$
	\item $s_i\in T(\mathcal{F}\cup\mathcal{N},\mathcal{Y})$ is the right-hand side of $B_i$
	\item If $i=1$, we have $s_i=f(f(y_1,y_2),f(y_3,y_4))$ and if $i>1$, the tree $s_i$ is of the form shown in Fig.~\ref{fig:rhsOfB}, where $r=\mathsf{rank}(B_{i-1})=2^{2^{i-1}}$.
\end{itemize}
Regarding the nonterminals $A_2$ and $A_6$ from Example \ref{ex:unlimitedRank1}, we have $B_1=A_2$ and \mbox{$B_2=A_6$}, respectively. Let $i\in\{1,2,\ldots,k\}$. The following two equations hold:
\begin{align}
	\mathsf{rank}(B_i)&=2^{2^i}=2^{2^{i-1}}\cdot2^{2^{i-1}}=\mathsf{rank}(B_{i-1})^2\label{eq:rankQuad}\\
	|s_i|&=\mathsf{rank}(B_i)+\mathsf{rank}(B_{i-1})\label{eq:sizeOfTi}
\end{align}
For convenience, we define $\mathsf{rank}(B_0)=\mathsf{rank}(f)=2$.

Now assume that we have an unlimited maximal rank allowed for a nonterminal. After  $\mathsf{it}(n)$ iterations on $t_{n+1}\in M$ we have obtained the nonterminals $B_1,B_2,\ldots,B_n$. The right-hand side of the start nonterminal is a $\mathsf{rank}(B_n)$-ary tree of height $2$ (see also Example~\ref{ex:unlimitedRank1}, where $n=2$). At this point, no further replacements are carried out. For each of the generated nonterminals $B_1,\ldots,B_n$ we have
\begin{equation}\label{eq:refBi}
	|\mathsf{ref}_{\mathcal{G}}(B_i)|=\mathsf{rank}(B_i)+1\enspace\mbox{,}
\end{equation}
where $i\in\{1,2,\ldots,n\}$ (\cf Fig.~\ref{fig:rhsOfB}). Hence, we have
\begin{align*}
	\mathsf{sav}_{\mathcal{G}}(B_i)&\overset{(\ref{def:savValue})}=|\mathsf{ref}_{\mathcal{G}}(B_i)|\cdot\mathsf{rank}(B_{i-1})-|s_{i}|\\
	&\overset{(\ref{eq:sizeOfTi})}=|\mathsf{ref}_{\mathcal{G}}(B_i)|\cdot\mathsf{rank}(B_{i-1})-\mathsf{rank}(B_i)-\mathsf{rank}(B_{i-1})\\
	&\overset{(\ref{eq:refBi})}=\mathsf{rank}(B_i)\cdot\mathsf{rank}(B_{i-1})-\mathsf{rank}(B_i)\\
	&\overset{(\ref{eq:rankQuad})}=\mathsf{rank}(B_{i-1})^3-\mathsf{rank}(B_{i-1})^2\,>\,0
\end{align*}
since $\mathsf{rank}(B_{i-1})\geq\mathsf{rank}(B_0)=2$. Therefore, none of the nonterminals $B_1,\ldots,B_n$ will be eliminated in the pruning step.

Now assume that the maximal rank is $m\in\mathbb{N}$, \ie, we have $m<\infty$. Choose the smallest $n\in\mathbb{N}$ such that
\begin{figure}[tb]
	\centering
	\begin{tikzpicture}[semithick]
		\tikzstyle{level 1}=[level distance=1.35cm]
		\tikzstyle{level 2}=[level distance=1.35cm]
		\begin{scope}[->,>=stealth']
			\node {$B_{n-1}$}
				child {node (y1) {$y_1$}}
				child {node (y2) {$y_2$}}
				child[missing] {node (y3) {$y_3$}}
				child {node (y4) {$y_{r-h}$}}
				child[sibling distance=2.5cm] {node (B2) {$B_{n-1}$}
					child[sibling distance=1.25cm] {node (y20) {$y_{r-h+1}$}}
					child[sibling distance=1.25cm] {node (y21) {$y_{r-h+2}$}}
					child[missing] {node {$y_i$}}
					child[sibling distance=0.75cm] {node (y22) {$y_{2r-h}$}}
				}
				child[missing] {node {$B_{n-1}$}}
				child[sibling distance=2.75cm] {node (Bn) {$B_{n-1}$}
					child[sibling distance=1.75cm] {node (y30) {$y_{1+h(r-1)}$}}
					child[sibling distance=1.75cm] {node (y31) {$y_{2+h(r-1)}$}}
					child[missing] {node {$y_i$}}
					child[sibling distance=1.2cm] {node (y32) {$y_{r+h(r-1)}$}}
				}
			;
		\end{scope}
		
		\draw[dotted] ([xshift=1cm]B2.center) -- ([xshift=-1cm]Bn.center);
		\draw[dotted] ([xshift=0.6cm]y2.center) -- ([xshift=-0.6cm]y4.center);
		\draw[dotted] ([xshift=1cm]y21.center) -- ([xshift=-1cm]y22.center);
		\draw[dotted] ([xshift=1cm]y31.center) -- ([xshift=-1cm]y32.center);
		
		\draw[decorate,decoration={brace,mirror}] ([xshift=-7pt]y1.south) -- ([xshift=8pt]y4.south) node[midway,sloped,below=2pt] {\footnotesize $(r-h)$ many};
		\draw[decorate,decoration={brace,mirror}] ([xshift=-10pt,yshift=4pt]B2.south) -- ([xshift=10pt,yshift=4pt]Bn.south) node[midway,sloped,below=2pt] {\footnotesize $h$ many};
		\draw[decorate,decoration={brace,mirror}] ([xshift=-13pt]y20.south) -- ([xshift=10pt]y22.south) node[midway,sloped,below=2pt] {\footnotesize $r$ many};
		\draw[decorate,decoration={brace,mirror}] ([xshift=-20pt]y30.south) -- ([xshift=18pt]y32.south) node[midway,sloped,below=2pt] {\footnotesize $r$ many};
	\end{tikzpicture}
	\caption{Right-hand side of the nonterminal $C$, where $r=\mathsf{rank}(B_{n-1})$.}\label{fig:rhsOfC}
\end{figure}
\begin{equation}\label{eq:BiggerThanM}
	2^{2^n}>m\enspace\mbox{.}
\end{equation}
Thus, $B_n$ is the first nonterminal in the sequence $B_1,B_2,\ldots$ with a rank bigger than $m$. Let us consider a run of Re-pair for Trees on a tree $t_j\in M$ with $j\geq n+1$. Then, as above, the nonterminals $B_1,\ldots,B_{n-1}$ will be obtained after $\mathsf{it}(n-1)$ iterations (if we would prune the corresponding grammar by now). At this point, the right-hand side of the start production is a $2^{2^{n-1}}$-ary tree of height $\nicefrac{2^j}{2^{n-1}}\geq4$, where all inner nodes are labeled by the nonterminal $B_{n-1}$. Now, we can carry out $h$ additional \tp replacements leading to the nonterminals $C_1,C_2,\ldots,C_h\in\mathcal{N}$, where
\begin{equation}
h=\max\{l\in\mathbb{N}\mid r+l\cdot(r-1)\leq m\}\label{eq:sizeOfH}
\end{equation}
and $r=\mathsf{rank}(B_{n-1})=2^{2^{n-1}}$. We claim that
\begin{equation}\label{eq:rankBiggerH}
	r=\mathsf{rank}(B_{n-1})>h
\end{equation}
holds. To see this, let us assume that $r=\mathsf{rank}(B_{n-1})\leq h$. We have
\[m\overset{(\ref{eq:sizeOfH})}\geq r+h\cdot(r-1)\geq r+r\cdot(r-1)=r^2=2^{2^n}\enspace\mbox{.}\]
However, this contradicts (\ref{eq:BiggerThanM}).
\begin{figure}[tb]
	\centering
	\newcommand{\ldistance}{7}
	\newcommand{\innersize}{0.6}
	\begin{tikzpicture}
		[level distance=\ldistance mm,
		inner/.style={fill=black,circle,inner sep=\innersize pt},
		A/.style={inner},
		B/.style={inner},
		leaf/.style={fill=white,draw=black,circle,inner sep=0.5pt}
		]
		
		\node[B] {}
			child[sibling distance=20mm] {node[B] {}
				child[sibling distance=9.6mm] foreach \iiii in {1,2,3,4} {node[B] {}
					child[sibling distance=1.2mm] {node[B] {}
						child[sibling distance=0.6mm] foreach \iiiiii in {1,2,3,4} {node[leaf] {}}
					}
					child[sibling distance=0.7mm] {node[B] {}
						child[sibling distance=0.6mm] foreach \iiiiii in {1,2,3,4} {node[leaf] {}}
					}
					child[sibling distance=0.6mm] foreach \iiiiii in {1,2,3,4,5,6,7,8} {node[leaf] {}}
				}
			}
			child[sibling distance=15mm] {node[B] {}
				child[sibling distance=9.6mm] foreach \iiii in {1,2,3,4} {node[B] {}
					child[sibling distance=1.2mm] {node[B] {}
						child[sibling distance=0.6mm] foreach \iiiiii in {1,2,3,4} {node[leaf] {}}
					}
					child[sibling distance=0.7mm] {node[B] {}
						child[sibling distance=0.6mm] foreach \iiiiii in {1,2,3,4} {node[leaf] {}}
					}
					child[sibling distance=0.6mm] foreach \iiiiii in {1,2,3,4,5,6,7,8} {node[leaf] {}}
				}
			}
			child[sibling distance=9.6mm] {node[B] {}
				child[sibling distance=1.2mm] {node[B] {}
					child[sibling distance=0.6mm] foreach \iiiiii in {1,2,3,4} {node[leaf] {}}
				}
				child[sibling distance=0.7mm] {node[B] {}
					child[sibling distance=0.6mm] foreach \iiiiii in {1,2,3,4} {node[leaf] {}}
				}
				child[sibling distance=0.6mm] foreach \iiiiii in {1,2,3,4,5,6,7,8} {node[leaf] {}}
			}
			child[sibling distance=9.6mm] foreach \iiii in {1,2,3,4,5,6,7} {node[B] {}
				child[sibling distance=1.2mm] {node[B] {}
					child[sibling distance=0.6mm] foreach \iiiiii in {1,2,3,4} {node[leaf] {}}
				}
				child[sibling distance=0.7mm] {node[B] {}
					child[sibling distance=0.6mm] foreach \iiiiii in {1,2,3,4} {node[leaf] {}}
				}
				child[sibling distance=0.6mm] foreach \iiiiii in {1,2,3,4,5,6,7,8} {node[leaf] {}}
			}
		;
	\end{tikzpicture}
	\caption{Right-hand side of the current start production after replacing the \tp $(A_3,3,A_2)$.}\label{fig:unlimitedRankLimitedExample}
\end{figure}

In case $h>0$, we can argue as follows: After the pruning step, the nonterminals $C_1,C_2,\ldots,C_h$ form one nonterminal $C\in\mathcal{N}$ with $\mathsf{rank}(C)=h\cdot r+r-h=r + h\cdot(r-1)$ (see Fig.~\ref{fig:rhsOfC}). It occurs at least $2^{2^n}+1=r^2+1$ many times according to (\ref{eq:refBi}) (the nonterminal $C$ occurs as often as $B_n$ does after $\mathsf{it}(n)$ iterations on $t_j$ in the unlimited case). Each occurrence of $C$ reduces the size of the corresponding grammar by $h$ edges and the right-hand side of $C$ consists of $r+h\cdot r$ edges (see Fig.~\ref{fig:rhsOfC}). Now, let us consider the $\mathsf{sav}$-value of $C$ (assuming that $\mathcal{G}$ is the current grammar after $\mathsf{it}(n)+h$ iterations):
\begin{align*}
	\mathsf{sav}_{\mathcal{G}}(C)&\overset{(\ref{def:savValue})}
             =|\mathsf{ref}_{\mathcal{G}}(C)|\cdot h-(r+h\cdot r)\\
	&\geq(r^2+1)\cdot h-h\cdot r-r\\
	&=(r^2-r+1)\cdot h-r\\
	&\geq r^2-2r+1\\
	&=(r-1)^2
\end{align*}
Thus, we have $\mathsf{sav}_{\mathcal{G}}(C)>0$, \ie, the nonterminal $C$ is not eliminated during the pruning step. 
\begin{example}\label{ex:unlimitedRank2}
	Let us assume that the maximal rank for a nonterminal is restricted to $10$ in Example \ref{ex:unlimitedRank1}. In this case we are able to undertake exactly one additional \tp replacement in the tree from Fig.~\subref{fig:unlimitedRank4} resulting in a new nonterminal $A_4\in\mathcal{N}_{10}$. If we replace the \tp $(A_3,3,A_2)$, we obtain the tree shown in Fig.~\ref{fig:unlimitedRankLimitedExample}. We have $n=2$, $h=2$ and $C=A_4$. After the pruning step, the right-hand side of $A_4$ is of the form 
	\[A_2(y_1,y_2,A_2(y_3,y_4,y_5,y_6),A_2(y_7,y_8,y_9,y_{10}))\enspace\mbox{.}\]
\end{example}
\bigskip
We can further state that the nonterminal $B_{n-1}$ is not eliminated since it occurs $h+1$ times in the right-hand side of $C$ (see Fig.~\ref{fig:rhsOfC}) and $(r-h)\cdot|\mathsf{ref}_{\mathcal{G}}(C)|\geq(r-h)\cdot(r^2+1)$ times in the right-hand side of the current start production (below each occurrence of $C$ there are $r-h$ occurrences of $B_{n-1}$ and $C$ occurs at least $r^2+1$ times). Therefore, we have
\begin{align*}
|\mathsf{ref}_{\mathcal{G}}(B_{n-1})|&\geq h+1+(r-h)\cdot(r^2+1)\\
&=h+1+r^3-hr^2+r-h\\
&=r^3-hr^2+r+1\enspace\mbox{.}
\end{align*}
Because of (\ref{eq:rankBiggerH}), the inequality $|\mathsf{ref}_{\mathcal{G}}(B_{n-1})|>r+1$ holds. As shown before for the unlimited rank, in this case $B_{n-1}$ has a $\mathsf{sav}$-value bigger than $0$ and therefore the nonterminals $B_1,B_2,\ldots,B_{n-1}$ are not eliminated.

Let ${\mathcal{H}^m}$ be the grammar which is obtained after $\mathsf{it}(n-1)+h$ iterations on the tree $t_j$ when restricting the maximal rank to $m$ and let ${\mathcal{H}^\infty}$ be the current grammar after $\mathsf{it}(n)$ iterations on $t_j$ when an unlimited rank is allowed. We can conclude that $|{\mathcal{H}^m}|>|{\mathcal{H}^\infty}|$ holds --- no matter whether we have $h>0$ or $h=0$ --- because of the following two facts:
\begin{enumerate}[(1)]
	\item Each occurrence of $B_n$ saves $\mathsf{rank}(B_{n-1})$ edges (see Fig.~\ref{fig:rhsOfB}) and therefore according to (\ref{eq:rankBiggerH}) more than an occurrence of $C$ does. The nonterminals $B_n$ and $C$ occur equally often. However, $C$ is only existent if $h>0$.
	\item The nonterminals $B_1,B_2,\ldots,B_{n-1}$ (which are existent in both grammars, ${\mathcal{H}^m}$ and ${\mathcal{H}^\infty}$) and the nonterminals $B_n$ and $C$ are not eliminated during the pruning step.
\end{enumerate}
Let $\mathcal{G}^m$ ($\mathcal{G}^\infty$) be the final grammar which is generated by a run of Re-pair for Trees on the tree $t_j$ when restricting the maximal rank of a nonterminal to $m$ (not restricting the maximal rank). We have $\mathcal{G}^m={\mathcal{H}^m}$ and $|\mathcal{G}^\infty|\leq|{\mathcal{H}^\infty}|$. The latter holds because with every additional \tp replacement at least one edge is absorbed and because during the pruning step only nonterminals with a $\mathsf{sav}$-value smaller than or equal to $0$ are eliminated. Therefore $|\mathcal{G}^m|>|\mathcal{G}^\infty|$ holds. Thus, we have shown that, in general, the trees from $M$ can be compressed best if there are no restrictions on the maximal rank allowed for a nonterminal.
\begin{example}
	Table \ref{tbl:comparisonRuns} shows a comparison of the grammars generated by different runs of our algorithm on the trees $t_2$, $t_3$ and $t_4$ from $M$. By $\mathcal{G}^4$ ($\mathcal{G}^\infty$) we denote the final grammar which is generated when restricting the maximal rank to $4$ (not restricting the maximal rank).
\end{example}

\subsection{Limiting the Maximal Rank}\label{sec:limitingTheMaximalRank}

\begin{table}[tb]
	\centering
		\begin{tabular}{crrrr}
			\toprule
			Tree	$t_i$	&$\mathsf{depth}(t_i)$	&$|t_i|$	&$|\mathcal{G}^4|$&$|\mathcal{G}^\infty|$\\
			\midrule
			$t_2$		&$4$					&$30$	&$26$			&$26$\\
			$t_3$		&$8$					&$510$	&$346$			&$298$\\
			$t_4$		&$16$				&$131070$&$87386$			&$66090$\\
			\bottomrule
		\end{tabular}
	\caption{Comparison of the sizes of the final grammars.}\label{tbl:comparisonRuns}
\end{table}

In the preceding section we investigated a set of trees whose compressibility was best if we did not restrict the maximal rank of a nonterminal. Now, we want to construct a set of trees which behaves contrarily, \ie, we construct trees which can be compressed best if we limit the maximal rank of a nonterminal to $1$. In order to make it easier to quickly understand the following definition we want to refer the reader to Fig.~\subref{fig:unlimitedCaseRhsG0} which shows one of the trees we define in the sequel.

First of all, let us define a labeling function $l:\mathbb{N}\to\mathcal{F}_0$, where
\[l(i)=
	\begin{cases}
		a\qquad&\text{if }i\equiv 0\mod 5\\
		b&\text{if }i\equiv 1\mod 5\\
		c&\text{if }i\equiv 2\mod 5\\
		d&\text{if }i\equiv 3\mod 5\\
		e&\text{if }i\equiv 4\mod 5
	\end{cases}
\]
and $i\in\mathbb{N}$. Now, we define for all $n\in\mathbb{N}$ the tree $s_n=(\mathsf{dom}_{s_n},\lambda_{s_n})\in T(\mathcal{F})$, where
\[\mathsf{dom}_{s_n}=\left(\bigcup_{i=0}^{2^n}[2]^i\right)\cup\left(\bigcup_{i=0}^{2^n-1}[2]^i[1]\right)\]
and
\[
	\lambda_{s_n}(v)=\begin{cases}
		f\in\mathcal{F}_2&\text{if }v=[2]^i,\;0\leq i<2^n\\
		l(i)\in\mathcal{F}_0&\text{if }v=[2]^i[1],\;0\leq i<2^n\\
		l(2^n)\in\mathcal{F}_0\quad&\text{if }v=[2]^{2^n}\enspace\mbox{.}
	\end{cases}
\]
Let us define $U=\{s_n\mid n\in\mathbb{N},n\geq3\}$. In the following we will show that for every run of Re-pair for Trees on a tree $s\in U$ we have $|\mathcal{G}^1|<|\mathcal{G}^\infty|$, where $\mathcal{G}^1$ is the grammar generated when allowing a maximal rank of $1$ for a nonterminal and $\mathcal{G}^\infty$ is the resulting grammar when there is no restriction on the maximal rank.

Let us consider a run $\mathcal{G}^\infty_0,\mathcal{G}^\infty_1,\ldots,\mathcal{G}^\infty_{n-1}$ of the Re-pair for Trees algorithm on the tree $s_n$ with no restrictions on the maximal rank of a nonterminal, where $\mathcal{G}^\infty_i=(N_i,P_i,S_i)$, $(S_i\to t_i)$ is the start production of $\mathcal{G}^\infty_i$ and $i\in\{0,1,\ldots,n-1\}$. In the first iteration of our computation the \tp $(f,2,f)$ is the most frequent \tp, \ie, $\mathsf{max}(t_0)=(f,2,f)$. This is because of $|\mathsf{occ}_{s_n}\big((f,2,f)\big)|=2^{n-1}$ whereas for every $x\in\{a,b,c,d,e\}$ the inequality
\[|\mathsf{occ}_{s_n}\big((f,1,x)\big)|\leq\left\lceil\nicefrac{2^n}{5}\right\rceil\]
holds. Therefore, we replace the \tp $(f,2,f)$ by a new nonterminal $A_1$ and obtain $\mathcal{G}^\infty_1$. In every subsequent iteration $i$ we replace $\mathsf{max}(t_{i-1})=(A_{i-1},2^{i-1}+1,A_{i-1})$ by a new nonterminal $A_i$, where $i\in\{2,3,\ldots,n-1\}$. For every $1\leq i\leq n-1$ the right-hand side of the start production of the grammar $\mathcal{G}^\infty_i$ is given by the tree $t_i=(\mathsf{dom}_{t_i},\lambda_{t_i})$, where
\begin{align*}
	\mathsf{dom}_{t_i}&=\left(\bigcup_{j=0}^{2^{n-i}}
{\left[2^i+1\right]}^j\right)\cup\left(\bigcup_{k=1}^{2^i}\bigcup_{j=0}^{2^{n-i}-1}{\left[\left(2^i+1\right)\right]}^j[k]\right)\\
	\intertext{and}
	\lambda_{t_i}(v)&=\begin{cases}
		A_i\in\mathcal{N}_{2^i+1}&\text{if }v=[2^i+1]^j\text{ with }0\leq j\leq 2^{n-i}-1\enspace\mbox{,}\\
		l(j\cdot 2^i + k-1)&\text{if }v=[2^i+1]^j[k]\text{ with }0\leq j\leq 2^{n-i}-1\text{ and }1\leq k\leq 2^i\enspace\mbox{,}\\
		l(2^n)&\text{if }v=[2^i+1]^{2^{n-i}}\enspace\mbox{.}
	\end{cases}
\end{align*}

\begin{figure}[p]
	\small
	\centering
	\parbox{14.5cm}{
		\subfigure[The tree $s_4\in U$ which is the right-hand side of $\mathcal{G}_0$'s start production.]{
			\begin{tikzpicture}[->,>=stealth',semithick,level distance=0.8cm, sibling distance=1.1cm]
				\begin{scope}[xshift=-2.5cm]
				\draw[draw=white] (-1,0) -- (-1,-1);
				\node (f1) {$f$} [grow=-110]
					child {node {$a$}}
					child {node (f2) {$f$}
						child {node {$b$}}
						child {node (f3) {$f$}
							child {node {$c$}}
							child {node (f4) {$f$}
								child {node {$d$}}
								child {node (f5) {$f$}
									child {node {$e$}}
									child {node (f6) {$f$}
										child {node {$a$}}
										child {node (f7) {$f$}
											child {node {$b$}}
											child {node (f8) {$f$}
												child {node {$c$}}
												child[missing] {node {$d$}}
											}	
										}	
									}			
								}
							}
						}
					}
				;

				\end{scope}
				\begin{scope}[xshift=0.25cm]
				\node (f9) {$f$} [grow=-110]
					child {node {$d$}}
					child {node (f10) {$f$}
						child {node {$e$}}
						child {node (f11) {$f$}
							child {node {$a$}}
							child {node (f12) {$f$}
								child {node {$b$}}
								child {node (f13) {$f$}
									child {node {$c$}}
									child {node (f14) {$f$}
										child {node {$d$}}
										child {node (f15) {$f$}
											child {node {$e$}}
											child {node (f16) {$f$}
												child {node {$a$}}
												child {node {$b$}}
											}	
										}	
									}			
								}
							}
						}
					}
				;
				\end{scope}
				
				\clip (2,0.7) rectangle (-3,-7.5); 
				\draw (f8) .. controls +(3,-3) and +(-3,3) .. (f9);
			\end{tikzpicture}
			\label{fig:unlimitedCaseRhsG0}
		}
		\hfill
		\subfigure[The right-hand side of $\mathcal{G}_1$'s start production.]{
			\begin{tikzpicture}[->,>=stealth',semithick,level distance=1cm, sibling distance=1cm,scale=0.9]
				\tikzstyle{fan}=[anchor=north,isosceles triangle, shape border uses incircle,inner sep=0.5pt,shape border rotate=90,draw]
			
				\node (f1) {$A_1$} [grow=-105]
					child {node {$a$}}
					child {node {$b$}}
					child {node (f2) {$A_1$}
						child {node {$c$}}
						child {node {$d$}}
						child {node (f3) {$A_1$}
							child {node {$e$}}
							child {node {$a$}}
							child {node (f4) {$A_1$}
								child {node {$b$}}
								child {node {$c$}}
								child {node (f5) {$A_1$}
									child {node {$d$}}
									child {node {$e$}}
									child {node (f6) {$A_1$}
										child {node {$a$}}
										child {node {$b$}}
										child {node (f7) {$A_1$}
											child {node {$c$}}
											child {node {$d$}}
											child {node (f8) {$A_1$}
												child {node {$e$}}
												child {node {$a$}}
												child {node {$b$}}
											}
										}
									}
								}
							}
						}
					}
				;
			\end{tikzpicture}
			\label{fig:unlimitedCaseRhsG1}
		}
		\\
		\vspace{1cm}
		\\
		\subfigure[The right-hand side of $\mathcal{G}_2$'s start production.]{
			\begin{tikzpicture}[->,>=stealth',semithick,level distance=1cm, sibling distance=1cm,scale=0.95]
				\tikzstyle{fan}=[anchor=north,isosceles triangle, shape border uses incircle,inner sep=0.5pt,shape border rotate=90,draw]
			
				\node (f1) {$A_2$} [grow=-125]
					child {node {$a$}}
					child {node {$b$}}
					child {node {$c$}}
					child {node {$d$}}
					child {node (f2) {$A_2$}
						child {node {$e$}}
						child {node {$a$}}
						child {node {$b$}}
						child {node {$c$}}
						child {node (f3) {$A_2$}
							child {node {$d$}}
							child {node {$e$}}
							child {node {$a$}}
							child {node {$b$}}
							child {node (f4) {$A_2$}
								child {node {$c$}}
								child {node {$d$}}
								child {node {$e$}}
								child {node {$a$}}
								child {node {$b$}}
							}
						}
					}
				;
			\end{tikzpicture}
			\label{fig:unlimitedCaseRhsG2}
		}
		\hfill
		\subfigure[The right-hand side of $\mathcal{G}_3$'s start production.]{
			\begin{tikzpicture}[->,>=stealth',semithick,level distance=1cm, sibling distance=0.9cm,scale=0.85]
				\tikzstyle{fan}=[anchor=north,isosceles triangle, shape border uses incircle,inner sep=0.5pt,shape border rotate=90,draw]
			
				\node {$A_3$} [grow=-133]
					child {node {$a$}}
					child {node {$b$}}
					child {node {$c$}}
					child {node {$d$}}
					child {node {$e$}}
					child {node {$a$}}
					child {node {$b$}}
					child {node {$c$}}
					child {node {$A_3$}
						child {node {$d$}}
						child {node {$e$}}
						child {node {$a$}}
						child {node {$b$}}
						child {node {$c$}}
						child {node {$d$}}
						child {node {$e$}}
						child {node {$a$}}
						child {node {$b$}}
					}
				;
			\end{tikzpicture}
			\label{fig:unlimitedCaseRhsG3}
		}
	}
	\stepcounter{figure}
\end{figure}
\begin{example}
The Figs.~\subref{fig:unlimitedCaseRhsG0}, \subref{fig:unlimitedCaseRhsG1}, \subref{fig:unlimitedCaseRhsG2} and \subref{fig:unlimitedCaseRhsG3} show the right-hand sides of the start productions of the grammars $\mathcal{G}_0$, $\mathcal{G}_1$, $\mathcal{G}_2$ and $\mathcal{G}_3$ generated by a run of our algorithm on the tree $s_4$.
\end{example}
In order to argue that we have $\mathsf{max}(t_{i})=(A_i,2^i+1,A_i)=:\alpha_i$ for every $0<i<n$, we investigate the number of occurrences of all \tps occurring in the right-hand side of $\mathcal{G}^\infty_i$'s start production. Firstly, it is easy to verify that $|\mathsf{occ}_{t_i}(\alpha_i)|=2^{n-i-1}$. In contrast, for every $1\leq k\leq 2^i$ and $x\in\{a,b,c,d,e\}$ the inequality $\mathsf|occ_{t_i}\big((A_{i-1},k,x)\big)|\leq\left\lfloor\nicefrac{2^{n-i}}{5}\right\rfloor$ holds. This is because every power of $2$ is not divisible by $5$, \ie, for every $1\leq k\leq 2^i$ and every $0\leq j\leq 2^{n-i}-5$ we have 
\begin{multline*}
	\lambda_{t_i}([2^i+1]^{j}[k])\;\neq\;\lambda_{t_i}([2^i+1]^{j+1}[k])\;\neq\;\lambda_{t_i}([2^i+1]^{j+2}[k])\;\\
	\neq\;\lambda_{t_i}([2^i+1]^{j+3}[k])\;\neq\;\lambda_{t_i}([2^i+1]^{j+4}[k])\enspace\mbox{.}
\end{multline*}
Due to the fact that we do not replace \tps with child symbols $a$, $b$, $c$, $d$ or $e$, the right-hand side of $\mathcal{G}^\infty_{n-1}$'s start production has to contain at least $2^n$ nodes labeled by these symbols, \ie, we can conclude that $|\mathcal{G}^\infty_{n-1}|\geq 2^n$. Therefore the compression ratio cannot be better than 50\%.

In contrast, a run $\mathcal{G}^1_0,\mathcal{G}^1_1,\ldots,\mathcal{G}^1_k$ of our
algorithm on the tree $s_n$ leads to a significantly better compression ratio 
when restricting the maximal rank of a nonterminal to $1$, where 
\mbox{$k\in\mathbb{N}_{>0}$}, \mbox{$\mathcal{G}^1_i=(N_i,P_i,S_i)$},
$(S_i\to t_i)$ is the start production of $\mathcal{G}^1_i$ and
$i\in\{0,1,\ldots,k\}$. In the first iteration we have
$\mathsf{max}_1(t_0)\neq(f,2,f)$, since a replacement of $(f,2,f)$
would result in a nonterminal with a rank greater than $1$. Therefore
only the \tps $(f,1,a)$, $(f,1,b)$, $(f,1,c)$, $(f,1,d)$, $(f,1,e)$
and subsequent \tps can be replaced. It turns out that after the first
nine iterations the pattern $f(a,f(b,f(c,f(d,f(e,\ldots))))$ is
represented by a new nonterminal $A_9$ with
$\mathsf{rank}(A_9)=1$. The actual order of the replacements within
the first nine iterations depends on the method used to choose a most
frequent \tp when there are multiple most frequent \tps. Refer to
Example \ref{ex:limitedRankBetterLimitedRank} for one possible proceeding.

\begin{table}[tb]
	\centering
		\begin{tabular}{cccc}
			\toprule
			Iteration	&Replaced \tp		&New nonterminal	&\cf Figure\\
			\midrule
			$1$			&$(f,1,a)$		&$A_1$			&\subref{fig:limitedCaseRhsG1}\\
			$2$			&$(f,1,b)$		&$A_2$			&\subref{fig:limitedCaseRhsG2}\\
			$3$			&$(A_1,1,A_2)$	&$A_3$			&\subref{fig:limitedCaseRhsG3}\\
			$4$			&$(f,1,c)$		&$A_4$			&\subref{fig:limitedCaseRhsG4}\\
			$5$			&$(f,1,d)$		&$A_5$			&\subref{fig:limitedCaseRhsG5}\\
			$6$			&$(A_3,1,A_4)$	&$A_6$			&\subref{fig:limitedCaseRhsG6}\\
			$7$			&$(A_6,1,A_5)$	&$A_7$			&\subref{fig:limitedCaseRhsG7}\\
			$8$			&$(f,1,e)$		&$A_8$			&\subref{fig:limitedCaseRhsG8}\\
			$9$			&$(A_7,1,A_8)$	&$A_9$			&\subref{fig:limitedCaseRhsG9}\\
			\bottomrule
		\end{tabular}
	\caption{A run of Re-pair for Trees on the tree $s_4\in U$ with a maximal nonterminal rank of $1$.}\label{tbl:limitedRankBetterLimitedCaseExample}
\end{table}
The right-hand side of $\mathcal{G}_9$'s start production is a degenerated tree mainly consisting of consecutive nonterminals $A_9$. The corresponding nodes --- there are roughly $\nicefrac{2^n}{5}$ of them --- are then boiled down using approximately $\log_2(\nicefrac{2^n}{5})$ \tp replacements. Therefore the number of total edges of the resulting grammar is in $\mathcal{O}(n)$, \ie, it is of logarithmic size (the size of the input tree $s_n$ is $2^{n+1}+1$). Thus, we were able to construct a set of trees which exhibit a better compressibility when restricting the maximal rank of a nonterminal to $1$.
\begin{example}\label{ex:limitedRankBetterLimitedRank}
	Let us consider a run of Re-pair for Trees on the tree $s_4\in U$ when restricting the maximal rank of a nonterminal to $1$ (see Fig.~\subref{fig:limitedCaseRhsG0} for a depiction of $s_4$). Table \ref{tbl:limitedRankBetterLimitedCaseExample} shows one of several possible orders of \tp replacements and the Fig.~\ref{fig:limitedCaseRhsG0-G9} shows how the right-hand sides of the start productions evolve.
\end{example}
			
\begin{figure}[p]
	\centering
	\parbox{14.5cm}{
		\subfigure
		  {
				\begin{tikzpicture}[scale=0.8,->,>=stealth',semithick,level distance=0.8cm, sibling distance=1.1cm]
					\begin{scope}[xshift=-2.5cm]
					\node (f1) {$f$} [grow=-110]
						child {node {$a$}}
						child {node (f2) {$f$}
							child {node {$b$}}
							child {node (f3) {$f$}
								child {node {$c$}}
								child {node (f4) {$f$}
									child {node {$d$}}
									child {node (f5) {$f$}
										child {node {$e$}}
										child {node (f6) {$f$}
											child {node {$a$}}
											child {node (f7) {$f$}
												child {node {$b$}}
												child {node (f8) {$f$}
													child {node {$c$}}
													child[missing] {node {$d$}}
												}	
											}	
										}			
									}
								}
							}
						}
					;
	
					\end{scope}
					\begin{scope}[xshift=0.25cm]
					\node (f9) {$f$} [grow=-110]
						child {node {$d$}}
						child {node (f10) {$f$}
							child {node {$e$}}
							child {node (f11) {$f$}
								child {node {$a$}}
								child {node (f12) {$f$}
									child {node {$b$}}
									child {node (f13) {$f$}
										child {node {$c$}}
										child {node (f14) {$f$}
											child {node {$d$}}
											child {node (f15) {$f$}
												child {node {$e$}}
												child {node (f16) {$f$}
													child {node {$a$}}
													child {node {$b$}}
												}	
											}	
										}			
									}
								}
							}
						}
					;
					\end{scope}
					
					\clip (2,0.7) rectangle (-3,-7.5); 
					\draw (f8) .. controls +(3,-3) and +(-3,3) .. (f9);
				\end{tikzpicture}
				\label{fig:limitedCaseRhsG0}
		}
		\hfill
		\subfigure
		   {
				\begin{tikzpicture}[scale=0.8,->,>=stealth',semithick,level distance=0.8cm, sibling distance=1.1cm]
					\begin{scope}[xshift=-2.5cm]
					\node {$A_1$} [grow=-110]
						child[missing] {node {$a$}}
						child {node {$f$}
							child {node {$b$}}
							child {node {$f$}
								child {node {$c$}}
								child {node {$f$}
									child {node {$d$}}
									child {node {$f$}
										child {node {$e$}}
										child {node {$A_1$}
											child[missing] {node {$a$}}
											child {node {$f$}
												child {node {$b$}}
												child {node (f8) {$f$}
													child {node {$c$}}
													child[missing] {node {$d$}}
												}	
											}	
										}			
									}
								}
							}
						}
					;
					\end{scope}
					\begin{scope}[xshift=0.25cm]
					\node (f9) {$f$} [grow=-110]
						child {node {$d$}}
						child {node {$f$}
							child {node {$e$}}
							child {node {$A_1$}
								child[missing] {node {$a$}}
								child {node {$f$}
									child {node {$b$}}
									child {node {$f$}
										child {node {$c$}}
										child {node {$f$}
											child {node {$d$}}
											child {node {$f$}
												child {node {$e$}}
												child {node {$A_1$}
													child[missing] {node {$a$}}
													child {node {$b$}}
												}	
											}	
										}			
									}
								}
							}
						}
					;
					\end{scope}
					
					\clip (2,0.7) rectangle (-3,-7.5); 
					\draw (f8) .. controls +(3,-3) and +(-3,3) .. (f9);
				\end{tikzpicture}
				\label{fig:limitedCaseRhsG1}
		}
		\hfill
		\subfigure
		    {
				\begin{tikzpicture}[scale=0.8,->,>=stealth',semithick,level distance=0.8cm, sibling distance=1.1cm]
					\begin{scope}[xshift=-2.5cm]
					\node {$A_1$} [grow=-110]
						child[missing] {node {$a$}}
						child {node {$A_2$}
							child[missing] {node {$b$}}
							child {node {$f$}
								child {node {$c$}}
								child {node {$f$}
									child {node {$d$}}
									child {node {$f$}
										child {node {$e$}}
										child {node {$A_1$}
											child[missing] {node {$a$}}
											child {node {$A_2$}
												child[missing] {node {$b$}}
												child {node (f8) {$f$}
													child {node {$c$}}
													child[missing] {node {$d$}}
												}	
											}	
										}			
									}
								}
							}
						}
					;
					\end{scope}
					\begin{scope}[xshift=0.25cm]
					\node (f9) {$f$} [grow=-110]
						child {node {$d$}}
						child {node {$f$}
							child {node {$e$}}
							child {node {$A_1$}
								child[missing] {node {$a$}}
								child {node {$A_2$}
									child[missing] {node {$b$}}
									child {node {$f$}
										child {node {$c$}}
										child {node {$f$}
											child {node {$d$}}
											child {node {$f$}
												child {node {$e$}}
												child {node {$A_1$}
													child[missing] {node {$a$}}
													child {node {$b$}}
												}	
											}	
										}			
									}
								}
							}
						}
					;
					\end{scope}
					
					\clip (2,0.7) rectangle (-3,-7.5); 
					\draw (f8) .. controls +(3,-3) and +(-3,3) .. (f9);
				\end{tikzpicture}
				\label{fig:limitedCaseRhsG2}
		}
		\\
		\vspace{0.5cm}
		\subfigure
		    {
				\begin{tikzpicture}[scale=0.8,->,>=stealth',semithick,level distance=0.8cm, sibling distance=1.1cm]
					\begin{scope}[xshift=-2.5cm]
					\node {$A_3$} [grow=-110]
							child[missing] {node {$b$}}
							child {node {$f$}
								child {node {$c$}}
								child {node {$f$}
									child {node {$d$}}
									child {node {$f$}
										child {node {$e$}}
										child {node {$A_3$}
												child[missing] {node {$b$}}
												child {node (f8) {$f$}
													child {node {$c$}}
													child[missing] {node {$d$}}
												}	
										}			
									}
								}
							}
					;
					\end{scope}
					\begin{scope}[xshift=0.25cm]
					\node (f9) {$f$} [grow=-110]
						child {node {$d$}}
						child {node {$f$}
							child {node {$e$}}
							child {node {$A_3$}
									child[missing] {node {$b$}}
									child {node {$f$}
										child {node {$c$}}
										child {node {$f$}
											child {node {$d$}}
											child {node {$f$}
												child {node {$e$}}
												child {node {$A_1$}
													child[missing] {node {$a$}}
													child {node {$b$}}
												}	
											}	
										}			
									}
							}
						}
					;
					\end{scope}
					
					\clip (2,0.9) rectangle (-3,-6.5); 
					\draw (f8) .. controls +(3,-3) and +(-3,3) .. (f9);
				\end{tikzpicture}
				\label{fig:limitedCaseRhsG3}
		}
		\hfill
		\subfigure
		   {
				\begin{tikzpicture}[scale=0.8,->,>=stealth',semithick,level distance=0.8cm, sibling distance=1.1cm]
					\begin{scope}[xshift=-2.5cm]
					\node {$A_3$} [grow=-110]
							child[missing] {node {$b$}}
							child {node {$A_4$}
								child[missing] {node {$c$}}
								child {node {$f$}
									child {node {$d$}}
									child {node {$f$}
										child {node {$e$}}
										child {node {$A_3$}
												child[missing] {node {$b$}}
												child {node (f8) {$A_4$}
													child[missing] {node {$c$}}
													child[missing] {node {$d$}}
												}	
										}			
									}
								}
							}
					;
					\end{scope}
					\begin{scope}[xshift=0.25cm]
					\node (f9) {$f$} [grow=-110]
						child {node {$d$}}
						child {node {$f$}
							child {node {$e$}}
							child {node {$A_3$}
									child[missing] {node {$b$}}
									child {node {$A_4$}
										child[missing] {node {$c$}}
										child {node {$f$}
											child {node {$d$}}
											child {node {$f$}
												child {node {$e$}}
												child {node {$A_1$}
													child[missing] {node {$a$}}
													child {node {$b$}}
												}	
											}	
										}			
									}
							}
						}
					;
					\end{scope}
					
					\clip (2,0.9) rectangle (-3,-6.5); 
					\draw (f8) .. controls +(3,-3) and +(-3,3) .. (f9);
				\end{tikzpicture}
				\label{fig:limitedCaseRhsG4}
		}
		\hfill
		\subfigure
		    {
				\begin{tikzpicture}[scale=0.8,->,>=stealth',semithick,level distance=0.8cm, sibling distance=1.1cm]
					\begin{scope}[xshift=-2.5cm]
					\node {$A_3$} [grow=-110]
							child[missing] {node {$b$}}
							child {node {$A_4$}
								child[missing] {node {$c$}}
								child {node {$A_5$}
									child[missing] {node {$d$}}
									child {node {$f$}
										child {node {$e$}}
										child {node {$A_3$}
												child[missing] {node {$b$}}
												child {node (f8) {$A_4$}
													child[missing] {node {$c$}}
													child[missing] {node {$d$}}
												}	
										}			
									}
								}
							}
					;
					\end{scope}
					\begin{scope}[xshift=0.25cm]
					\node (f9) {$A_5$} [grow=-110]
						child[missing] {node {$d$}}
						child {node {$f$}
							child {node {$e$}}
							child {node {$A_3$}
									child[missing] {node {$b$}}
									child {node {$A_4$}
										child[missing] {node {$c$}}
										child {node {$A_5$}
											child[missing] {node {$d$}}
											child {node {$f$}
												child {node {$e$}}
												child {node {$A_1$}
													child[missing] {node {$a$}}
													child {node {$b$}}
												}	
											}	
										}			
									}
							}
						}
					;
					\end{scope}
					
					\clip (2,0.9) rectangle (-3,-6.5); 
					\draw (f8) .. controls +(3,-3) and +(-3,3) .. (f9);
				\end{tikzpicture}
				\label{fig:limitedCaseRhsG5}
		}
		\\
		\vspace{0.5cm}
		\subfigure
		   {
				\begin{tikzpicture}[scale=0.8,->,>=stealth',semithick,level distance=0.8cm, sibling distance=1.1cm]
					\begin{scope}[xshift=-2.5cm]
					\node {$A_6$} [grow=-110]
								child[missing] {node {$c$}}
								child {node {$A_5$}
									child[missing] {node {$d$}}
									child {node {$f$}
										child {node {$e$}}
										child {node (f8) {$A_6$}
													child[missing] {node {$c$}}
													child[missing] {node {$d$}}
										}			
									}
								}
					;
					\end{scope}
					\begin{scope}[xshift=0.25cm]
					\node (f9) {$A_5$} [grow=-110]
						child[missing] {node {$d$}}
						child {node {$f$}
							child {node {$e$}}
							child {node {$A_6$}
										child[missing] {node {$c$}}
										child {node {$A_5$}
											child[missing] {node {$d$}}
											child {node {$f$}
												child {node {$e$}}
												child {node {$A_1$}
													child[missing] {node {$a$}}
													child {node {$b$}}
												}	
											}	
										}			
							}
						}
					;
					\end{scope}
					
					\clip (2,1) rectangle (-3,-5.5); 
					\draw (f8) .. controls +(3,-3) and +(-3,3) .. (f9);
				\end{tikzpicture}
				\label{fig:limitedCaseRhsG6}
		}
		\hfill
		\subfigure
		    {
				\begin{tikzpicture}[scale=0.8,->,>=stealth',semithick,level distance=0.8cm, sibling distance=1.1cm]
					\node {$A_7$} [grow=-110]
						child[missing] {node {$d$}}
						child {node {$f$}
							child {node {$e$}}
							child {node (f8) {$A_7$}
								child[missing] {node {$c$}}
								child {node {$f$}
									child {node {$e$}}
									child {node {$A_7$}
										child[missing] {node {$d$}}
										child {node {$f$}
											child {node {$e$}}
											child {node {$A_1$}
												child[missing] {node {$a$}}
												child {node {$b$}}
											}	
										}			
									}
								}
							}			
						}
					;
				\end{tikzpicture}
				\label{fig:limitedCaseRhsG7}
		}
		\hfill
		\subfigure
		    {
				\begin{tikzpicture}[scale=0.8,->,>=stealth',semithick,level distance=0.8cm, sibling distance=1.1cm]
					\node {$A_7$} [grow=-110]
						child[missing] {node {$d$}}
						child {node {$A_8$}
							child[missing] {node {$e$}}
							child {node (f8) {$A_7$}
								child[missing] {node {$c$}}
								child {node {$A_8$}
									child[missing] {node {$e$}}
									child {node {$A_7$}
										child[missing] {node {$d$}}
										child {node {$A_8$}
											child[missing] {node {$e$}}
											child {node {$A_1$}
												child[missing] {node {$a$}}
												child {node {$b$}}
											}	
										}		
									}
								}
							}			
						}
					;
				\end{tikzpicture}
				\label{fig:limitedCaseRhsG8}
		}
		\hfill
		\subfigure
		   {
				\begin{tikzpicture}[scale=0.8,->,>=stealth',semithick,level distance=0.8cm, sibling distance=1.1cm]
					\node {$A_9$} [grow=-110]
							child[missing] {node {$e$}}
							child {node (f8) {$A_9$}
									child[missing] {node {$e$}}
									child {node {$A_9$}
											child[missing] {node {$e$}}
											child {node {$A_1$}
												child[missing] {node {$a$}}
												child {node {$b$}}
											}	
									}
							}			
					;
				\end{tikzpicture}
				\label{fig:limitedCaseRhsG9}
		}
	}
	\caption{\label{fig:limitedCaseRhsG0-G9}The right-hand sides for the nonterminals $S_0, \ldots, S_9$}
	\stepcounter{figure}
\end{figure}

\section{Implementation Details}\label{ch:implementationDetails}

We implemented a prototype of the Re-pair for Trees algorithm, named \trp, running on XML documents. In the sequel, we demonstrate that it produces for any XML document tree in $\mathcal{O}(|t|)$ time a linear $k$-bounded SLCF tree grammar ${\cal G}$, where $k\in\mathbb{N}$ is a constant, $\mathsf{val}(\mathcal{G})=t$ and $t\in T(\mathcal{F})$ is the binary representation of the input tree.

There are several reasons to restrict the maximal rank to a constant $k$. One of them is that only this way we are able to obtain a linear-time implementation. Another reason is that for every $k$-bounded linear SLCF tree grammar $\mathcal{G}$ generated by \trp it can be checked in polynomial time if a given tree automaton accepts $\mathsf{val}(\mathcal{G})$ (using a result from \cite{Lohrey2006complexity}). Last but not least, Sect.~\ref{sec:limitingTheMaximalRank} on page \pageref{sec:limitingTheMaximalRank} showed us that for flat XML documents leading to a right-leaning binary tree it is quite promising to restrict the maximal rank. The latter reason is also supported by our experiments with different maximal ranks on our test set of XML documents. 

On average, a maximal rank of $4$ leads to the best compression performance (\cf Sect.~\ref{sec:resultsWithDifferentMaximalRanks} on page \pageref{sec:resultsWithDifferentMaximalRanks}). Due to this fact \trp generates $4$-bounded linear SLCF tree grammars by default. This can be adjusted by using the \verb|-max_rank| switch.

\subsection{Reading the Input Tree}

The XML document tree of the input file can be directly transformed into a binary \mbox{$\mathcal{F}$-labeled} tree $t=(\mathsf{dom}_t,\lambda_t)\in T(\mathcal{F})$.\footnote{Refer to Sect.~\ref{sec:binaryTreeModel} on page \pageref{sec:binaryTreeModel} for an explanation of the binary tree model.} The XML document is parsed by a SAX-like parser calling the functions \verb|start-element| and \verb|end-element| (see Figs.~\ref{lst:functionStartElement} and \ref{lst:functionEndElement}) of an object taking care of the tree construction. The latter is called \emph{tree constructor} in the sequel. 

The tree constructor uses three stacks to properly encode the SAX events. Firstly, the stack \verb|index_stack| keeps track of the index\footnote{Analogously to our definition for ranked trees: If an element is the $n$-th child of its parent element, then the index of this element is $n$.} of the current element read. The stack \verb|name-stack| stores the element types of the elements in order to be able to update the labeling function $\lambda_t$ within the \verb|end-element| function. Together with the stack \verb|hierarchy_stack|, which is used to maintain the current sequence of parents within $t$, enough information stands by to encode the SAX events.

\begin{figure}[tb]
\lstset{emph={FUNCTION,ENDFUNC,return,traverse,endtraverse,for,if,else,then,do,endif,endfor}, emphstyle=\bfseries, emph={[2]index_stack,hierarchy_stack,name_stack,name}, emphstyle={[2]\slshape}, morecomment=[l]{//}}
\begin{lstlisting}{}
FUNCTION start-element(name)
	if (hierarchy_stack is not empty) then
		(*@$i:=\;$@*)index_stack.top()(*@$\;+\;1$@*);
		index_stack.pop();
		index_stack.push((*@$i$@*));

		(*@$v:=\;$@*)hierarchy_stack.top();
		
		if ((*@$i=1$@*)) then (*@$u:=v1$@*)
		else	(*@$u:=v2$@*)
		endif
		
		name_stack.push(name);
	else
		(*@$u:=\varepsilon$@*);
		(*@$\lambda_t(\varepsilon):=\;$@*)name(*@${}^{10}$@*);
	endif
	
	(*@$\mathsf{dom}_t:=\mathsf{dom}_t\cup\{u\}$@*);

	index_stack.push((*@$0$@*));
	hierarchy_stack.push((*@$u$@*));
ENDFUNC
\end{lstlisting}
\caption{The start-element function which is called for every start-tag.}\label{lst:functionStartElement}
\end{figure}

To be more precise, if the parser encounters a start-tag, it extracts the element type of the element and passes it to the tree constructor by calling the function \verb|start-element|. If it is the first call of \verb|start-element|, we must be dealing with the root of the document. Thus, the stack \verb|hierarchy_stack| is empty and the \verb|else|-part beginning in line 15 is processed. First of all, the variable $u$ is identified with $\varepsilon$ (and later added to the set $\mathsf{dom}_t$). Afterwards, the labeling function $\lambda_t$ is updated accordingly. Since, in the binary tree model, the root has no sibling nodes and since it is assumed that the input tree consists of at least two nodes, it is clear that the terminal symbol labeling the root node will have a left child but no right child (therefore the superscript $10$ in line 16).

If we consider a subsequent call of \verb|start-element|, the hierarchy stack is not empty and therefore the \verb|if|-part is processed. Firstly, the index stack is updated in the lines 3--5 and after that the node $v\in\mathsf{dom}_t$ is retrieved from the hierarchy stack (line 7). The tree node $v$ will be the parent of the node which is added in the following. We introduce a new node $u$ which is later (but still in the same call of this function) added to $\mathsf{dom}_t$ (line 19). The node $u$ becomes the left child of $v$ if it represents the first child element of the element which is represented by $v$. In contrast, $u$ becomes a right child if the current index $i$ is greater than one, \ie, if the element being processed is a sibling element of the element represented by $v$. Regarding the node $u$, we are unable to update the labeling function $\lambda_t$ at this time since we do not know if the XML element being processed has children or sibling elements. 

\begin{figure}[tb]
\lstset{emph={FUNCTION,ENDFUNC,return,traverse,endtraverse,for,if,else,then,do,endif,endfor,repeat,endrepeat,times}, emphstyle=\bfseries, emph={[2]index_stack,hierarchy_stack,name_stack,name,lchild,rchild}, emphstyle={[2]\slshape}, morecomment=[l]{//}}
\begin{lstlisting}{}
FUNCTION end-element
	(*@$i:=\;$@*)index_stack.top();
	repeat (*@$i$@*) times
		(*@$v:=\;$@*)hierarchy_stack.top();
		name(*@$\;:=\;$@*)name_stack.top();

		(*@$l:=0$@*), (*@$r:=0$@*);
		if ((*@$v1\in\mathsf{dom}_t$@*)) then
			(*@$l:=1$@*);
		endif
		if ((*@$v2\in\mathsf{dom}_t$@*)) then
			(*@$r:=1$@*);
		endif

		(*@$\lambda_t(v):=\;$@*)name(*@${}^{lr}$@*);

		hierarchy_stack.pop();
		name_stack.pop();
	endrepeat
	index_stack.pop();
ENDFUNC
\end{lstlisting}
\caption{The end-element function which is called for every end-tag encountered in the input XML document.}\label{lst:functionEndElement}
\end{figure}

If an end-tag is encountered by the input parser, the function \verb|end-element| listed in Fig.~\ref{lst:functionEndElement} is called. Now, the index of the current XML element is consulted in order to bubble up the sequence of parents stored by the hierarchy stack the correct number of times. Lastly, after processing the \verb|repeat| loop, the node representing the first child element of the current XML element (the end-tag of its last child element was just read) is on top of the hierarchy stack. For every node $v\in\mathsf{dom}_t$ which is removed from the hierarchy stack within the \verb|repeat| loop the labeling function $\lambda_t$ is updated.

\begin{example}
Fig.~\ref{fig:constructionOfBinaryTree} shows the evolution of the
data structures after the first calls to the functions \texttt{start-element()}
and \texttt{end-element()}, respectively, when parsing the input tree
from Fig.~\ref{fig:XmlDocumentTree}. It shows the content of the three
stacks after the body of the corresponding function has been executed,
where \textsf{is} denotes the index stack, \textsf{hs} denotes the
hierarchy stack and \textsf{ns} denotes the name stack. Regarding
Fig.~\ref{fig:constructionOfBinaryTree}, the element on top of the
stack is always the upper element in the depiction of the
corresponding stack. If there has not been assigned a label to a node,
\ie, the labeling function $\lambda$ has not been updated accordingly
yet, the node is depicted in brackets.
\end{example}
\begin{figure}[p]
	\small
	\centering
	\parbox{16cm}{
	\begin{multicols}{2}
	\begin{enumerate}[(1)]
		\item Function call \texttt{start-element(books)}
			
			\begin{raggedleft}
				\parbox{3cm}{
					\begin{center}
						\begin{tabular}{ccc}
							$0$&$\varepsilon$&\\
							\midrule
							\textsf{is}&\textsf{hs}&\textsf{ns}\\\bottomrule
						\end{tabular}
					\end{center}
				}
				\hspace{1cm}
				\parbox{1.5cm}{
					\begin{raggedleft}
						\begin{tikzpicture}
							[bend angle=25, node distance=0.25cm,text height=1.5ex,
							place/.style={},
							pre/.style={<-,>=stealth'},
							post/.style={->,>=stealth'}]
							\node[place] (books) {$\mathsf{books}^{10}$};
						\end{tikzpicture}
					\end{raggedleft}
				}
			\end{raggedleft}
		\item Function call \texttt{start-element(book)}
		
			\begin{raggedleft}
				\parbox{3cm}{
					\begin{center}
						\begin{tabular}{ccc}
							$0$&$1$&\\
							$1$&$\varepsilon$&$\mathsf{book}$\\
							\midrule
							\textsf{is}&\textsf{hs}&\textsf{ns}\\\bottomrule
						\end{tabular}
					\end{center}
				}
				\hspace{1cm}
				\parbox{1.5cm}{
					\begin{raggedleft}
						\begin{tikzpicture}
							[bend angle=25, node distance=0.25cm,text height=1.5ex,
							place/.style={},
							pre/.style={<-,>=stealth'},
							post/.style={->,>=stealth'}]
							\node[place] (books) {$\mathsf{books}^{10}$};
							
							\node[place] (book1) [below=of books] {$(1)$}
								edge [pre] (books);
						\end{tikzpicture}
					\end{raggedleft}
				}
			\end{raggedleft}
		\item Function call \texttt{start-element(author)}
		
			\begin{raggedleft}
				\parbox{3cm}{
					\begin{center}
						\begin{tabular}{ccc}
							$0$&$11$&\\
							$1$&$1$&$\mathsf{author}$\\
							$1$&$\varepsilon$&$\mathsf{book}$\\
							\midrule
							\textsf{is}&\textsf{hs}&\textsf{ns}\\\bottomrule
						\end{tabular}
					\end{center}
				}
				\hspace{1cm}
				\parbox{1.5cm}{
					\begin{raggedleft}
						\begin{tikzpicture}
							[bend angle=25, node distance=0.25cm,text height=1.5ex,
							place/.style={},
							pre/.style={<-,>=stealth'},
							post/.style={->,>=stealth'}]
							\node[place] (books) {$\mathsf{books}^{10}$};
							
							\node[place] (book1) [below=of books] {$(1)$}
								edge [pre] (books);
							\node[place] (author1) [below=of book1] {$(11)$}
								edge [pre] (book1);
						\end{tikzpicture}
					\end{raggedleft}
				}
			\end{raggedleft}
		\item Function call \texttt{end-element()}
		
			\begin{raggedleft}
				\parbox{3cm}{
					\begin{center}
						\begin{tabular}{ccc}
							&$11$&\\
							$1$&$1$&$\mathsf{author}$\\
							$1$&$\varepsilon$&$\mathsf{book}$\\
							\midrule
							\textsf{is}&\textsf{hs}&\textsf{ns}\\\bottomrule
						\end{tabular}
					\end{center}
				}
				\hspace{1cm}
				\parbox{1.5cm}{
					\begin{raggedleft}
						\begin{tikzpicture}
							[bend angle=25, node distance=0.25cm,text height=1.5ex,
							place/.style={},
							pre/.style={<-,>=stealth'},
							post/.style={->,>=stealth'}]
							\node[place] (books) {$\mathsf{books}^{10}$};
							
							\node[place] (book1) [below=of books] {$(1)$}
								edge [pre] (books);
							\node[place] (author1) [below=of book1] {$(11)$}
								edge [pre] (book1);
						\end{tikzpicture}
					\end{raggedleft}
				}
			\end{raggedleft}
		\item Function call \texttt{start-element(title)}
		
			\begin{raggedleft}
				\parbox{3cm}{
					\begin{center}
						\begin{tabular}{ccc}
							&$112$&\\
							$0$&$11$&$\mathsf{title}$\\
							$2$&$1$&$\mathsf{author}$\\
							$1$&$\varepsilon$&$\mathsf{book}$\\
							\midrule
							\textsf{is}&\textsf{hs}&\textsf{ns}\\\bottomrule
						\end{tabular}
					\end{center}
				}
				\hspace{1cm}
				\parbox{1.5cm}{
					\begin{raggedleft}
						\begin{tikzpicture}
							[bend angle=25, node distance=0.25cm,text height=1.5ex,
							place/.style={},
							pre/.style={<-,>=stealth'},
							post/.style={->,>=stealth'}]
							\node[place] (books) {$\mathsf{books}^{10}$};
							
							\node[place] (book1) [below=of books] {$(1)$}
								edge [pre] (books);
							\node[place] (author1) [below=of book1] {$(11)$}
								edge [pre] (book1);
							\node[place] (title1) [below=of author1] {$(112)$}
								edge [pre] (author1);
						\end{tikzpicture}
					\end{raggedleft}
				}
			\end{raggedleft}
		\item Function call \texttt{end-element()}
		
			\begin{raggedleft}
				\parbox{3cm}{
					\begin{center}
						\begin{tabular}{ccc}
							&$112$&\\
							&$11$&$\mathsf{title}$\\
							$2$&$1$&$\mathsf{author}$\\
							$1$&$\varepsilon$&$\mathsf{book}$\\
							\midrule
							\textsf{is}&\textsf{hs}&\textsf{ns}\\\bottomrule
						\end{tabular}
					\end{center}
				}
				\hspace{1cm}
				\parbox{1.5cm}{
					\begin{raggedleft}
						\begin{tikzpicture}
							[bend angle=25, node distance=0.25cm,text height=1.5ex,
							place/.style={},
							pre/.style={<-,>=stealth'},
							post/.style={->,>=stealth'}]
							\node[place] (books) {$\mathsf{books}^{10}$};
							
							\node[place] (book1) [below=of books] {$(1)$}
								edge [pre] (books);
							\node[place] (author1) [below=of book1] {$(11)$}
								edge [pre] (book1);
							\node[place] (title1) [below=of author1] {$(112)$}
								edge [pre] (author1);
						\end{tikzpicture}
					\end{raggedleft}
				}
			\end{raggedleft}
			\needspace{4\baselineskip}
		\item Function call \texttt{start-element(isbn)}
		
			\begin{raggedleft}
				\parbox{3cm}{
					\begin{center}
						\begin{tabular}{ccc}
								&$1122$			&\\
								&$112$			&$\mathsf{isbn}$\\
							$0$	&$11$			&$\mathsf{title}$\\
							$3$	&$1$				&$\mathsf{author}$\\
							$1$	&$\varepsilon$	&$\mathsf{book}$\\
							\midrule
							\textsf{is}&\textsf{hs}&\textsf{ns}\\\bottomrule
						\end{tabular}
					\end{center}
				}
				\hspace{1cm}
				\parbox{1.5cm}{
					\begin{raggedleft}
						\begin{tikzpicture}
							[bend angle=25, node distance=0.25cm,text height=1.5ex,
							place/.style={},
							pre/.style={<-,>=stealth'},
							post/.style={->,>=stealth'}]
							\node[place] (books) {$\mathsf{books}^{10}$};
							
							\node[place] (book1) [below=of books] {$(1)$}
								edge [pre] (books);
							\node[place] (author1) [below=of book1] {$(11)$}
								edge [pre] (book1);
							\node[place] (title1) [below=of author1] {$(112)$}
								edge [pre] (author1);
							\node[place] (isbn1) [below=of title1] {$(1122)$}
								edge [pre] (title1);
						\end{tikzpicture}
					\end{raggedleft}
				}
			\end{raggedleft}
		\item Function call \texttt{end-element()}
		
			\begin{raggedleft}
				\parbox{3cm}{
					\begin{center}
						\begin{tabular}{ccc}
								&$1122$			&\\
								&$112$			&$\mathsf{isbn}$\\
								&$11$			&$\mathsf{title}$\\

							$3$	&$1$				&$\mathsf{author}$\\
							$1$	&$\varepsilon$	&$\mathsf{book}$\\
							\midrule
							\textsf{is}&\textsf{hs}&\textsf{ns}\\\bottomrule
						\end{tabular}
					\end{center}
				}
				\hspace{1cm}
				\parbox{1.5cm}{
					\begin{raggedleft}
						\begin{tikzpicture}
							[bend angle=25, node distance=0.25cm,text height=1.5ex,
							place/.style={},
							pre/.style={<-,>=stealth'},
							post/.style={->,>=stealth'}]
							\node[place] (books) {$\mathsf{books}^{10}$};
							
							\node[place] (book1) [below=of books] {$(1)$}
								edge [pre] (books);
							\node[place] (author1) [below=of book1] {$(11)$}
								edge [pre] (book1);
							\node[place] (title1) [below=of author1] {$(112)$}
								edge [pre] (author1);
							\node[place] (isbn1) [below=of title1] {$(1122)$}
								edge [pre] (title1);
						\end{tikzpicture}
					\end{raggedleft}
				}
			\end{raggedleft}
		\item Function call \texttt{end-element()}
		
			\begin{raggedleft}
				\parbox{3cm}{
					\begin{center}
						\begin{tabular}{ccc}
								&$1$				&\\
							$1$	&$\varepsilon$	&$\mathsf{book}$\\
							\midrule
							\textsf{is}&\textsf{hs}&\textsf{ns}\\\bottomrule
						\end{tabular}
					\end{center}
				}
				\hspace{1cm}
				\parbox{1.5cm}{
					\begin{raggedleft}
						\begin{tikzpicture}
							[bend angle=25, node distance=0.25cm,text height=1.5ex,
							place/.style={},
							pre/.style={<-,>=stealth'},
							post/.style={->,>=stealth'}]
							\node[place] (books) {$\mathsf{books}^{10}$};
							
							\node[place] (book1) [below=of books] {$(1)$}
								edge [pre] (books);
							\node[place] (author1) [below=of book1] {$\mathsf{author}^{01}$}
								edge [pre] (book1);
							\node[place] (title1) [below=of author1] {$\mathsf{title}^{01}$}
								edge [pre] (author1);
							\node[place] (isbn1) [below=of title1] {$\mathsf{isbn}^{00}$}
								edge [pre] (title1);
						\end{tikzpicture}
					\end{raggedleft}
				}
			\end{raggedleft}
		\item Function call \texttt{start-element(book)}
		
			\begin{raggedleft}
				\parbox{3cm}{
					\begin{center}
						\begin{tabular}{ccc}
								&$12$			&\\
							$0$	&$1$				&$\mathsf{book}$\\
							$2$	&$\varepsilon$	&$\mathsf{book}$\\
							\midrule
							\textsf{is}&\textsf{hs}&\textsf{ns}\\\bottomrule
						\end{tabular}
					\end{center}
				}
				\hspace{1cm}
				\parbox{1.5cm}{
					\begin{raggedleft}
						\begin{tikzpicture}
							[bend angle=25, node distance=0.25cm,text height=1.5ex,
							place/.style={},
							pre/.style={<-,>=stealth'},
							post/.style={->,>=stealth'}]
							\node[place] (books) {$\mathsf{books}^{10}$};
							
							\node[place] (book1) [below=of books] {$(1)$}
								edge [pre] (books);
							\node[place] (author1) [below=of book1] {$\mathsf{author}^{01}$}
								edge [pre] (book1);
							\node[place] (title1) [below=of author1] {$\mathsf{title}^{01}$}
								edge [pre] (author1);
							\node[place] (isbn1) [below=of title1] {$\mathsf{isbn}^{00}$}
								edge [pre] (title1);
								
							\node[place] (book2) [right=of book1] {$(12)$}
								edge [pre] (book1);
						\end{tikzpicture}
					\end{raggedleft}
				}
			\end{raggedleft}
		\item Function call \texttt{start-element(book)}
		
			\begin{raggedleft}
				\parbox{3cm}{
					\begin{center}
						\begin{tabular}{ccc}
								&$121$			&\\
							$0$	&$12$			&$\mathsf{author}$\\
							$1$	&$1$				&$\mathsf{book}$\\
							$2$	&$\varepsilon$	&$\mathsf{book}$\\
							\midrule
							\textsf{is}&\textsf{hs}&\textsf{ns}\\\bottomrule
						\end{tabular}
					\end{center}
				}
				\hspace{1cm}
				\parbox{1.5cm}{
					\begin{raggedleft}
						\begin{tikzpicture}
							[bend angle=25, node distance=0.25cm,text height=1.5ex,
							place/.style={},
							pre/.style={<-,>=stealth'},
							post/.style={->,>=stealth'}]
							\node[place] (books) {$\mathsf{books}^{10}$};
							
							\node[place] (book1) [below=of books] {$(1)$}
								edge [pre] (books);
							\node[place] (author1) [below=of book1] {$\mathsf{author}^{01}$}
								edge [pre] (book1);
							\node[place] (title1) [below=of author1] {$\mathsf{title}^{01}$}
								edge [pre] (author1);
							\node[place] (isbn1) [below=of title1] {$\mathsf{isbn}^{00}$}
								edge [pre] (title1);
								
							\node[place] (book2) [right=of book1] {$(12)$}
								edge [pre] (book1);
							\node[place] (author2) [below=of book2] {$(121)$}
								edge [pre] (book2);
						\end{tikzpicture}
					\end{raggedleft}
				}
			\end{raggedleft}
	\end{enumerate}
	\end{multicols}
	}
	\caption{Content of the stacks after each call of the \texttt{start-element()} and \texttt{end-element()}, respectively, functions when parsing the tree from Fig.~\ref{fig:XmlDocumentTree}. In addition at each step their is a depiction of the binary tree which is constructed so far.}\label{fig:constructionOfBinaryTree}
\end{figure}
The binary representation of the input tree can be obtained in linear runtime since the function \verb|start-element| and the function \verb|end-element|, respectively, are each called only once for every node of the input tree. Furthermore, the body of the \verb|repeat| loop of the latter function is executed once for every input node (except for the root node).

\paragraph{Re-pair for Trees on Multiary Trees}\label{par:RemarksMultiaryModel}

Another way of modeling an XML document tree in a ranked way is the multiary tree model. In contrast to the binary tree model (which we described in Sect.~\ref{sec:binaryTreeModel} on page \pageref{sec:binaryTreeModel}), this model does not encode the input tree by a binary tree but it turns the input tree into a ranked tree by introducing a terminal symbol for each element type/number of children combination which occurs in the input tree. Let us assume that an element type occurs three times and that there are three different numbers of children attach to the corresponding elements. In the multiary tree model, there are introduced three different terminal symbols.

During our investigations we also evaluated a \trp version based on the multiary tree model. However, this modified version of our algorithm was outperformed by the original version in terms of compression ratio. This is due to the nature of typical XML documents. XML elements encountered in real-world XML documents often exhibit a long list of children elements. Therefore, compared to the binary tree model, a multiary tree model representation of an XML document leads to a higher number of different \tps occurring less often. This, in turn, reduces \trp's ability to compress the XML document tree by the same degree as it is possible for the binary case.
\begin{example}
Consider for example the XML document tree from Fig.~\ref{fig:XmlDocumentTree}. The element of type \verb|books| has five children elements of type $\mathsf{book}$, \ie, each of the five \tps
\[(\mathsf{books},1,\mathsf{book}),(\mathsf{books},2,\mathsf{book}),\ldots,(\mathsf{books},5,\mathsf{book})\]
occurs only once. None of these \tps is replaced by TreeRePair since a replacement is only reasonable if the corresponding \tp occurs at least twice. In contrast, the binary tree model  leads to two occurrences of the \tp $(\mathsf{book},1,\mathsf{book})$ which can be replaced by a new nonterminal symbol in a run of \trp (\cf Fig.~\ref{fig:binaryRepresentation}).
\end{example}

\subsection{Representing the Input Tree in Memory}\label{sec:representingTreeInMemory}

In this section we show that the ranked input tree of our algorithm can be efficiently stored as a DAG in memory. This DAG representation can be made nearly transparent to the rest of the algorithm (\cf Sect.~\ref{sec:impactDag} on page \pageref{sec:impactDag}).\footnote{Note that the DAG representation can also be circumvented by using the \texttt{-no\_dag} switch. In this case the whole binary tree with all its possible redundancy is constructed in main memory.} Thus, by default, the tree constructor of our prototype does not only directly transform the XML document tree into a ranked representation but also infers the corresponding minimal $0$-bounded SLCF tree grammar $\mathcal{G}=(N,P,S)$, \ie, the minimal DAG, of the latter on the fly.

In \cite{Buneman03path} it has been demonstrated that the representation of XML document trees based on the concept of sharing subtrees is highly efficient. Their experiments have shown that in several cases the size of the DAG was less than 10\% of the uncompressed XML document tree. Therefore, the sharing of common subtrees enables us to load large XML documents trees which would have otherwise exceeded the computation resources. In addition to that it avoids time consuming swapping and the repetitive re-computation of the same results concerning subtrees that are shared.

Now, let us elaborate on how one can infer the DAG of the ranked representation $t=(\mathsf{dom}_t,\lambda_t)\in T(\mathcal{F})$ of the XML document tree. The tree constructor must check for every node which is removed from the hierarchy stack in the \verb|end-element| function if the subtree rooted at this node can be shared. This can be accomplished by calling the function \verb|share-subtree| listed in 
Fig.~\ref{lst:methodShareTree}. To better understand this function, let us assume that we want to check if the subtree $t'\in T(\mathcal{F})$ rooted at a node $v\in\mathsf{dom}_t$ can be shared. If we already encountered an exact copy of $t'$ while reading the input tree, all subtrees of $t'$ must have been shared before. Thus, the tree $t'$ must be of depth 1 and all children nodes must be labeled by nonterminals of the DAG grammar $\mathcal{G}$. Therefore, it is only necessary to compare the labels of the root of $t'$ and its direct children with those of all subtrees encountered until now. This can be done in constant time with the help of a hash table.

Now, let us assume that we have processed an exact copy of $t'$ earlier, \ie, $t'$ can be shared. Thus, the condition in line 3 is evaluated to \verb|true| and the \verb|subtrees_ht| hash table contains $t'$. Hence, the \verb|else|-part beginning in line 6 is processed. If there already exists a nonterminal $B\in N$ with right-hand side $t'$ then we set $A:=B$. We can check this in $\mathcal{O}(1)$ time because with each entry of the hash table \verb|subtrees_ht| we can store a pointer to the corresponding production. Otherwise, \ie, if there exists no $(B\to t'')\in P$ with $t'=t''$, we introduce a new nonterminal $A\in\mathcal{N}_0\setminus N$ with right-hand side $t'$ and replace the first occurrence $u$ of the subtree $t'$ by $A$. There can be only one earlier occurrence of the subtree $t'$ since otherwise we would already have inserted a corresponding production. Furthermore, we can guarantee constant time access to $u$ because with each entry in the hash table \verb|subtrees_ht| we can store a pointer to the corresponding first occurrence. Finally, we add the subtree rooted at the node $\mathsf{parent}(u)$ to the hash table if all of its subtrees are shared. We do not need to insert the subtree rooted at the node $\mathsf{parent}(v)$ since we will process $\mathsf{parent}(v)$ in a later step (since we are traversing the input tree in postorder). In contrast, if $t'$ was not encountered until now, we add it to the hash table \verb|subtrees_ht| (line 5) in order to be able to share possible later occurrences of it.

Initially, \ie, after reading the input tree, all shared subtrees are of depth 1. In order to reduce the number of nonterminals of the DAG grammar (without increasing the number of total edges) all productions referenced only once are eliminated. All in all, the inferring of the DAG grammar needs linear time and can be conveniently combined with the step of transforming the input tree into a ranked tree.
\begin{figure}[tb]
\lstset{emph={FUNCTION,ENDFUNC,return,traverse,endtraverse,for,if,else,then,do,endif,endfor,repeat,endrepeat,times}, emphstyle=\bfseries, emph={[2]subtrees_ht}, emphstyle={[2]\slshape}, morecomment=[l]{//}}
\begin{lstlisting}{}
FUNCTION share-subtree((*@$v$@*))
	let (*@$t'$@*) be the subtree rooted at (*@$v$@*);
	if ((*@$\forall 1\leq i\leq\mathsf{rank}(\lambda_t(v)):\lambda_t(vi)\in\mathcal{N}_0$@*)) then
		if (subtrees_ht does not contain (*@$t'$@*)) then
			insert (*@$t'$@*) into subtrees_ht;
		else
			if ((*@$\exists B\in\mathcal{N}_0:(B\to t')\in P$@*)) then
				(*@$A:=B$@*);
			else
				choose nonterminal (*@$A\in\mathcal{N}_0\setminus N$@*);
				(*@$N:=N\cup\{A\}$@*); (*@$P:=P\cup\{(A\to t')\}$@*);
				let (*@$u$@*) be the node at which the first 
							occurrence of (*@$t'$@*) is rooted;
				replace subtree rooted at (*@$u$@*) by (*@$A$@*);
				
				(*@$w:=\mathsf{parent}(u)$@*);
				if ((*@$\forall 1\leq i\leq\mathsf{rank}(\lambda_t(w)):\lambda_t(wi)\in\mathcal{N}_0$@*)) then
					let (*@$t''$@*) be the subtree rooted at (*@$w$@*);
					insert (*@$t''$@*) into subtrees_ht;
				endif
			endif
			
			replace subtree rooted at (*@$v$@*) by (*@$A$@*);
		endif
	endif
ENDFUNC
\end{lstlisting}
\caption{The function \texttt{share-subtree} which checks for the subtree rooted at the node $v\in\mathsf{dom}_t$ if it can be shared. If this is the case then the sharing is performed.}\label{lst:methodShareTree}
\end{figure}

\subsection{Utilized Data Structures}

The data structures we use in our implementation are similar to those used in \cite{larsson2000off}. In order to be able to focus on the essentials, we do not pay attention to the fact that, internally, the input tree is represented by a DAG.

Let us assume that the binary input tree $t=(\mathsf{dom}_t,\lambda_t)\in T(\mathcal{F})$ has been generated by our implementation after reading a corresponding XML document tree. Hence, the tree $t$ is the ranked representation of the latter. In main memory, every node $v\in\mathsf{dom}_t$ is represented by an object exhibiting several pointers. These allow constant time access to the parent and all children of the node $v$ and to the possible next and previous occurrences of the \tp $\alpha=\big(\lambda_t(v),i,\lambda_t(vi)\big)$, where $i\in\{1,2,\ldots,\mathsf{rank}(\lambda_t(v))\}$. The pointers to the next and previous occurrences of $\alpha$ form a doubly linked list of all the occurrences in $\mathsf{occ}_t(\alpha)$. We call this type of list an \emph{occurrences list (of $\alpha$)} in the sequel.\footnote{During our investigations we also implemented a \trp version avoiding these doubly linked lists of occurrences. Instead, for every \tp, we used a hashed set storing pointers to all occurrences. However, this version had no benefits compared to the doubly linked list approach but lead to slightly longer runtimes. Considering the memory usage, in some cases it achieved better results while in others a substantial increase was noticed.} The specific order of the occurrences in an occurrences list is not relevant.

Every \tp is represented by a special object. It exhibits two pointers which reference the first and the last element of the corresponding occurrences list. Let us consider a \tp $\alpha\in\Pi$ with $|\mathsf{occ}_t(\alpha)|=m$, where $m<\lfloor\sqrt{n}\rfloor$ and $n=|t|$. Then the corresponding object exhibits two more pointers which point to the next and previous, respectively, \tp $\beta\in\Pi$ with $|\mathsf{occ}_t(\beta)|=m$. These pointers form a doubly linked list of all \tps occurring $m$ times. We denote this type of list the \emph{$m$-th \tp list}. In contrast, all \tps $\gamma\in\Pi$ with $|\mathsf{occ}_t(\gamma)|\geq\lfloor\sqrt{n}\rfloor$ are organized in one doubly linked list which is called the \emph{top \tp list}.

These doubly linked lists of \tps are again referenced by a \emph{\tp priority queue}. This queue consists of $\lfloor\sqrt{n}\rfloor$ entries. The $i$-th entry stores a pointer to the head of the $i$-th \tp list, where $1\leq i<\lfloor\sqrt{n}\rfloor$. The $\lfloor\sqrt{n}\rfloor$-th entry references the head of the top \tp list. Refer to Sect.~\ref{sec:complexityReplacementStep} on page \pageref{sec:complexityReplacementStep} for an explanation on why we designed the \tp lists and priority queue as described above. Lastly, there is a \emph{\tp hash table} storing pointers to all occurring \tps. It allows constant time access to all \tps and therefore constant time access to the first occurrence of each \tp.

Let us consider the following example to see how the utilized data structures work.
\begin{figure}[t]
	\centering
	\begin{tikzpicture}[->,>=stealth',semithick]
		\tikzstyle{level 1}=[level distance=0.8cm, sibling distance=1.5cm]
		\tikzstyle{level 2}=[level distance=0.8cm, sibling distance=0.75cm]
		\tikzstyle{level 3}=[level distance=0.8cm, sibling distance=0.75cm]
		\tikzstyle{fan}=[anchor=north,isosceles triangle, shape border uses incircle,inner sep=0.5pt,shape border rotate=90,draw]

		\node {$f$}
			child {node {$f$}
				child {node {$a$}}
				child {node {$f$} [child anchor=north]
					child {node {$a$}}
					child {node {$a$}}
				}
			}
			child {node {$f$}
				child {node {$a$}}
				child {node {$f$} [child anchor=north]
					child {node {$a$}}
					child {node {$a$}}
				}
			}
		;
	\end{tikzpicture}
	\caption{The tree $t\in T(\mathcal{F})$  modeled by the node objects from Fig.~\ref{fig:dataStructures}.}\label{fig:dataStructuresTree}
\end{figure}
\begin{example}
	Let us assume that the tree $t=(\mathsf{dom}_t,\lambda_t)\in T(\mathcal{F})$ shown in Fig.~\ref{fig:dataStructuresTree} has been generated by our implementation after reading a corresponding XML document tree. Then Fig.~\ref{fig:dataStructures} shows a simplified depiction of the data structures used to efficiently replace the \tps in the replacement step. All non-null pointers are represented by arrows starting in a filled circle and ending in an empty circle. A filled circle without an outgoing arrow denotes a null pointer.
	
	With respect to Fig.~\ref{fig:dataStructures}, there is a total of $11$ node objects representing tree nodes labeled by the two symbols $f\in\mathcal{F}_2$ and $a\in\mathcal{F}_0$. An instance of a tree node $v\in\mathsf{dom}_t$ is represented by a tabular box as it is shown in Fig.~\ref{fig:representationOfNode}. Unlike depicted, in our implementation a symbol is not directly stored within the node structure but for every unique symbol there is an object which is referenced by the corresponding nodes. The upper left empty circle of the box represents the memory address of the tree node instance. Thus, every arrow representing a pointer to the latter will end in this empty circle. 
\begin{figure}[t]
	\renewcommand{\subfigcapmargin}{-0.5cm}
	\hfill
	\subfigure[{A graphical representation of an object representing a tree node labeled by $f\in\mathcal{F}$.}]{
		\begin{tikzpicture}[semithick,scale=0.095,text height=1.5ex,text depth=.25ex]
			\begin{scope}
			\draw (0,0) -- +(35,0) -- +(35,16) -- +(0,16) -- +(0,0);
			\draw (10,0) -- (10,16);
			\draw[thin] (27,0) -- (27,16);
			\draw[thin] (31,0) -- (31,12);
			\foreach \y/\name in {12/parent,8/children,4/next,0/previous} {
				\draw (10,\y) -- (35,\y);
				\draw (10,\y-0.5) node[anchor=south west] {\footnotesize\textsf{\name}};
			}
			\draw (5,8) node[anchor=center] {\Large $f$};
		
			\begin{scope}[darkgray]
				\fill (31,14) circle (0.75);
				\fill (29,10) circle (0.75);
				\fill (33,10) circle (0.75);
				\fill (29,6) circle (0.75);
				\fill (33,6) circle (0.75);
				\fill (29,2) circle (0.75);
				\fill (33,2) circle (0.75);
			\end{scope}
		
			\filldraw[draw=black,fill=white] (0,16) circle (1);
			\end{scope}
		\end{tikzpicture}
		\label{fig:representationOfNode}
	}
	\hfill
	\subfigure[{A graphical representation of a \tp $(f,1,a)\in\Pi$.}]{
		\begin{tikzpicture}[semithick,scale=0.095,text height=1.5ex,text depth=.25ex]
			\draw (0,0) -- (26,0) -- (26,14) -- (0,14) -- (0,0);
			\draw (0,4) -- (26,4);
			\draw (0,8) -- (26,8);
			
			\filldraw[fill=white,draw=black] (0,14) circle (1);
	
			\draw (13,11) node {$(f,1,a)$};

			\draw (13,0) -- (13,8);
			\draw[thin] (9,0) -- (9,8);
			\draw[thin] (22,0) -- (22,8);
	
			\draw (0-0.5,4-0.5) node[anchor=south west] {\footnotesize\textsf{prev}};
			\draw (13-0.5,4-0.5) node[anchor=south west] {\footnotesize\textsf{next}};
			\draw (0-0.5,0-0.5) node[anchor=south west] {\footnotesize\textsf{first}};
			\draw (13-0.5,0-0.5) node[anchor=south west] {\footnotesize\textsf{last}};
			
			\begin{scope}[darkgray]
				\fill (11,6) circle (0.75);	
				\fill (24,6) circle (0.75);	
				\fill (11,2) circle (0.75);	
				\fill (24,2) circle (0.75);
			\end{scope}
		\end{tikzpicture}
		\label{fig:representationOfPair}
	}
	\hfill
\end{figure}
The filled circle in the first row of the tabular box represents the pointer to the possible parent node $\mathsf{parent}(v)$. The pointer to the $i$-th child $vi$ of the node $v$ is depicted by an arrow starting at the filled circle in the $i$-th column of the \verb|children| row, where $i\in\{1,2,\ldots,\mathsf{rank}(\lambda_t(v))\}$. Analogously, a pointer to a possible next (previous) occurrence of the \tp $\alpha=\big(\lambda_t(v),i,\lambda_t(vi)\big)$ is represented by a filled circle in the $i$-th column of the row labeled by \verb|next| (\verb|previous|, respectively), where $i\in\{1,2,\ldots,\mathsf{rank}(\lambda_t(v))\}$.
	
	Each \tp $(f,1,f)$, $(f,2,a)$, $(f,2,f)$ and $(f,1,a)$ is represented by a tabular box (see Fig.~\ref{fig:representationOfPair}). Again, unlike depicted, in our implementation a symbol is not directly stored within the \tp structure but the latter contains two pointers to the objects representing $a$ and $b$.
	The first and the last element of the occurrences list of the \tp $\alpha$ are referenced by the \verb|first| and \verb|last| pointers of the object representing the \tp $\alpha$. The pointers \verb|prev| (previous) and \verb|next| are part of the $|\mathsf{occ}_t(\alpha)|$-th \tp list if $|\mathsf{occ}_t(\alpha)|<\lfloor\sqrt{n}\rfloor$ and $n=|t|$. Otherwise they belong to the top \tp list.
	
	The \tp $(f,1,f)$ forms a trivial doubly linked list, namely, the 1st \tp list. The latter is referenced by the entry $1$ of the priority queue. The \tp $(f,1,a)$ forms the (trivial) top \tp list which is referenced by the entry $3$ of the priority queue. In contrast, the \tps $(f,2,a)$ and $(f,2,f)$ each occur twice and therefore point to each other with their \texttt{next} and \texttt{previous} pointers, respectively. The first element of the resulting 2nd \tp list is referenced by the entry $2$ of the priority queue. The \tp hash table stores the pointers to all four occurring \tps.
\end{example}

\begin{figure}[p]
	\vspace{-0.5cm}
	\hspace{-4.2cm}
	\begin{tikzpicture}[semithick,scale=0.095,text height=1.5ex,text depth=.25ex]
	
		\filldraw[fill=black!15,draw=black!20] (60,35) -- (110,35) -- (110,130) -- (60,130) -- (60,35);
		\draw (85,125) node[anchor = north] { \ltp Hash Table };
		
		\begin{scope}[xshift=74cm,yshift=40cm]
			\draw (0,0) -- (22,0) -- (22,75) -- (0,75) -- (0,0);
			\draw[thin] (6,0) -- (6,75);
			
			\foreach \y/\pair in { 14/{},20/{(f,2,f)}, 30/{}, 36/{(f,1,a)}, 42/{(f,1,f)}, 56/{}, 62/{(f,2,a)} } {
				\draw (0,\y) -- (22,\y);
				\draw (14,\y - 3) node[anchor=center] {$\pair$};
			}
			
			\draw[dotted] (14,2) -- (14,12);
			\draw[dotted] (14,22) -- (14,28);
			\draw[dotted] (14,44) -- (14,54);
			\draw[dotted] (14,64) -- (14,73);

			
			\draw (3,42 - 3) node[inner sep=0pt] (ht1pointer) {};
			\fill[darkgray] (3,42 - 3) circle (0.75);
			
			\draw (3,62 - 3) node[inner sep=0pt] (ht2pointer) {};
			\fill[darkgray] (3,62 - 3) circle (0.75);
			
			\draw (3,20 - 3) node[inner sep=0pt] (ht3pointer) {};
			\fill[darkgray] (3,20 - 3) circle (0.75);
			
			\draw (3,36 - 3) node[inner sep=0pt] (ht4pointer) {};
			\fill[darkgray] (3,36 - 3) circle (0.75);
		\end{scope}
	
		\filldraw[fill=black!15,draw=black!20] (-70,100) -- (55,100) -- (55,130) -- (-70,130) -- (-70,100);
		\draw (-65,125) node[anchor = north west] { \ltp Priority Queue };
		
		\begin{scope}[xshift=-62.5cm,yshift=106cm]
			\draw (0,0) -- (108,0) -- (108,10) -- (0,10) -- (0,0);
			\draw (0,4)[thin] -- (108,4);
			\foreach \x/\number/\name in {36/1/1,72/2/2,108/3/\geq3} {
				\draw (\x,0) -- (\x,10);
				\draw (\x-18,7) node[anchor=center] {\large $\name$};
				\draw (\x-18,2) node[inner sep=0pt] (queue\number pointer) {};
				\fill (\x-18,2) circle (0.75);
			}
		\end{scope}
	
		\filldraw[fill=black!15,draw=black!20] (-70,35) -- (55,35) -- (55,95) -- (-70,95) -- (-70,35);
		\draw (-65,40) node[anchor = south west] { Doubly Linked \ltps };
				
		\fill[fill=black!3] (-60,48) -- (-30,48) -- (-30,88) -- (-60,88) -- (-60,48);
		\fill[fill=black!3] (-23,48) -- (7,48) -- (7,88) -- (-23,88) -- (-23,48);
		\fill[fill=black!3] (14,48) -- (44,48) -- (44,88) -- (14,88) -- (14,48);
		
		\foreach \xshift/\yshift/\name/\pair in
		 {-58/72/pair1/{(f,1,f)},-21/72/pair2/{(f,2,a)},
		 -21/50/pair3/{(f,2,f)},16/72/pair4/{(f,1,a)}} {
			\begin{scope}[xshift=\xshift cm,yshift=\yshift cm]
				\draw (0,0) -- (26,0) -- (26,14) -- (0,14) -- (0,0);
				\draw (0,4) -- (26,4);
				\draw (0,8) -- (26,8);
				
				\draw (0,14) node[inner sep=2] (\name address) {};
				\filldraw[fill=white,draw=black] (0,14) circle (1);
		
				\draw (13,11) node {$\pair$};

				\draw (13,0) -- (13,8);
				\draw[thin] (9,0) -- (9,8);
				\draw[thin] (22,0) -- (22,8);
	
				\draw (0-0.5,4-0.5) node[anchor=south west] {\footnotesize\textsf{prev}};
				\draw (13-0.5,4-0.5) node[anchor=south west] {\footnotesize\textsf{next}};
				\draw (0-0.5,0-0.5) node[anchor=south west] {\footnotesize\textsf{first}};
				\draw (13-0.5,0-0.5) node[anchor=south west] {\footnotesize\textsf{last}};
			
				\begin{scope}[darkgray]
				\draw (11,6) node[inner sep=0pt] (\name prev) {};
				\fill (11,6) circle (0.75);	
				\draw (24,6) node[inner sep=0pt] (\name next) {};
				\fill (24,6) circle (0.75);	
				\draw (11,2) node[inner sep=0pt] (\name first) {};
				\fill (11,2) circle (0.75);	
				\draw (24,2) node[inner sep=0pt] (\name last) {};
				\fill (24,2) circle (0.75);
				\end{scope}
			\end{scope}
		}
		
		\begin{scope}[->,>=stealth',darkgray,shorten >=-4pt]
			\draw (ht1pointer) .. controls +(-20,30) and +(5,20) .. (pair1address);
			\draw (ht2pointer) .. controls +(-20,10) and +(5,10) .. (pair2address);
			\draw (ht3pointer) .. controls +(-20,10) and +(5,5) .. (pair3address);
			\draw (ht4pointer) .. controls +(-20,20) and +(5,10) .. (pair4address);
		\end{scope}
		
		\begin{scope}[->,>=stealth',darkgray,dashed]
			\draw (queue1pointer) .. controls +(10,-10) and +(-10,10) .. (pair1address);
			\draw (queue2pointer) .. controls +(10,-10) and +(-10,10) .. (pair2address);
			\draw (queue3pointer) .. controls +(10,-10) and +(-10,10) .. (pair4address);
		\end{scope}
		
		\begin{scope}[->,>=stealth',darkgray,densely dotted]
			\draw (pair3prev) .. controls +(-22,-12) and +(-8,5) .. (pair2address);
			\draw (pair2next) .. controls +(20,-15) and +(-15,10) .. (pair3address);
		\end{scope}

		\filldraw[fill=black!15,draw=black!20] (-70,30) -- (110,30) -- (110,-100) -- (-70,-100) -- (-70,30);
		
		\fill[fill=black!3] (-14,-23) -- (47,-23) -- (17,8) -- (-14,-23);
		\fill[fill=black!3] (-43,-54) -- (-3,-54) -- (-14,-23) -- (-43,-54);
		\fill[fill=black!3] (36,-54) -- (77,-54) -- (47,-23) -- (36,-54);
		\fill[fill=black!3] (-39,-85) -- (1,-85) -- (-3,-54) -- (-39,-85);
		\fill[fill=black!3] (46,-85) -- (86,-85) -- (77,-54) -- (46,-85);
		
		\draw (-65,25) node[anchor=north west] { Tree Nodes };
	
		\foreach \xshift/\yshift/\name/\symbol in {0/0/node1/f,-30/-31/node2/f,30/-31/node3/f,-20/-62/node5/f,60/-62/node7/f} {
			\begin{scope}[xshift=\xshift cm,yshift=\yshift cm]
				\draw (0,0) -- +(35,0) -- +(35,16) -- +(0,16) -- +(0,0);
				\draw (10,0) -- (10,16);
				\draw[thin] (27,0) -- (27,16);
				\draw[thin] (31,0) -- (31,12);
				\foreach \y/\name in {12/parent,8/children,4/next,0/previous} {
					\draw (10,\y) -- (35,\y);
					\draw (10,\y-0.5) node[anchor=south west] {\footnotesize\textsf{\name}};
				}
				\draw (5,8) node[anchor=center] {\Large $\symbol$};
			
				\begin{scope}[darkgray]
				\draw (31,14) node[inner sep=0pt] (\name parent) {};
				\fill (31,14) circle (0.75);
				\draw (29,10) node[inner sep=0pt] (\name child1) {};
				\fill (29,10) circle (0.75);
				\draw (33,10) node[inner sep=0pt] (\name child2) {};
				\fill (33,10) circle (0.75);
				\draw (29,6) node[inner sep=0pt] (\name next1) {};
				\fill (29,6) circle (0.75);
				\draw (33,6) node[inner sep=0pt] (\name next2) {};
				\fill (33,6) circle (0.75);
				\draw (29,2) node[inner sep=0pt] (\name prev1) {};
				\fill (29,2) circle (0.75);
				\draw (33,2) node[inner sep=0pt] (\name prev2) {};
				\fill (33,2) circle (0.75);
				\end{scope}
			
				\draw (0,16) node[inner sep=2] (\name address) {};
				\filldraw[fill=white,draw=black] (0,16) circle (1);
				
				\draw (0,0) node[inner sep=0pt] (\name southwest) {};
				\draw (35,0) node[inner sep=0pt] (\name southeast) {};
				\draw (35,16) node[inner sep=0pt] (\name northeast) {};
				\draw (0,16) node[inner sep=0pt] (\name northwest) {};
			\end{scope}
		}
		
		\foreach \xshift/\yshift/\name/\symbol in {-60/-62/node4/a,20/-62/node6/a,-55/-93/node8/a,-15/-93/node9/a,30/-93/node10/a,70/-93/node11/a} {
			\begin{scope}[xshift=\xshift cm,yshift=\yshift cm]
				\draw (0,0) -- +(35,0) -- +(35,16) -- +(0,16) -- +(0,0);
				\draw (10,0) -- (10,16);
				\draw[thin] (27,12) -- (27,16);
				\foreach \y/\name in {12/parent} {
					\draw (10,\y) -- (35,\y);
					\draw (10,\y-0.5) node[anchor=south west] {\footnotesize\textsf{\name}};
				}
				\draw (5,8) node[anchor=center] {\Large $\symbol$};
			
				\draw (31,14) node[inner sep=0pt] (\name parent) {};
				\fill (31,14) circle (0.75);
			
				\draw (0,16) node[inner sep=2] (\name address) {};
				\filldraw[fill=white,draw=black] (0,16) circle (1);
			\end{scope}
		}
		
		\begin{scope}[->,>=stealth',darkgray]
			\draw (node2parent) .. controls +(-5,15) and +(-30,-30) .. (node1address);
			\draw (node3parent) .. controls +(-10,20) and +(-40,-40) .. (node1address);
			\draw (node4parent) .. controls +(-5,15) and +(-30,-30) .. (node2address);
			\draw (node5parent) .. controls +(-10,20) and +(-40,-40) .. (node2address);
			\draw (node6parent) .. controls +(-5,15) and +(-30,-30) .. (node3address);
			\draw (node7parent) .. controls +(-10,20) and +(-40,-40) .. (node3address);
			
			\begin{scope}[shorten >=-4pt]
				\draw (node8parent) .. controls +(7.5,10) and +(-7.5,-25) .. (node5address);
				\draw (node9parent) .. controls +(-5,15) and +(-11,-40) .. (node5address);
				\draw (node10parent) .. controls +(7.5,10) and +(-7.5,-25) .. (node7address);
				\draw (node11parent) .. controls +(-5,15) and +(-11,-40) .. (node7address);
			\end{scope}
		\end{scope}
		
		\begin{scope}[->,>=stealth',gray,densely dotted]
			\draw (node1child1) .. controls +(50,-30) and +(-35,20) .. (node2address);
			\draw (node1child2) .. controls +(38,-25) and +(-25,15) .. (node3address);
			\draw (node2child1) .. controls +(50,-30) and +(-35,20) .. (node4address);
			\draw (node2child2) .. controls +(38,-25) and +(-25,15) .. (node5address);
			\draw (node3child1) .. controls +(50,-30) and +(-35,20) .. (node6address);
			\draw (node3child2) .. controls +(38,-25) and +(-25,15) .. (node7address);
			
			\draw (node5child1) .. controls +(50,-30) and +(-35,20) .. (node8address);
			\draw (node5child2) .. controls +(38,-25) and +(-25,15) .. (node9address);
			\draw (node7child1) .. controls +(50,-30) and +(-35,20) .. (node10address);
			\draw (node7child2) .. controls +(38,-25) and +(-25,15) .. (node11address);
		\end{scope}
		
		\begin{scope}[->,>=stealth',darkgray,densely dashed]
			\draw (node5next2) .. controls +(10,-10) and +(-10,15) .. (node7address);
			\draw (node7prev2) .. controls +(10,-20) and +(-10,15) .. (node5address);
			
			\draw (node2next2) .. controls +(15,-15) and +(-15,15) .. (node3address);
			\draw (node3prev2) .. controls +(-50,-30) and +(50,30) .. (node2address);
			
			\draw (node2next1) .. controls +(38,-25) and +(-15,15) .. (node5address);
			\draw (node5prev1) .. controls +(-60,-25) and +(-15,15) .. (node2address);
			
			\draw (node5next1) .. controls +(38,-25) and +(-15,15) .. (node3address);
			\draw (node3prev1) .. controls +(-10,-15) and +(-35,20) .. (node5address);
			
			\draw (node3next1) .. controls +(38,-25) and +(-15,15) .. (node7address);
			\draw (node7prev1) .. controls +(40,15) and +(20,25) .. (node3address);
		\end{scope}
		
		\begin{scope}[->,>=stealth',darkgray]
			\draw (pair1first) .. controls +(25,-20) and +(-35,5) .. (node1address);
			\draw (pair1last) .. controls +(15,-20) and +(-35,5) .. (node1address);
			
			\draw (pair2first) .. controls +(90,-24) and +(-90,40) .. (node5address);
			\draw (pair2last) .. controls +(15,-5) and +(60,15) .. (node7address);
			
			\draw (pair3first) .. controls +(-15,-20) and +(20,5) .. (node2address);
			\draw (pair3last) .. controls +(-15,-20) and +(60,40) .. (node3address);
			
			\draw (pair4first) .. controls +(25,-20) and +(-35,5) .. (node2address);
			\draw (pair4last) .. controls +(-5,-25) and +(60,15) .. (node7address);
		\end{scope}
	\end{tikzpicture}
	\caption{A simplified depiction of a part of the data structures used by our implementation.}\label{fig:dataStructures}
\end{figure}

\subsection{Complexity of the \trp Algorithm}\label{sec:complexityReplacementStep}

\begin{theorem}
For any given input tree with $n$ edges, \trp produces in time
$\mathcal{O}(|t|)$ a $k$-bounded linear SLCF tree grammar
${\cal G}$, where $k\in\mathbb{N}$ is a constant, $t\in T(\mathcal{F})$ 
is the binary representation of the input tree, and
$\mathsf{val}({\cal G})=t$. 
\end{theorem}
It is straightforward to come up with a linear time implementation of the pruning step of the Re-pair for Trees algorithm (\cf Sect.~\ref{sec:pruningStep} on page \pageref{sec:pruningStep}). Therefore, we just want to investigate the complexity of the replacement step which was described in Sect.~\ref{sec:replacementStep} on page \pageref{sec:replacementStep}.

With every replacement of a \tp occurrence one edge of the input tree is absorbed. Therefore, a run of \trp can consist of at most $n-1$ iterations, where $n$ is the size of the input tree. Each replacement of an occurrence can be accomplished in $\mathcal{O}(1)$ time since at most $k$ children need to be reassigned --- in our implementation, the reassignment of a child node is just a matter of updating two pointers.\footnote{As already mentioned at the beginning of this section on page \pageref{ch:implementationDetails}: The maximal rank of a nonterminal of a grammar generated by \trp is $k\in\mathbb{N}$. The constant $k$ can be specified by a command line switch.} For every production which is introduced during a run of our algorithm it holds that the right-hand side $t$ is of size $|t|<2+k$, \ie, it can be constructed in constant time. 

However, to show that the replacement step can be performed in linear time two more aspects need to be considered. Imagine that we are in the $i$-th iteration of our algorithm (and $\mathcal{G}_{i-1}$ is the current grammar). Let $t\in T(\mathcal{F}\cup\mathcal{N})$ be the right-hand side of $\mathcal{G}_{i-1}$'s start production.
\begin{enumerate}[(1)]
	\item \emph{Updating the sets of non-overlapping occurrences}
	
		In every iteration of our algorithm we need to know the number of occurrences of each \tp. Only in that case we are able to determine the most frequent \tp. In addition, for replacing the \tp $\mathsf{max}_k(t)$, we need to know $\mathsf{occ}_t(\mathsf{max}_k(t))$. How can we compute the set $\mathsf{occ}_t(\alpha)$ for every \tp $\alpha\in\Pi$ without traversing the whole right-hand side of the current start production in each iteration?
	\item \emph{Retrieving the most frequent \tp} 
	
		Let us assume that there is an up to date set $\mathsf{occ}_t(\alpha)$ available for every $\alpha\in\Pi$ occurring in $t$ (in the form of occurrences lists). How do we determine the most frequent \tp in constant time?
\end{enumerate}
In the following we consider each of the above aspects in detail.

\subsubsection{Updating the Sets of Non-overlapping Occurrences}

Let the binary tree $t=(\mathsf{dom}_t,\lambda_t)\in T(\mathcal{F})$ be our input tree. At the beginning of the replacement step the set $\mathsf{occ}_t(\alpha)$ for every \tp $\alpha\in\Pi$ occurring in $t$ is initially constructed. This is done by parsing the tree $t$ in a similar way as it is done in the function \texttt{retrieve-occurrences} which is listed in Fig.~\ref{lst:functionRetrieveOccurrences}. However, during the traversal not only one \tp is considered but for every encountered \tp $\alpha\in\Pi$ the set $\mathsf{occ}_t(\alpha)$ is constructed. Fig.~\ref{lst:functionRetrieveAllOccs} shows a possible function which accomplishes this task.
\begin{figure}[tb]
	\lstset{emph={FUNCTION,ENDFUNC,return,for,if,else,then,do,endif,endfor,while,endwhile}, emphstyle=\bfseries,morecomment=[l]{//},commentstyle=\color{gray}}
	\begin{lstlisting}
FUNCTION retrieve-all-occs((*@$t$@*))
	(*@$v:=\varepsilon$@*);
	while (true) do
		(*@$v:=\;$@*)next_in_postorder((*@$t$@*), (*@$v$@*));
		if ((*@$v\neq\varepsilon$@*)) then
			(*@$\alpha:=(\lambda_t(\mathsf{parent}(v)),\mathsf{index}(v),\lambda_t(v))$@*);
			if ((*@$v\notin\mathsf{occ}_t(\alpha)$@*)) then
				(*@$\mathsf{occ}_t(\alpha):=\mathsf{occ}_t(\alpha)\cup\{\mathsf{parent}(v)\}$@*)
			endif
		else
			return;
		endif
	endwhile
ENDFUNC
	\end{lstlisting}
	\caption{The function \texttt{retrieve-all-occs} which is used to construct the set $\mathsf{occ}_t(\alpha)$ for every \tp $\alpha\in\Pi$ occurring in the tree $t\in T(\mathcal{F}\cup\mathcal{N})$. It uses the function \texttt{next-in-postorder} listed in Fig.~\ref{lst:traversalAlgorithm}.}\label{lst:functionRetrieveAllOccs}
\end{figure}

Therefore, in the first iteration of our computation we have up to date sets of non-overlapping occurrences at hand. However, we cannot afford to redo this traversal in every subsequent iteration. In this case we would not be able to achieve a linear runtime of our algorithm.

Fortunately, there is another way of keeping track of the sets of non-overlapping occurrences. It relies on the fact that every replacement of an \tp occurrence $v$ only involves those occurrences in the neighborhood of $v$ which overlap with $v$.
\begin{example}
	Let us consider the tree $t'=(\mathsf{dom}_{t'},\lambda_{t'})\in T(\mathcal{F})$ which is depicted in Fig.~\ref{fig:absorbedOccurrences}. The occurrences which would be absorbed by the replacement of the occurrence $2\in\mathsf{dom}_{t'}$ of the \tp $(f,1,g)$ are highlighted.
\end{example}
For every \tp $\alpha\in\Pi$ we set $\mathsf{occ}_t'(\alpha):=\mathsf{occ}_t(\alpha)$ and base all upcoming computations on the set $\mathsf{occ}_t'(\alpha)$. In particular we use them to determine the most frequent \tp in each iteration. 

Let us consider the $i$-th iteration of a run $\mathcal{G}_0,\mathcal{G}_1,\ldots,\mathcal{G}_h$ of Re-pair for Trees on the input tree $t\in T(\mathcal{F})$, where $h\in\mathbb{N}$ and $i\in\{1,2,\ldots,h\}$. Then $\mathcal{G}_{i-1}=(N_{i-1},P_{i-1},S_{i-1})$ is the current grammar. Let $t_{i-1}\in T(\mathcal{F})$ be the right-hand side of $S_{i-1}$. Let us assume that an up to date set $\mathsf{occ}_{t_{i-1}}'(\beta)$ for every $\beta\in\Pi$ which is occurring in $t_{i-1}$ is at hand. Further, let us assume that $\mathsf{max}(t_{i-1})=(a,j,b)=:\alpha$ and let $v\in\mathsf{occ}_{t_{i-1}}'(\alpha)$.

\begin{figure}[tb]
	\centering\small
	\pgfdeclarelayer{background layer}
	\pgfsetlayers{background layer,main}
	\begin{tikzpicture}[semithick,level distance=1.3cm,sibling distance=0.75cm,text height=1.5ex,scale=0.9]
		\tikzstyle{fan}=[anchor=north,isosceles triangle, shape border uses incircle,inner sep=0.5pt,shape border rotate=90,draw]
		
		\begin{scope}[->,>=stealth']
		\node (epsilon) {$g$}
			child[sibling distance=4cm] {node (1) {$g$}
				child[sibling distance=1cm] {node (11) {$a$}}
				child[sibling distance=1cm] {node (12) {$b$}}
				child[sibling distance=1cm] {node (12) {$c$}}
				child[sibling distance=1cm] {node (12) {$d$}}
			}
			child {node (2) {$f$} 
				child {node[outer sep=3pt] (21) {$g$}
					child[sibling distance=1.5cm] {node[outer sep=2pt] (211) {$h$}
						child[sibling distance=1cm] {node (2111) {$a$}}
						child[sibling distance=1cm] {node (2112) {$b$} }
						child[sibling distance=1cm] {node (2113) {$c$} }
					}
					child[sibling distance=1.5cm] {node (212) {$a$}}
					child[sibling distance=1cm] {node[outer sep=2pt] (213) {$h$}
						child {node (2131) {$a$}}
						child {node (2132) {$b$} }
						child {node (2133) {$c$} }
					}
					child[missing] {node {$v4$}}
					child {node (214) {$b$} }
				}
				child[sibling distance=1cm] {node (22) {$a$}}
			}
			child[sibling distance=3cm] {node (3) {$f$}
				child[sibling distance=1cm] {node (31) {$a$}}
				child[sibling distance=1cm] {node[outer sep=2pt] (32) {$h$}
					child {node (321) {$a$}}
					child {node (322) {$b$} }
					child {node (323) {$c$} }
				}
			}
		;
		\end{scope}
		
		\begin{pgfonlayer}{background layer}
			\begin{scope}[rounded corners,fill=black!10,draw=gray,dashed]
			\filldraw[fill=black!10!white] ([xshift=6pt,yshift=7pt] 2.center) -- ([xshift=-3pt,yshift=7pt] 2.center) -- ([xshift=-5pt,yshift=5pt] 21.center) -- ([xshift=-10pt,yshift=2pt] 211.center) -- ([xshift=-10pt,yshift=-5pt] 211.center) -- ([xshift=-9pt,yshift=-5pt] 213.center) -- ([xshift=-7pt,yshift=0pt] 2131.center) -- ([xshift=-7pt,yshift=-5pt] 2131.center) -- ([xshift=7pt,yshift=-5pt] 2133.center) -- ([xshift=7pt,yshift=0pt] 2133.center) -- ([xshift=9pt,yshift=-5pt] 213.center) -- ([xshift=10pt,yshift=-5pt] 214.center) -- ([xshift=10pt,yshift=2pt] 214.center) -- ([xshift=6pt,yshift=5pt] 21.center) -- cycle;
			\end{scope}
		\end{pgfonlayer}		
	\end{tikzpicture}
	\caption{The tree $t'\in T(\mathcal{F})$. All occurrences which would be absorbed by the replacement are highlighted.}\label{fig:absorbedOccurrences}
\end{figure}
\begin{figure}[tb]
	\lstset{emph={FUNCTION,ENDFUNC,return,traverse,endtraverse,for,if,else,then,do,endif,endfor,repeat,endrepeat,times}, emphstyle=\bfseries, emph={[2]subtrees_ht}, emphstyle={[2]\slshape}, morecomment=[l]{//}}
	\begin{lstlisting}
FUNCTION remove-absorbed-occs((*@$t,v,j$@*))
	if ((*@$v\neq\varepsilon$@*)) then
		(*@$\alpha:=\big(\lambda_{t}(\mathsf{parent}(v)),\mathsf{index}(v),\lambda_{t}(v)\big)$@*);
		(*@$\mathsf{occ}_{t}'(\alpha):=\mathsf{occ}_{t}'(\alpha)\setminus\{\mathsf{parent}(v)\}$@*);
	endif
	
	for ((*@$l\in\{1,2,\ldots,\mathsf{rank}(\lambda_{t}(v))\}$@*)) do
		(*@$\alpha:=\big(\lambda_{t}(v),l,\lambda_{t}(vl)\big)$@*);
		(*@$\mathsf{occ}_{t}'(\alpha):=\mathsf{occ}_{t}'(\alpha)\setminus\{v\}$@*);
	endfor
	
	for ((*@$l\in\{1,2,\ldots,\mathsf{rank}(\lambda_{t}(vj))\}$@*)) do
		(*@$\alpha:=\big(\lambda_{t}(vj),l,\lambda_{t}(vjl)\big)$@*);
		(*@$\mathsf{occ}_{t}'(\alpha):=\mathsf{occ}_{t}'(\alpha)\setminus\{vj\}$@*);
	endfor
ENDFUNC	
	\end{lstlisting}
	\caption{Listing of the function \texttt{remove-absorbed-occs} which removes all absorbed occurrences from the $\mathsf{occ}_{t}'$ sets.}\label{lst:remove-absorbed-occs}
\end{figure}
\begin{figure}[tb]
\lstset{emph={FUNCTION,ENDFUNC,return,traverse,endtraverse,for,if,else,then,do,endif,endfor,repeat,endrepeat,times}, emphstyle=\bfseries, emph={[2]subtrees_ht}, emphstyle={[2]\slshape}, morecomment=[l]{//}}
	\begin{lstlisting}
FUNCTION add-new-occs((*@$t,u$@*))
	if ((*@$u\neq\varepsilon$@*)) then
		(*@$\alpha:=\big(\lambda_{t}(\mathsf{parent}(u)),\mathsf{index}(u),\lambda_{t}(u)\big)$@*);
		(*@$\mathsf{occ}_{t}(\alpha)':=\mathsf{occ}_{t}'(\alpha)\cup\{\mathsf{parent}(u)\}$@*);
	endif
	
	for ((*@$l\in\{1,2,\ldots,\mathsf{rank}(\lambda_t(u))\}$@*)) do
		(*@$\alpha:=\big(\lambda_t(u),l,\lambda_{t}(ul)\big)$@*);
		(*@$\mathsf{occ}_{t}'(\alpha):=\mathsf{occ}_{t}'(\alpha)\cup\{u\}$@*);
	endfor
ENDFUNC	
	\end{lstlisting}
	\caption{Listing of the function \texttt{add-new-occs} which adds all newly created occurrences to the $\mathsf{occ}_{t}'$ sets.}\label{lst:add-new-occs}
\end{figure}
Before the actual replacement of the occurrence $v$ we make use of the function listed in Fig.~\ref{lst:remove-absorbed-occs}. The function call \texttt{remove-absorbed-occs(}$t_{i-1},v,j$\texttt{)} removes all occurrences which will be absorbed by the upcoming replacement from the sets $\mathsf{occ}_{t_{i-1}}'$. After the replacement of $v$ by a new node $u$ with $\lambda_{t_i}(u)=A_i\in\mathcal{N}$ we call the function \verb|add-new-occs| (which is listed Fig.~\ref{lst:add-new-occs}) and pass the tree $t_{i-1}$ and the node $u$. The function \verb|add-new-occs| adds all new occurrences which arose by the introduction of $u$ to the sets of non-overlapping occurrences. Finally, after all occurrences from $\mathsf{occ}_{t_{i-1}}'(\alpha)$ have been replaced, we set $\mathsf{occ}_{t_i}'(\beta):=\mathsf{occ}_{t_{i-1}}'(\beta)$ for all $\beta\in\Pi$ occurring in $t_i$.

Let $\alpha\in\Pi$ be a \tp occurring in $t_i$. The above computed set $\mathsf{occ}_{t_i}'(\alpha)$ may not be equal to the actual set $\mathsf{occ}_{t_i}(\alpha)$ as it would be constructed by a complete postorder traversal of $t_i$ using the function \texttt{retrieve-occurrences} from Fig.~\ref{lst:functionRetrieveOccurrences}. 
\begin{example}
Consider, for instance, the tree $t''\in T(\mathcal{F})$ depicted in Fig.~\ref{fig:treeComputationOccSets}. Let $\alpha=(f,2,f)$. In the first iteration of our algorithm, we would obtain $\mathsf{occ}_{t''}'(\alpha):=\mathsf{occ}_{t''}(\alpha)=\{2\}$. Now, let us assume that we replace the \tp $(f,1,c)$ (we could easily enlarge $t''$ such that $(f,1,c)$ is the most frequent \tp and still show the same). After performing this replacement and especially after calling the functions \verb|remove-absorbed-occs| and \texttt{add-new-occs} we would have $\mathsf{occ}_{t''}'(\alpha)=\emptyset$. However, a postorder traversal of the updated tree $t''$ would result in $\mathsf{occ}_{t''}(\alpha)=\{\varepsilon\}$. 
\end{example}
\begin{figure}[t]
	\centering
	\begin{tikzpicture}[->,>=stealth',semithick,level distance=0.8cm, sibling distance=1.5cm]
		\tikzstyle{fan}=[anchor=north,isosceles triangle, shape border uses incircle,inner sep=0.5pt,shape border rotate=90,draw]
			
		\node {$f$} 
			child {node {$a$}}
			child {node {$f$}
				child {node {$b$}}
				child {node {$f$}
					child {node {$c$}}
					child {node {$d$}}
				}
			}
		;
	\end{tikzpicture}
	\caption{Tree $t''\in T(\mathcal{F})$ consisting of nodes labeled by the terminal symbols $a,b,c,d,f\in\mathcal{F}$. We have to deal with three overlapping occurrences of the \tp $(f,2,f)$.}\label{fig:treeComputationOccSets}
\end{figure}
Updating the sets of non-overlapping occurrences
takes constant time per occurrence replacement. At most $2k+1$
occurrences need to be removed by the function 
\texttt{remove-absorbed-occs} and at most $k+1$
occurrences need to be added by the function \texttt{add-new-occs}. An
occurrence $v$ of a \tp $\alpha$ can be removed from the occurrences
list of $\alpha$ in constant time by setting the \verb|next| and
\verb|previous| pointers of the corresponding node object to null. In
addition, if $v$ is the first (last) occurrence in the occurrence list
of $\alpha$ the \verb|first| (\verb|last|) pointer of the object
representing the \tp $\alpha$ needs to be updated. This can also be
accomplished in constant time by using the \tp hash table. 
Analogously, an occurrence can be added to an occurrences list in $\mathcal{O}(1)$ time.

\subsubsection{Retrieving the Most Frequent \ltp}

We now investigate the time needed to obtain the most frequent \tp in an iteration of our algorithm. First of all, let us state the following fact: Let $m\in\mathbb{N}\cup\{\infty\}$ and let $\mathcal{G}_0,\mathcal{G}_1,\ldots,\mathcal{G}_n$ be a run of Re-pair for Trees, where $n\in\mathbb{N}_{>0}$, $\mathcal{G}_i=(N_i,P_i,S_i)$ and $(S_i\to t_i)\in P_i$ for every $i\in\{0,1,\ldots,n\}$. Then
\[|\mathsf{occ}_{t_i}(\mathsf{max}_m(t_i))|\geq|\mathsf{occ}_{t_{i+1}}(\mathsf{max}_m(t_{i+1}))|\]
holds for every $i\in\{0,1,\ldots,n-1\}$.\footnote{Intuitively, we define $|\mathsf{occ}_{t_n}(\mathsf{max}_m(t_n))|=0$ if $\mathsf{max}_m(t_n)=\mathsf{undefined}$.} For every \tp $\alpha\in\Pi$ occurring in $t_i$ it holds that $|\mathsf{occ}_{t_i}(\alpha)|\geq|\mathsf{occ}_{t_{i+1}}(\alpha)|$ and for every \tp $\beta\in\Pi$ which was introduced in $\mathcal{G}_{i+1}$ it holds that $|\mathsf{occ}_{t_{i+1}}(\beta)|\leq|\mathsf{occ}_{t_i}(\mathsf{max}_m(t_{i}))|$, where $i\in\{0,1,\ldots,n-1\}$.

It is easy to see that, if the top \tp list is empty, we can obtain the most frequent \tp in constant time. We just need to walk down the remaining $\lfloor\sqrt{n}\rfloor-1$ \tp lists and choose the first element of the first non-empty list. In every iteration, after we have determined the most frequent \tp, we remember the first non-empty \tp list in order to save ourself the needless and time-consuming rechecking of the empty \tp lists.

Now, let us assume that the top \tp list, \ie, the doubly linked list of all \tps occurring at least $\lfloor\sqrt{n}\rfloor$ times, is not empty. We need to scan all elements in it since the \tps contained are not ordered by their frequency. There can be roughly at most $\sqrt{n}$ \tps in the top \tp list. Therefore, we need roughly $\mathcal{O}(\sqrt{n})$ time to retrieve the most frequent \tp. However, by the replacement of this \tp at least $\lfloor\sqrt{n}\rfloor$ edges are absorbed. It is easy to see that, all in all, obtaining the most frequent \tp needs constant time on average.

In a run of \trp we can replace at most $n-1$ \tp occurrences and, as shown before, the replacement of each occurrence, the update of the sets of non-overlapping occurrences and the determination of the most frequent pair can be accomplished in constant time per occurrence replacement. Thus, the whole replacement step can be completed in linear time.

\subsection{Impact of the DAG Representation}\label{sec:impactDag}

In the preceding section, dealing with the complexity of our implementation of the Re-pair for Trees algorithm, we did not pay attention to the underlying DAG representation of the input tree. This enabled us to concentrate on the essentials. Nevertheless, we have to clarify the impact of this representation, particularly concerning the compression performance and the runtime of our implementation, since \trp uses it by default. Only by starting \trp with the \texttt{-no\_dag} switch it forgos the DAG representation and loads the whole input tree into main memory.

Let $\mathcal{G}=(N,P,S)$ be a $0$-bounded SLCF tree grammar. We assume without loss of generality that for every $B\in N$ it holds that $B\leadsto_{\mathcal{G}}^\ast S$. Let $(A\to t)\in P$, $t=(\mathsf{dom}_t,\lambda_t)\in T(\mathcal{F})$ and $v\in\mathsf{dom}_t$. We define the function \texttt{unfold} using the algorithm listed in Fig.~\ref{lst:functionUnfold}.
It holds that $\texttt{unfold}(\mathcal{G},t,v)\subseteq\mathsf{dom}_{\mathsf{val}(\mathcal{G})}$ and it also holds that
\[\bigcup_{\substack{(A\to t)\in P,\\v\in\mathsf{dom}_t}}\texttt{unfold}(\mathcal{G},t,v)=\mathsf{dom}_{\mathsf{val}(\mathcal{G})}\enspace\mbox{.}\]
Let us consider a run $\mathcal{G}_0,\mathcal{G}_1,\ldots,\mathcal{G}_h$ of \trp, where $\mathcal{G}_i=(N_i,P_i,S_i)$, $(S_i\to t_i)\in P_i$, $h\in\mathbb{N}$ and $i\in\{0,1,\ldots,h\}$. Then, in our implementation, $t_i$ is represented by a \mbox{$0$-bounded} (linear) SLCF tree grammar $\overline{\mathcal{G}}_i=(\overline{N}_i,\overline{P}_i,\overline{S}_i)$, \ie, we have $\mathsf{val}(\overline{\mathcal{G}}_i)=t_i$, by default.

\begin{figure}[tb]	\lstset{emph={FUNCTION,ENDFUNC,return,traverse,endtraverse,for,if,else,then,do,endif,endfor,repeat,endrepeat,times}, emphstyle=\bfseries, emph={[2]subtrees_ht}, emphstyle={[2]\slshape}, morecomment=[l]{//}}
	\begin{lstlisting}
FUNCTION unfold((*@$\mathcal{G},t,v$@*))
	let (*@$\mathcal{G}=(N,P,S)$@*) and (*@$A\to t\in P$@*);
	if (*@$\mathsf{ref}_{\mathcal{G}}(A)\neq\emptyset$@*) then
		(*@$M:=\emptyset$@*);
		for each (*@$(t',v')\in\mathsf{ref}_{\mathcal{G}}(A)$@*) do
			(*@$M:=M\cup\{uv\mid u\in\texttt{unfold}(\mathcal{G},t',v')\}$@*);
		endfor
	else
		(*@$M:=\{v\}$@*);
	endif
	return (*@$M$@*);
ENDFUNC
	\end{lstlisting}
	\caption{The algorithm which computes $\texttt{unfold}(\mathcal{G},t,v)$, where we have $t\in T(\mathcal{F}\cup\mathcal{N})$ and $v\in\mathsf{dom}_t$.}\label{lst:functionUnfold}
\end{figure}
\begin{figure}[tb]
	\lstset{emph={FUNCTION,ENDFUNC,return,for,if,else,then,do,endif,endfor,while,endwhile}, emphstyle=\bfseries,morecomment=[l]{//},commentstyle=\color{gray}}
	\begin{lstlisting}
FUNCTION retrieve-all-occs-dag((*@$\overline{t}$@*))
	(*@$v:=\varepsilon$@*);
	while (true) do
		(*@$v:=\;$@*)next_in_postorder((*@$\overline{t}$@*), (*@$v$@*));
		if ((*@$v\neq\varepsilon$@*)) then
			if ((*@$\lambda_{\overline{t}}(v)\notin\mathcal{N}$@*)) then
				(*@$\alpha:=(\lambda_{\overline{t}}(\mathsf{parent}(v)),\mathsf{index}(v),\lambda_{\overline{t}}(v))$@*);
				if ((*@$v\notin\mathsf{occ}_{\overline{t}}'(\alpha)$@*)) then
					(*@$\mathsf{occ}_{\overline{t}}'(\alpha):=\mathsf{occ}_{\overline{t}}'(\alpha)\cup\{\mathsf{parent}(v)\}$@*);
				endif
			else
				let (*@$\overline{t}'$@*) be the right-hand side of (*@$\lambda_{\overline{t}}(v)$@*);
				if ((*@$\lambda_{\overline{t}'}(\varepsilon)\neq \lambda_{\overline{t}}(\mathsf{parent}(v))$@*) then
					(*@$\alpha:=(\lambda_{\overline{t}}(\mathsf{parent}(v)),\mathsf{index}(v),\lambda_{\overline{t}'}(\varepsilon))$@*);
					(*@$\mathsf{occ}_{\overline{t}}'(\alpha):=\mathsf{occ}_{\overline{t}}'(\alpha)\cup\{\mathsf{parent}(v)\}$@*);
				endif
			endif
		else
			return;
		endif
	endwhile
ENDFUNC
	\end{lstlisting}
	\caption{The function \texttt{retrieve-all-occs} listed in Fig.~\ref{lst:functionRetrieveAllOccs} adapted for the DAG case. For every $\alpha\in\Pi$ the set $\mathsf{occ}_{\overline{t}}(\alpha)$ is initially set to $\emptyset$.}\label{lst:functionRetrieveAllOccsDAG}
\end{figure}

\subsubsection{Constructing the Sets of Non-overlapping Occurrences} 

In the first iteration of Tree\-Re\-Pair we need to construct the set $\mathsf{occ}_{t_0}(\alpha)$ for every \tp $\alpha\in\Pi$ occurring in $t_0$. Our first try to accomplish this could be a postorder traversal of all the right-hand sides of $\overline{P}_0$'s productions using the function \texttt{retrieve-all-occs} listed in Fig.~\ref{lst:functionRetrieveAllOccs} on page \pageref{lst:functionRetrieveAllOccs}. However, when traversing the right-hand sides of the DAG grammar $\overline{\mathcal{G}}_0$ individually, we do not consider occurrences spanning two productions of the DAG.
\begin{figure}[t]
	\centering
	\begin{tikzpicture}[->,>=stealth',semithick]
		\tikzstyle{level 1}=[level distance=0.8cm, sibling distance=1.5cm]
		\tikzstyle{level 2}=[level distance=0.8cm, sibling distance=0.5cm]
		\node {$f$}
			child {node {$g$}
				child {node {$a$}}
				child {node {$b$}}
				child {node {$c$}}
			}
			child {node {$g$}
				child {node {$a$}}
				child {node {$b$}}
				child {node {$c$}}
			}
		;
	\end{tikzpicture}
	\caption{The tree $t\in T(\mathcal{F})$ which can be represented by a DAG grammar with productions $(S\to f(A,A))$ and $(A\to g(a,b,c))$.}\label{fig:missedOccurrences}
\end{figure}
\begin{example}
Consider the DAG grammar $\mathcal{G}=(N,P,S)$, where $N=\{S,A\}$ and $P$ contains the two productions $(S\to f(A,A))$ and $(A\to g(a,b,c))$. It is a compressed representation of the tree $t\in T(\mathcal{F})$ depicted in Fig.~\ref{fig:missedOccurrences}. If we would use the function \texttt{retrieve-all-occs} to determine all \tp occurrences in the right-hand sides of $P$'s productions, we would not capture the node $\varepsilon\in\mathsf{dom}_t$ which is an occurrence for both the \tp $(f,1,g)$ and the \tp $(f,2,g)$.
\end{example}
As we have seen, it is necessary to modify the \texttt{retrieve-all-occs} function slightly to also take occurrences spanning two productions into account. We use the algorithm listed in Fig.~\ref{lst:functionRetrieveAllOccsDAG} to obtain the set $\mathsf{occ}_{\overline{t}}'(\alpha)$ for every right-hand side $\overline{t}$ of $\overline{\mathcal{G}}_0$'s productions and every \tp $\alpha\in\Pi$ occurring in $t_0$. After that, we set
\[\mathsf{occ}_{t_0}'(\alpha):=\bigcup_{\substack{(\overline{A}\to\overline{t})\in\overline{P}_0,\\ v\in\mathsf{occ}_{\overline{t}}'(\alpha)}}\texttt{unfold}(\overline{\mathcal{G}},\overline{t},v)\enspace\mbox{.}\]
We test in line 13 of the \texttt{retrieve-all-occs} function if $\alpha$ has equal parent and child symbols. If this proves to be true, we do not add the corresponding occurrence to $\mathsf{occ}_{\overline{t}}'(\alpha)$, \ie, we do not consider occurrences of a \tp with equal parent and child symbols spanning two productions of the DAG. If we would do so, we would possibly register overlapping occurrences and run into problems during a later replacement of $\alpha$. Consider the following example:
\begin{figure}[t]
	\centering
	\begin{tikzpicture}[->,>=stealth',semithick]
		\tikzstyle{level 1}=[level distance=0.8cm, sibling distance=1.75cm]
		\tikzstyle{level 2}=[level distance=0.8cm, sibling distance=1cm]
		\tikzstyle{level 3}=[level distance=0.8cm, sibling distance=0.75cm]	
		\node {$f$}
			child {node {$f$}
				child {node {$a$}}
				child {node {$f$}
					child {node {$a$}}
					child {node {$a$}}
				}
			}
			child {node {$f$}
				child {node {$a$}}
				child {node {$f$}
					child {node {$a$}}
					child {node {$a$}}
				}
			}
		;
	\end{tikzpicture}
	\caption{The tree $t'\in T(\mathcal{F})$ which can be represented by a DAG grammar with productions $A_1\to f(A_2,A_2)$ and $A_2\to f\big(a,f(a,a)\big)$.}\label{fig:impactDagExample}
\end{figure}
\begin{example}
Consider the DAG grammar $\mathcal{G}=(N,P,A_1)$ given by the productions $(A_i\to t_i)\in P$, where $i\in\{1,2\}$, $t_1=f(A_2,A_2)$ and $t_2=f(a,f(a,a))$. It is a compressed representation of the tree $t'\in T(\mathcal{F})$ depicted in Fig.~\ref{fig:impactDagExample}. We use the algorithm from Fig.~\ref{lst:functionRetrieveAllOccsDAG} to obtain the sets $\mathsf{occ}_{t_i}'(\alpha)$ for $i\in\{1,2\}$ and every \tp $\alpha\in\Pi$ occurring $t'$. Let us assume that we omit the check in line 13, \ie, we also consider occurrences of \tps with equal parent and child symbols spanning two productions. The union
	\[\bigcup_{\substack{i\in\{1,2\},\\v\in\mathsf{occ}'_{t_i}((f,2,f))}}\texttt{unfold}(\mathcal{G},t_i,v)=\{\varepsilon,1,2\}\]
	contains the overlapping occurrences $\varepsilon$ and $2$ of the \tp $(f,2,f)$.
\end{example}
The precaution from line 13 leads sometimes to situations in which we replace fewer occurrences of a \tp with equal parent and child symbols as we would replace when not using the DAG representation.
\begin{figure}[t]
	\centering
	\begin{tikzpicture}[->,>=stealth',semithick]
		\tikzstyle{level 1}=[level distance=0.8cm, sibling distance=1.75cm]
		\tikzstyle{level 2}=[level distance=0.8cm, sibling distance=1cm]
		\tikzstyle{level 3}=[level distance=0.8cm, sibling distance=0.75cm]	
		\node {$f$}
			child {node {$f$}
				child {node {$b$}}
				child {node {$f$}
					child {node {$a$}}
					child {node {$f$}
						child {node {$a$}}
						child {node {$f$}
							child {node {$a$}}
							child {node {$a$}}
						}
					}
				}
			}
			child {node {$f$}
				child {node {$c$}}
				child {node {$f$}
					child {node {$a$}}
					child {node {$f$}
						child {node {$a$}}
						child {node {$f$}
							child {node {$a$}}
							child {node {$a$}}
						}
					}
				}
			}
		;
	\end{tikzpicture}
	\caption{Tree $t''\in T(\mathcal{F})$ with seven overlapping occurrences of the \tp $(f,2,f)$.}\label{fig:impactDagExample2}
\end{figure}
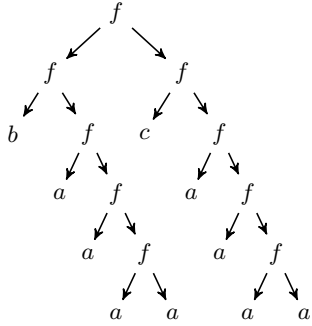
\begin{example}
Consider the tree $t''\in T(\mathcal{F})$ from
Fig.~\ref{fig:impactDagExample2} which can be represented by the DAG
grammar consisting of the two productions $(S\to t_1)$ and $(A\to
t_2)$ with $t_1=f(f(b,A),f(c,A))$ and $t_2=f(a,f(a,f(a,a)))$. After
careful counting one can tell that $t''$ exhibits at most four
non-overlapping occurrences of the \tp $\alpha=(f,2,f)$. However, if
we use the above function \verb|retrieve-all-occs-dag| we only capture
three of them. We obtain 
$\mathsf{occ}'_{t_1}(\alpha)=\{\varepsilon\}$, $\mathsf{occ}'_{t_2}(\alpha)=\{2\}$ and therefore
\[\mathsf{occ}_{t''}'(\alpha)=\bigcup_{\substack{i\in\{1,2\},\\v\in\mathsf{occ}'_{t_i}(\alpha)}}\texttt{unfold}(\mathcal{G},t_i,v)=\{\varepsilon,122,222\}\enspace\mbox{.}\]
\end{example}
Even though this approach does not capture all the occurrences which could be captured when not using the DAG representation, it still achieves a competitive compression performance on our set of test files (\cf Sect.~\ref{sec:resultsWithoutDag} on page \pageref{sec:resultsWithoutDag}). It seems that a more involved method of dealing with \tps with equal parent and child symbols spanning two productions would necessitate a partial unfolding of the DAG. The latter, however, would certainly result in a longer runtime.

\begin{figure}[tb]
	\lstset{emph={FUNCTION,ENDFUNC,return,traverse,endtraverse,for,each,if,else,then,do,endif,endfor,repeat,endrepeat,times}, emphstyle=\bfseries, emph={[2]subtrees_ht}, emphstyle={[2]\slshape}, morecomment=[l]{//}}
	\begin{lstlisting}
FUNCTION remove-absorbed-occs-dag((*@$\overline{t},v,j$@*))
	if ((*@$v\neq\varepsilon$@*)) then
		remove-occ-dag((*@$\overline{t},\mathsf{parent}(v),\mathsf{index}(v)$@*));
	else
		let (*@$\overline{A}$@*) be the right-hand side of (*@$\overline{t}$@*);
		for each (*@$(\overline{t}',u)\in\mathsf{ref}_{\overline{\mathcal{G}}_i}(\overline{A})$@*) do
			remove-occ-dag((*@$\overline{t}',\mathsf{parent}(u),\mathsf{index}(u)$@*));
		endfor
	endif
	
	for ((*@$l\in\{1,2,\ldots,\mathsf{rank}(\lambda_{\overline{t}}(v))\}$@*)) do
		remove-occ-dag((*@$\overline{t},v,l$@*));
	endfor
	
	for ((*@$l\in\{1,2,\ldots,\mathsf{rank}(\lambda_{\overline{t}}(vj))\}$@*)) do
		remove-occ-dag((*@$\overline{t},vj,l$@*));
	endfor
ENDFUNC	

FUNCTION remove-occ-dag((*@$\overline{t},v,j$@*))
	if ((*@$\lambda_{\overline{t}}(vj)\notin\mathcal{N}$@*)) then
		(*@$\alpha:=\big(\lambda_{\overline{t}}(v),j,\lambda_{\overline{t}}(vj)\big)$@*);
	else
		let (*@$\overline{t}'$@*) be the right-hand side of (*@$\lambda_{\overline{t}}(vj)$@*);
		(*@$\alpha:=\big(\lambda_{\overline{t}}(v),j,\lambda_{\overline{t}'}(\varepsilon)\big)$@*);
	endif
	(*@$\mathsf{occ}_{\overline{t}}'(\alpha):=\mathsf{occ}_{\overline{t}}'(\alpha)\setminus\{v\}$@*);
ENDFUNC	
	\end{lstlisting}
	\caption{Listing of the function \texttt{remove-absorbed-occs-dag} which removes all absorbed occurrences from the $\mathsf{occ}_{\overline{t}}'$ sets when using the DAG mode.}\label{lst:remove-absorbed-occs-dag}
\end{figure}
\begin{figure}[tb]
\lstset{emph={FUNCTION,ENDFUNC,return,traverse,endtraverse,for,if,else,then,do,endif,endfor,repeat,endrepeat,times}, emphstyle=\bfseries, emph={[2]subtrees_ht}, emphstyle={[2]\slshape}, morecomment=[l]{//}}
	\begin{lstlisting}
FUNCTION add-new-occs-dag((*@$\overline{t},v$@*))
	if ((*@$v\neq\varepsilon$@*)) then
		add-occ-dag((*@$\overline{t},\mathsf{parent}(v),\mathsf{index}(v)$@*));
	else
		let (*@$\overline{A}$@*) be the right-hand side of (*@$\overline{t}$@*);
		for each (*@$(\overline{t}',u)\in\mathsf{ref}_{\overline{\mathcal{G}}_i}(\overline{A})$@*) do
			add-occ-dag((*@$\overline{t}',\mathsf{parent}(u),\mathsf{index}(u)$@*));
		endfor
	endif
	
	for ((*@$l\in\{1,2,\ldots,\mathsf{rank}(\lambda_{\overline{t}}(v))\}$@*)) do
		add-occ-dag((*@$\overline{t},v,l$@*));
	endfor
ENDFUNC	

FUNCTION add-occ-dag((*@$\overline{t},v,j$@*))
	if ((*@$\lambda_{\overline{t}}(vj)\notin\mathcal{N}$@*)) then
		(*@$\alpha:=\big(\lambda_{\overline{t}}(v),j,\lambda_{\overline{t}}(vj)\big)$@*);
	else
		let (*@$\overline{t}'$@*) be the right-hand side of (*@$\lambda_{\overline{t}}(vj)$@*);
		(*@$\alpha:=\big(\lambda_{\overline{t}}(v),j,\lambda_{\overline{t}'}(\varepsilon)\big)$@*);
	endif
	(*@$\mathsf{occ}_{\overline{t}}'(\alpha):=\mathsf{occ}_{\overline{t}}'(\alpha)\cup\{v\}$@*);
ENDFUNC	
	\end{lstlisting}
	\caption{Listing of the function \texttt{add-new-occs-dag} which adds all new occurrences to the $\mathsf{occ}_{\overline{t}}'$ sets when using the DAG mode.}\label{lst:add-new-occs-dag}
\end{figure}

\subsubsection{Updating the Sets of Non-overlapping Occurrences}

Considering the graph representation of a DAG, a tree node can exhibit multiple parent nodes. In fact, a node has multiple parent nodes if it is the root of the right-hand side of a production of the corresponding DAG grammar and if this production is referenced multiple times.

To capture all \tp occurrences which are absorbed by the replacement of a \tp we need to take care of the above fact. The \texttt{remove-absorbed-occs} function listed in Fig.~\ref{lst:remove-absorbed-occs} needs to be adapted accordingly. Instead of removing one occurrence formed by the node being replaced and its parent, we need to iterate over possibly multiple parents and remove all corresponding occurrences. In Fig.~\ref{lst:remove-absorbed-occs-dag} the function \texttt{remove-absorbed-occs-dag} is listed which incorporates this necessary modification. Analogously, the function \texttt{add-new-occs} listed in Fig.~\ref{lst:add-new-occs} must be modified to work properly in the DAG mode. Fig.\ref{lst:add-new-occs-dag} shows an adapted version.

It is easy to see that our linear runtime is not negatively affected by this loop over all parents. Far from it --- as mentioned earlier, the DAG representation saves us time by avoiding repetitive re-calculations.

\subsubsection{Replacing the \ltps} 

The third and last scenario in which we have to take special care of the DAG representation is when replacing an occurrence of a \tp $\alpha\in\Pi$ spanning two productions of the DAG grammar. Due to our restriction on \tps with equal parent and child symbols the \tp $\alpha$ has to have different parent and child symbols. In the following we want to use an example to describe what needs to be done when replacing the \tp $\alpha$.
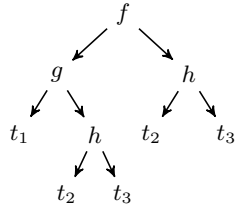
\begin{figure}[t]
	\centering
	\begin{tikzpicture}[->,>=stealth',semithick]
		\tikzstyle{level 1}=[level distance=0.8cm, sibling distance=1.75cm]
		\tikzstyle{level 2}=[level distance=0.8cm, sibling distance=1cm]
		\tikzstyle{level 3}=[level distance=0.8cm, sibling distance=0.75cm]	
		\node {$f$}
			child {node {$g$}
				child {node {$t_1$}}
				child {node {$h$}
					child {node {$t_2$}}
					child {node {$t_3$}}
				}
			}
			child {node {$h$}
				child {node {$t_2$}}
				child {node {$t_3$}}
			}
		;
	\end{tikzpicture}
	\caption{Depiction of the $\mathcal{F}$-labeled tree $t$. We have $t_1,t_2,t_3\in T(\mathcal{F})$.}\label{fig:impactDagExample3}
\end{figure}
\begin{example}
Consider the DAG grammar given by the productions $S\to f\big(g(t_1,A),A\big)$ and $A\to h(t_2,t_3)$ which represents the $\mathcal{F}$-labeled tree $t$ depicted in Fig.~\ref{fig:impactDagExample3}.
Imagine that we want to replace the sole occurrence of the \tp $(f,2,h)$, \ie, an occurrence spanning two productions.\footnote{For the sake of convenience, our example uses a rather small tree and we decide to replace a \tp occurring only once. We could easily enlarge $t$ such that $(f,2,h)$ occurs multiple times and still show the following.} In order to do that we mainly have to complete the following three steps.
\begin{enumerate}[(1)]
	\item We first have to introduce for every child of the node labeled by $h$ a new production. Thus, we obtain two new productions $B\to t_2$ and $C\to t_3$. We can skip this step for every child node which is already labeled by a nonterminal of the DAG grammar. 
	\item We need to update the production with left-hand side $A$ to $A\to h(B,C)$. 
	\item Finally, we introduce a new nonterminal $D$ representing the \tp $(f,2,h)$ and update the production for $S$ to
		\[S\to D(g(t_1,A),B,C)\enspace\mbox{.}\]
\end{enumerate}
The above steps are only necessary if the production with left-hand side $A$ is referenced more than once. Otherwise we could have directly connected the children of $h$ to the newly introduced node labeled by $D$ and removed the production with left-hand side $A$ from the grammar.
\end{example} 
Since at most $k$ new productions need to be introduced, the replacement of a \tp occurrence can still be accomplished in constant time. All in all, it has become clear that even when representing the input tree of our algorithm as a DAG our implementation runs in linear time.

\subsection{Technical Details on the Prototype}

The source code of the \trp prototype and its documentation is available at the Google Code\texttrademark\space open source developer site. It can be accessed by visiting the following web page:
\begin{center}
	\url{http://code.google.com/p/treerepair}
\end{center}
However, the implementation should be considered to be of alpha quality. There is still a lot of testing to be done. 

We also implemented a decompressor called \mbox{TreeDePair} which is contained in the \trp distribution. It is not optimized in terms of time and memory usage.

The software is licensed under the GPLv3 license which is available at
\begin{center}
	\url{http://www.gnu.org/licenses/gpl-3.0.txt}
\end{center}
It is implemented using the C++ programming language and can be compiled at least under the Windows and Linux operating systems. For compile instructions and library requirements, see the \texttt{README.txt} file in the root directory of the \trp distribution.

\section{Succinct Coding}\label{ch:succinctCoding}

In order to achieve a compact representation of the input tree of our \trp algorithm we further compress the generated linear SLCF tree grammar by a binary succinct coding. The technique we use is loosely based on the DEFLATE algorithm described in \cite{Deutsch96deflate}. In fact, we use a combination of a fixed-length coding, multiple Huffman codings and a run-length coding to encode different aspects of the grammar (\cf Fig.~\ref{fig:encodingsHierarchy}).
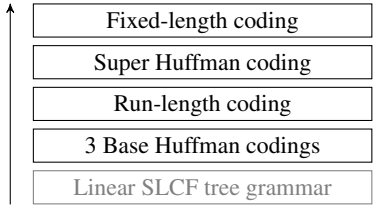
\begin{figure}[t]
	\centering\small
	\begin{tikzpicture}
		[bend angle=25, node distance=0.1cm,text depth=0.1ex,
		box/.style={draw=black,minimum width=4.5cm},
		arrow/.style={->,>=stealth'}]

		\node[box] (fixed) {Fixed-length coding};
		\node[box] (super) [below=of fixed] {Super Huffman coding};
		\node[box] (runbase) [below=of super] {Run-length coding};
		\node[box] (base) [below=of runbase] {3 Base Huffman codings};
		\node[box,draw=gray,font=\color{gray}] (content) [below=of base] {Linear SLCF tree grammar};
		
		\draw[arrow] ([xshift=-0.3cm]content.south west) -- ([xshift=-0.3cm]fixed.north west);
	\end{tikzpicture}
	\caption{Hierarchy of the employed encodings.}\label{fig:encodingsHierarchy}
\end{figure}

In spite of the fact that we obtain an extremely compact binary representation of the generated SLCF tree grammar we are still able to directly execute queries on it with little effort. Basically, we only have to reconstruct the Huffman trees to be able to partially decompress the grammar on demand.

In \cite{Maneth08xml} many different variants of succinct codings specialized in SLCF tree grammars were investigated. Among them there was one encoding scheme which turned out to achieve the best compression performance in general --- at least with respect to the set of sample SLCF tree grammars which was used in this work. However, our experiments show that, regarding the SLCF tree grammars generated by \trp, this encoding is outperformed by the succinct coding which we present in this section.

\subsection{General Remarks}\label{sec:succinctCodingGeneralRemarks}

In this section, we want to elaborate on the following topics: How do we need to modify the pruning step of our algorithm to make our succinct coding as efficient as possible? How does \trp efficiently deal with parameter nodes? How can we serialize a Huffman tree in a compact way?

\subsubsection{Inefficient Productions}

Our experiments showed that, at least for our set of test XML documents, we achieve better compression results in terms of the size of the output file if we slightly modify the pruning step of our algorithm. It turns out that our succinct coding, which we describe in the following sections, is most efficient if we prune all productions with a $\mathsf{sav}$-value smaller than or equal to $2$ (instead of pruning all productions with a $\mathsf{sav}$-value smaller than or equal to $0$ as it is described in Sect.~\ref{sec:pruningStep} on page \pageref{sec:pruningStep}). However, we use this modification only if we make the size of the output file a top priority (by using the switch \texttt{-optimize filesize}). Otherwise, when optimizing the number of edges of the final grammar (\ie, when using the switch \texttt{-optimize edges}), we stick to the original version of the pruning step.

\subsubsection{Handling of Parameter Nodes}\label{sec:handlingOfParameterNodes}

Let $\mathcal{G}=(N,P,S)$ be the linear SLCF tree grammar which was generated by a run of \trp. Then, for every production $(A\to t)\in P$ it holds that $y_i\in\mathcal{Y}$ labels the $i$-th parameter node of $t$ in preorder, where $i\in\{1,2,\ldots,\mathsf{rank}(A)\}$. Due to this fact it is sufficient to represent the parameter symbols $y_1,y_2,\ldots,y_{\mathsf{rank}(A)}\in\mathcal{Y}$ by a single parameter symbol $y\in\mathcal{Y}$. Let $(B\to t')\in P$ be another production and let $v\in\mathsf{dom}_{t'}$ with $\lambda_{t'}(v)=A$. Now, let us assume that we want to eliminate the production $(A\to t)$ and that we use only a single parameter symbol labeling all parameter nodes. It is clear that the $i$-th (in preorder) parameter node of $t$ must be replaced by the subtree which is rooted at the $i$-th child of $v$. 

Our implementation takes advantage of the above simplification, \ie, it uses only one parameter symbol $y$ for every occurring parameter node.

\subsubsection{Serializing Huffman trees}\label{sec:serializingHuffmanTrees}

As stated in \cite{Deutsch96deflate}, it is sufficient to only write out the lengths of the generated codes to be able to reconstruct a Huffman tree at a later date. However, this requires the decompressor to be aware of the following.
\begin{itemize}
	\item What symbols are encoded by the corresponding Huffman tree?
	\item In what order are their code lengths listed?
\end{itemize}
In our case only integers need to be encoded by Huffman codings because we will encode all symbols by integers (see Sect.~\ref{sec:contentsOutputFile} on page \pageref{sec:contentsOutputFile}). Hence, it is obvious to use the natural order of integers to list the lengths of the generated codes. Let us assume that $n\in\mathbb{N}$ is the biggest integer which needs to be encoded and which was assigned a code to, respectively. We just need to loop over all integers $m\leq n$ in their natural order and print out the corresponding code length for each of it. For every $k<n$ for which no code was assigned to we print out a code length of $0$.

In order to solely rely on the code lengths there is still something which needs to be considered. We are required to assign new codes to the integers based on the lengths of their original codes. More precisely, the new code assignment has to fulfill the following two requirements.
\begin{enumerate}[(1)]
	\item\label{CodeLengthsReq1} All codes of the same code length exhibit lexicographically consecutive values when ordering them in the natural order of the integers they represent.
	\item\label{CodeLengthsReq2} Shorter codes lexicographically precede longer codes.
\end{enumerate}
This reorganization of the Huffman codes does not affect the compression performance of the coding since only codes of the same length are swapped. The following example is based on an example from \cite{Deutsch96deflate}.
\begin{example}
	Imagine that we want to use a Huffman coding to encode the letters $a$, $b$, $c$ and $e$ which are each occurring multiple times in a data stream. Let us assume that we obtain the Huffman codes listed in Table \ref{tbl:huffmanCodingExampleBefReordering}. In order to be able to store the corresponding Huffman tree by only writing out the lengths of the Huffman codes we need to assign new codes to the letters. Table \ref{tbl:huffmanCodingExampleAfterReordering} shows the newly assigned codes which fulfill the above two requirements (\ref{CodeLengthsReq1}) and (\ref{CodeLengthsReq2}).
\begin{table}[t]
	\small
	\renewcommand{\subfigcapmargin}{-1cm}
	\hfill
	\subtable[{Huffman coding before the reorganization of the codes. The letters are listed in their natural order, \ie, in alphabetic order.}]{
	\centering
		\begin{tabular*}{2.5cm}{c@{\extracolsep{\fill}}c}
			\toprule
			Symbol&Code\\
			\midrule
			$a$	&$00$\\
			$b$	&$1$\\
			$c$	&$011$\\
			$e$	&$010$\\
			\bottomrule
		\end{tabular*}
	\label{tbl:huffmanCodingExampleBefReordering}
	}
	\hfill
	\subtable[{Huffman coding from Table \ref{tbl:huffmanCodingExampleBefReordering} after the reorganization of the codes.}]{
	\centering
		\begin{tabular*}{2.5cm}{c@{\extracolsep{\fill}}c}
			\toprule
			Symbol&Code\\
			\midrule
			$a$	&$10$\\
			$b$	&$0$\\
			$c$	&$110$\\
			$e$	&$111$\\
			\bottomrule
		\end{tabular*}
	\label{tbl:huffmanCodingExampleAfterReordering}
	}
	\hfill
\end{table}
Now, let us assume that the decompressor expects the code lengths to be the lengths of codes assigned to the letters of the Latin alphabet and that these code lengths are ordered in the natural order of the letters they represent. Then, the corresponding Huffman tree can be unambiguously represented by the following sequence of code lengths: $2, 1, 3, 0, 3$. Note that we need to insert a code length of $0$ at the position of the letter $d$ since there is no code assigned to the letter $d$.
\end{example}

\subsection{Contents of the Output File}\label{sec:contentsOutputFile}

In this section we want to elaborate on the information which needs to be stored in the output file of our algorithm in order to be able to reconstruct the generated linear SLCF tree grammar at a later date. We also want to demonstrate how this data can be efficiently represented. However, at this time we do not pay attention to the fixed-length, run-length or Huffman codings which are employed in a subsequent step of the encoding process. For the sake of simplicity we consider these encodings in separate subsections of this section.

Let $\mathcal{G}=(N,P,S)$ be the linear SLCF tree grammar which was generated by a run of \trp. Before we are able to compile the information which needs to be written out we need to assign to every symbol from $\mathcal{F}\cup (N\setminus\{S\})\cup\{y\}$ a unique integer. In fact, we assign to every symbol from $\mathcal{F}$ a unique ID from the set \mbox{$\{1,2,\ldots,|\mathcal{F}|\}\subset\mathbb{N}$}. We assign the ID $|\mathcal{F}|+1$ to $y$, \ie, to the special symbol labeling all parameter nodes in the right-hand sides of $P$'s productions. Finally, we associate with every symbol from the set of nonterminals $N\setminus\{S\}$ a unique ID from the set $\{|\mathcal{F}|+2,|\mathcal{F}|+3,\ldots,|\mathcal{F}|+|N|\}$. The IDs are assigned to the nonterminals in such a way that the nonterminal $A\in N\setminus\{S\}$ has a higher ID than the nonterminal $B\in N\setminus\{S\}$ if $B\leadsto_{\mathcal{G}}^+A$ holds.

\subsubsection{Writing out the Necessary Informations}

Now, we are able to write out the information needed to reconstruct $\mathcal{G}$ in four steps. Bear in mind that the values mentioned below are not directly written to the output file but that they are additionally encoded by a combination of multiple Huffman codings, a run-length coding and a fixed-length coding later on. 

\paragraph*{First step} In the first step, we write out the number of terminal symbols $|\mathcal{F}|$ and the number of introduced productions $|N|-1$, \ie, we are not counting the start production. By handing over this information to the decompressor we avoid the insertion of separators marking, for instance, the end of the enumeration of elements types (which are written out in the third step).

\paragraph*{Second step} In the second step, we directly append a representation of the children characteristics of the terminal symbols. By children characteristics we mean their rank and, concerning terminal symbols of rank $1$, if we are dealing with a left or a right child.\footnote{Consult Sect.~\ref{sec:binaryTreeModel} on page \pageref{sec:binaryTreeModel} for an explanation on why this information is necessary to reconstruct the input tree.} Due to the fact that all terminal symbols have a rank of at most two, we can encode this information using two bits per symbol. Table \ref{tbl:childrenCharacteristics} lists all the bit strings we use together with a brief description of their meanings.
\begin{table}[t]
	\centering\small
	\begin{tabular*}{4cm}{c@{\extracolsep{\fill}}l}
		\toprule
		Bit string&Description\\
		\midrule
		$00$&rank $0$\\
		$01$&rank $1$, right child\\
		$10$&rank $1$, left child\\
		$11$&rank $2$\\
		\bottomrule
	\end{tabular*}
	\caption{The bit strings encoding the children characteristics together with their meaning.}\label{tbl:childrenCharacteristics}
\end{table}
We write out the children characteristics as follows: Firstly, we print out a bit string from Table \ref{tbl:childrenCharacteristics} representing a certain children characteristic. After that we append the number of corresponding terminal symbols and finally we enumerate their IDs. We do this for the characteristics $00$, $01$ and $10$. We omit the enumeration of all terminal symbols with a rank of $2$ since their IDs can be reconstructed with the information in hand. In fact, we just need to subtract the set of IDs of all terminal symbols with children characteristics $00$, $01$ and $10$ from the set of IDs of all terminal symbols from $\mathcal{F}$ (which is $\{1,2,\ldots,|F|\}$). 

Furthermore, it is not necessary to print out the ranks of the nonterminals from $N$ since these can be easily reconstructed by counting the number of parameter nodes in the corresponding right-hand sides. The latter are written to the output file in the fourth step.

\paragraph*{Third step} In this step, we print the element types of the terminal symbols in the ascending order of their IDs to the output file. We do this by writing out the ASCII code of every single letter. The individual names are terminated by the ASCII character \texttt{ETX} which is assumed not to be used within the element types of the terminal symbols.

\paragraph*{Fourth step} In this last step we serialize the productions of $\mathcal{G}$ in the ascending order of the IDs of their left-hand sides. For every production $(A\to t)\in P$ we just write out the IDs of the labels of $t$'s nodes in preorder. We do not need to use special marker symbols to indicate the nesting structure of the symbols and their IDs, respectively. When parsing the output file this hierarchy can be easily obtained by taking care of the individual ranks of the symbols.

We can also omit the specification of the left-hand side $A$ since both, its ID and its rank, can be reconstructed with the information in hand. Imagine that we are parsing the output file to reconstruct the productions of $\mathcal{G}$. If we are parsing the $i$-th production, the ID of its left-hand side must be $|\mathcal{F}|+1+i$, where $i\in\{1,2,\ldots,|N|\}$. As already mentioned, the rank of the left-hand side can be obtained by counting the parameter nodes in the right-hand side once this has been reconstructed.

Note that it is superfluous to insert separators between the representations of the productions from $P$ since their boundaries can be calculated based on the ranks of the symbols. Again, imagine that we are trying to reconstruct the productions of $P$ by parsing the output file of our algorithm. Let $(A\to t)\in P$ be the first production we encounter. The tree $t$ can only consist of nodes labeled by terminal symbols, \ie, we must have $t\in T(\mathcal{F})$.\footnote{This is due to the fact that we have written out the productions in the ascending order of the IDs of their left-hand sides. These IDs were assigned to the nonterminals in such a way that the nonterminal $A\in N\setminus\{S\}$ has a higher ID than $B\in N\setminus\{S\}$ if $B\leadsto_{\mathcal{G}}^+A$ holds. Therefore, the right-hand side of $(A\to t)$, which is the first production which was written out, does not contain any node labeled by a nonterminal from $N$.} The ranks of all symbols from $\mathcal{F}$ are known since the necessary information was written to the compressed file in the second step. Therefore, we can easily reconstruct $t$ by iteratively parsing the corresponding IDs in the output file. While doing so we are also able to count the number of occurrences of the symbol $y\in\mathcal{Y}$ in $t$. Thus, we are aware of the value of $\mathsf{rank}(A)$. After that, we proceed with decoding the second production $(A'\to t')\in P$ by iteratively parsing the next IDs. We have $t'\in T(\mathcal{F}\cup\{A\})$, \ie, the ranks of all occurring symbols are known. That way all productions from $P$ can be reconstructed.
\begin{example}\label{ex:mainExampleEncoding}
	In order to get a clear picture of the representation described above we apply the previous four steps to the linear SLCF tree grammar $\mathcal{G}=(N,P,S_4)$ over the ranked alphabet $\mathcal{F}$ from Sect.~\ref{exampleGrammar} on page \pageref{exampleGrammar}, \ie, we have $N=\{S_4,A_2,A_3\}$ and $P$ is the following set of productions:
\begin{align*}
	S_4&\to\mathsf{books}^{10}(A_3(A_3(A_3(A_3(\mathsf{book}^{10}(A_2))))))\\
	A_3(y)&\to\mathsf{book}^{11}(A_2, y)\\
	A_2&\to\mathsf{author}^{01}(\mathsf{title}^{01}(\mathsf{isbn}^{00}))
\end{align*}
	First of all, we assign to every symbol from $\mathcal{F}\cup(N\setminus\{S_4\})\cup\{y\}$ a unique ID as it is shown in Fig.~\ref{tbl:succinctCodingAssignedIds}. After that we are able to write out the grammar exactly as described above resulting in the value sequence depicted in Fig.~\ref{fig:valueSequence}. We accomplish this task in four steps: 
\begin{figure}[t]
	\centering
	\begin{tabular*}{3cm}{l@{\extracolsep{\fill}}l}
		\toprule
		Symbol&ID\\
		\midrule
		$\mathsf{books}^{10}$		&$1$\\
		$\mathsf{isbn}^{00}$		&$2$\\
		$\mathsf{title}^{01}$		&$3$\\
		$\mathsf{author}^{01}$		&$4$\\
		$\mathsf{book}^{10}$		&$5$\\
		$\mathsf{book}^{11}$		&$6$\\
		$y$						&$7$\\
		$A_2$					&$8$\\
		$A_3$					&$9$\\
		\bottomrule
	\end{tabular*}
	\caption{All symbols with the ID assigned to them. The symbol $y$ is the symbol used to label the parameter nodes in the right-hand sides of $P$'s productions.}\label{tbl:succinctCodingAssignedIds}
\end{figure}
\begin{enumerate}[(1)]
	\item We begin by writing out the number of terminals ($6$) directly followed by the number of nonterminals minus the start nonterminal (2) --- see the values 0 and 1 in the depiction. 
	\item After that the children characteristics of all terminal symbols are written to the file. We begin by specifying all terminal symbols of rank $0$ (values 2--4). This is done by firstly writing out the bit string $00$ and the number of corresponding symbols ($1$). Finally, the ID $2$ of the terminal symbol $\mathsf{isbn}^{00}$, which is the sole terminal symbol of rank $0$, is listed. 
	
		Analogously, the terminal symbols with children characteristics $01$ and $10$ are enumerated (values 5--12).
	\item Now, the element types of all terminal symbols are exported to the output file (values 13--46). For each of them the decimal value of each ASCII character is written out. The element type \emph{books}, for instance, is encoded by the sequence $98$, $111$, $111$, $107$, $115$. 
	\item Finally, the productions from $P$ are written out in the ascending order of the IDs of their left-hand sides. Thus, the production with left-hand side $A_2$ is serialized as the very first production (values 47--49). It is encoded by the unambiguous sequence of IDs $4, 3, 2$ representing the terminal symbols $\mathsf{author}^{01}$, $\mathsf{title}^{01}$ and $\mathsf{isbn}^{00}$ of the right-hand side of $A_2$ in preorder.
		Afterwards the remaining productions with left-hand sides $A_3$ (values 50--52) and $S$ (values 53--59) are printed to the output file in this order. 
\end{enumerate}
\end{example}
\begin{figure}[tb]
         \centering
		\pgfdeclarelayer{background}
		\pgfsetlayers{background,main}
		\begin{tikzpicture}[semithick]
			\pgfmathsetmacro{\cols}{16}
			\pgfmathsetmacro{\zero}{0}
			\pgfmathsetmacro{\hei}{0.5}
			
			\pgfmathdivide{\hei}{2}
			\pgfmathsetmacro{\halfhei}{\pgfmathresult}
			
			\pgfmathsetmacro{\dist}{1.8}
			\newcounter{curval}
			\setcounter{curval}{0}
		
			\foreach \content in {6,2,00,1,2,01,2,3,4,10,2,1,5,98,111,111,107,115,3,105,115,98,11
0,3,116,105,116,108,101,3,97,117,116,104,111,114,3,98,111,111,
107,3,98,111,111,107,3,4,3,2,6,8,7,1,9,9,9,9,5,8} {
				\pgfmathmod{\value{curval}}{\cols}
				\pgfmathsetmacro\mod\pgfmathresult

					\pgfmathdivide{\value{curval}}{\cols}
					\pgfmathfloor{\pgfmathresult}
					\pgfmathsetmacro{\div}{\pgfmathresult}

				\pgfmathmultiply{\div}{\dist}
				\pgfmathsetmacro{\y}{\pgfmathresult}
				
				\begin{scope}[xshift=\mod cm,yshift=-\y cm]
				
					\draw (0.5,-0.5) node[inner sep=0pt, outer sep=0pt] (value\arabic{curval}southeast) {};
					\draw (0.5,0.5) node[inner sep=0pt, outer sep=0pt] (value\arabic{curval}northeast) {};
					\draw (-0.5,-0.5) node[inner sep=0pt, outer sep=0pt] (value\arabic{curval}southwest) {};
					\draw (-0.5,0.5) node[inner sep=0pt, outer sep=0pt] (value\arabic{curval}northwest) {};
					\draw (0,0) node[inner sep=0pt, outer sep=0pt] (value\arabic{curval}center) {};
					\draw (0,0.5) node[inner sep=0pt, outer sep=0pt] (value\arabic{curval}northcenter) {};
				
						\draw (-0.5,-0.5) -- (0.5,-0.5) -- (0.5,0.5) -- (-0.5,0.5) -- cycle;
					\draw (0,0) node[anchor=center] {$\content$};
					\draw (-0.5,0.5) node[anchor=north west, gray] {\scriptsize\arabic{curval}};
				\end{scope}
				
				\stepcounter{curval}
			}
			
			\begin{scope}[pattern color=gray]
				\filldraw[draw=black,fill=lightgray,pattern=north west lines] (value2northwest) -- (value12northeast) -- ([xshift=0pt,yshift=\hei cm]value12northeast.center) -- ([xshift=0pt,yshift=\hei cm]value2northwest.center) -- (value2northwest);
				\draw ([yshift=\halfhei cm]value7northcenter) node[anchor=center, black,fill=white, inner sep=2pt] {\footnotesize children characteristics};
			
				\filldraw[draw=black,fill=lightgray,pattern=dots] (value13northwest) -- (value15northeast) -- ([xshift=0pt,yshift=\hei cm]value15northeast.center) -- ([xshift=0pt,yshift=\hei cm]value13northwest.center) -- (value13northwest);
				\filldraw[draw=black,fill=lightgray,pattern=dots] (value16northwest) -- (value31northeast) -- ([xshift=0pt,yshift=\hei cm]value31northeast.center) -- ([xshift=0pt,yshift=\hei cm]value16northwest.center) -- (value16northwest);
				\filldraw[draw=black,fill=lightgray,pattern=dots] (value32northwest) -- (value46northeast) -- ([xshift=0pt,yshift=\hei cm]value46northeast.center) -- ([xshift=0pt,yshift=\hei cm]value32northwest.center) -- (value32northwest);
				\draw ([yshift=\halfhei cm]value24northwest) node[anchor=center,black,fill=white,inner sep=2pt] {\footnotesize element types};
				
				\filldraw[draw=black,fill=lightgray,pattern=north east lines] (value46northwest) -- (value47northeast) -- ([xshift=0pt,yshift=\hei cm]value47northeast.center) -- ([xshift=0pt,yshift=\hei cm]value47northwest.center) -- (value47northwest);
				\filldraw[draw=black,fill=lightgray,pattern=north east lines] (value48northwest) -- (value59northeast) -- ([xshift=0pt,yshift=\hei cm]value59northeast.center) -- ([xshift=0pt,yshift=\hei cm]value48northwest.center) -- (value48northwest);
				\draw ([xshift=0cm,yshift=\halfhei cm]value53northeast) node[anchor=center, black,fill=white, inner sep=2pt] {\footnotesize productions};
			\end{scope}
			
			\begin{pgfonlayer}{background}
				\foreach \l/\r/\lc/\rc in {2/4/lightgray/white,5/8/lightgray/white,9/12/lightgray/white,13/15/lightgray/black!05,16/17/black!05/white,19/22/lightgray/white,24/28/lightgray/white,30/31/lightgray/black!10,32/35/black!10/white,37/40/lightgray/white,42/45/lightgray/white,47/47/lightgray/black!10,48/49/black!10/white,50/52/lightgray/white,53/59/lightgray/white} {
					\shade[left color=\lc,right color=\rc] (value\l northwest.center) rectangle (value\r southeast.center);
				}
			\end{pgfonlayer}

			\begin{scope}[anchor=south east, darkgray]
			\draw (value13southeast) node {\scriptsize'b'};
			\draw (value14southeast) node {\scriptsize'o'};
			\draw (value15southeast) node {\scriptsize'o'};
			\draw (value16southeast) node {\scriptsize'k'};
			\draw (value17southeast) node {\scriptsize's'};
			\draw (value18southeast) node {\scriptsize'ETX'};
			\draw (value19southeast) node {\scriptsize'i'};
			\draw (value20southeast) node {\scriptsize's'};
			\draw (value21southeast) node {\scriptsize'b'};
			\draw (value22southeast) node {\scriptsize'n'};
			\draw (value23southeast) node {\scriptsize'ETX'};
			\draw (value24southeast) node {\scriptsize't'};
			\draw (value25southeast) node {\scriptsize'i'};
			\draw (value26southeast) node {\scriptsize't'};
			\draw (value27southeast) node {\scriptsize'l'};
			\draw (value28southeast) node {\scriptsize'e'};
			\draw (value29southeast) node {\scriptsize'ETX'};
			\draw (value30southeast) node {\scriptsize'a'};
			\draw (value31southeast) node {\scriptsize'u'};
			\draw (value32southeast) node {\scriptsize't'};
			\draw (value33southeast) node {\scriptsize'h'};
			\draw (value34southeast) node {\scriptsize'o'};
			\draw (value35southeast) node {\scriptsize'r'};
			\draw (value36southeast) node {\scriptsize'ETX'};
			\draw (value37southeast) node {\scriptsize'b'};
			\draw (value38southeast) node {\scriptsize'o'};
			\draw (value39southeast) node {\scriptsize'o'};
			\draw (value40southeast) node {\scriptsize'k'};
			\draw (value41southeast) node {\scriptsize'ETX'};
			\draw (value42southeast) node {\scriptsize'b'};
			\draw (value43southeast) node {\scriptsize'o'};
			\draw (value44southeast) node {\scriptsize'o'};
			\draw (value45southeast) node {\scriptsize'k'};
			\draw (value46southeast) node {\scriptsize'ETX'};
			\end{scope}
		\end{tikzpicture}
	\caption{Representation of the grammar $\mathcal{G}$ from Example \ref{ex:mainExampleEncoding}.}\label{fig:valueSequence}
\end{figure}

\subsubsection{Possible Optimizations}

Of course, there is still room to further reduce the data which needs to be written to the output file. Consider, for instance, terminal symbols of the same element type but different children characteristics. In the case of our implementation, the element type of these symbols is written to the file two or three times in the second step. However, an optimization with respect to this redundancy does only lead to marginally better compression results. This is due to the fact that typically the major part of the output file is the enumeration of the productions.

Still regarding the second step, we could at first determine the most frequent children characteristic and omit the enumeration of all corresponding terminal symbols. This dynamic approach certainly leads to a small reduction of the size of the output file compared to always skipping the children characteristic $11$.

Another aspect which confesses optimization potential are possible long lists of the parameter symbol $y$ which emerge when writing out the right-hand sides of productions with a higher rank. In this case, run-length coding can lead to a better compression performance. However, we did not further investigate this matter since we focus on generating grammars with nonterminals with a maximal rank of $4$.

\subsection{Employing Multiple Types of Encodings}\label{sec:huffmanCoding}

Even though a Huffman tree has to be serialized for every Huffman coding used within our output file, we decided in favor of using four distinct Huffman codings. We use three of them for encoding 
\begin{itemize}
	\item the start production, 
	\item the remaining productions, the children characteristics of the terminal symbols and the numbers of terminals and nonterminals, and finally
	\item the names of the terminals.
\end{itemize}
In the sequel, we call these three Huffman codings the \emph{base Huffman codings}. The fourth Huffman coding, which we call \emph{super Huffman coding}, is used to encode the Huffman trees of the above codings. Our tests with different numbers of Huffman codings revealed that, in general, the above approach leads to the best compression results. This is at least true for most of the XML test documents we used.

\subsubsection{Base Huffman Codings}

We serialize the three base Huffman codings by writing out the lengths of the generated codes as it is described in Sect.~\ref{sec:serializingHuffmanTrees} on page \pageref{sec:serializingHuffmanTrees}. However, we additionally apply a run-length coding and the super Huffman coding to achieve a compact binary representation. In Sect.~\ref{sec:run-lengthCoding} on page \pageref{sec:run-lengthCoding} we elaborate on how exactly the run-length coding works. We briefly call the length of a code of a base Huffman coding a \emph{base code length} in the sequel. Analogously, we denote the lengths of the codes of the super Huffman coding by the term \emph{super code lengths}.

We output the number of base code lengths in front of every serialized base Huffman coding, \ie, in front of every enumeration of base code lengths. That way the decompressor knows how many bits are part of this binary representation.
Let us point out that this number of code lengths is encoded using $k$ bits instead of using the super Huffman coding, where $k\in\mathbb{N}$ is a constant which is fixed at compile time. We do this due to the following fact. Let $n\in\mathbb{N}$ be the number of code lengths and let us assume that we encode $n$, which is usually many times larger than the maximum over all code lengths, using the super Huffman coding. This would result in a big gap of unused integers between the super code lengths and $n$. This again would lead to a long list of $0$'s when storing the super Huffman tree by enumerating its code lengths. In general, this leads to a reduced compression performance compared to a fixed-length coding of $n$ using $k$ bits.

\subsubsection{Super Huffman Coding}

The super Huffman coding will also be stored by the sequence of its code lengths. However, the relatively small set of integers is encoded by a fixed-length coding using $n\in\mathbb{N}$ bits, where $n$ is the smallest possible number of bits which can be used to encode all super code lengths. More precisely, we serialize the super Huffman coding in three steps:
\begin{enumerate}[(1)]
	\item First of all, we print out the binary representation of the number $n$ using $k$ bits, where $k\in\mathbb{N}$ is a fixed number of bits which is specified at compile time.
	\item Let $m\in\mathbb{N}$ be the biggest base code length. We print out the binary representation of $m$ using $k$ bits. With this information the decompressor knows that the next $n\cdot m$ bits make up the list of super code lengths.
	\item Finally, the binary representations of the $m$ many super code lengths are written to the output file using $n$ bits for each code length. The super code lengths are printed in the natural order of the integers which are represented by the corresponding codes.
\end{enumerate}

\subsubsection{Run-length Coding of the Base Code Lengths}\label{sec:run-lengthCoding}

In this section we explain the run-length coding which is applied to the enumerations of code lengths used to write all base Huffman codings to the output file. This additional encoding marks a major contribution to the compactness of our representation. The bigger a code length is, the more different codes of that length are possible. At the same time a sequence of several occurrences of the same code length within the enumeration of all code lengths becomes more likely. In addition, our experience shows that it frequently happens that there is a longer run of $0$'s in the list of all code lengths due to symbols which no codes were assigned to. 
\begin{example}
Consider, for instance, the example from Sect.~\ref{sec:mainExampleHuffman} and in particular the base Huffman coding $C_3$ which is listed in Table \ref{tbl:mainExampleHuffmanCodingNames} on page \pageref{tbl:mainExampleHuffmanCodingNames}. This Huffman coding does not assign codes to the symbols $4$--$96$. This results in a sequence of $94$ zeros within the enumeration of the code lengths of $C_3$.
\end{example}
\begin{definition}
	Let $m,k\in\mathbb{N}$, where $k\geq\lfloor\log_2(m)\rfloor+1$. In the following we denote by $\mathsf{bin}_k(m)$ the ($0$-padded) binary representation $b_kb_{k-1}\ldots b_0$ of $m$, \ie, the following holds:
\[m=\sum_{i=0}^kb_i\cdot2^i\]
\end{definition}
We encode an enumeration of code lengths using a run-length coding as follows: Let us assume that $n\in\mathbb{N}$ is the maximum code length. Then we use the three additional integers $n+1$, $n+2$ and $n+3$ to indicate certain types of runs --- we call them \emph{run indicators} in the sequel. Principally, all runs with a length less than or equal to $3$ are straightly written to the output file. In contrast, a run of a code length $m\in\mathbb{N}$ exceeding this bound is encoded as follows:
\begin{itemize}
	\item If we have $m>0$, we use the run indicator $n+1$ and a bit string with a length of $2$ to indicate $4$--$7$ repetitions of the code length $m$. If $k>3$ is the length of the run of $m$ and $l=k\mod 7$ (\ie, $l\in\{0,1,\ldots,6\}$), then this run is encoded as follows:
		\begin{itemize}
			\item if $l>3$:
				\[m\;\underbrace{(n+1)\mathsf{bin}_2(3)}_{\lfloor\nicefrac{k}{7}\rfloor\text{ times}}\;(n+1)\mathsf{bin}_2(l-4)\]
			\item if $l\leq 3$:
				\[m\;\underbrace{(n+1)\mathsf{bin}_2(3)}_{\lfloor\nicefrac{k}{7}\rfloor\text{ times}}\;{[m]}^l\]
		\end{itemize}
		Note that ${[m]}^l$ denotes $l$ many consecutive $m$'s.
	\item If we have $m=0$, we use the run indicator $n+2$ with an appended bit string of length $3$ to  denote $4$--$11$ repetitions of $m$. In contrast, we use the run indicator $n+3$ together with a bit string of length $7$ to encode $12$--$139$ repeated $0$'s.
	
		If $k>3$ is the length of the run of $0$'s and $l=k\mod 139$ (\ie, $l\in\{0,1,\ldots,138\}$), then this run is encoded as follows:
		\begin{itemize}
			\item if $l>11$:
				\[\underbrace{(n+3)\mathsf{bin}_7(127)}_{\lfloor\nicefrac{k}{139}\rfloor\text{ times}}\;(n+3)\mathsf{bin}_7(l-12)\]
			\item if $3<l\leq11$:
				\[\underbrace{(n+3)\mathsf{bin}_7(127)}_{\lfloor\nicefrac{k}{139}\rfloor\text{ times}}\;(n+2)\mathsf{bin}_3(l-4)\]
			\item if $l\leq3$:
				\[\underbrace{(n+3)\mathsf{bin}_7(127)}_{\lfloor\nicefrac{k}{139}\rfloor\text{ times}}\;{[m]}^l\]
		\end{itemize}
\end{itemize}
\begin{table}
	\footnotesize
	\hfill
	\subtable[{Huffman coding $C_1$ used to encode the start production.}]{
	\centering
	\begin{tabular*}{4.7cm}[b]{r@{\extracolsep{\fill}}r@{\extracolsep{\fill}}r}
		\toprule
		Symbol&Old code&New code\\
		\midrule
		$0$&$10$		&$10$\\
		$1$&$1110$	&$1110$\\
		$5$&$1111$	&$1111$\\
		$8$&$110$	&$110$\\
		$9$&$0$		&$0$\\
		\bottomrule
	\end{tabular*}
	\label{tbl:mainExampleHuffmanCodingStartProduction}
	}
	\hfill
	\subtable[{Huffman coding $C_2$ used to encode the productions from $P\setminus\{S\}$, the children characteristics, and numbers of terminals and nonterminals.}]{
	\centering
	\begin{tabular*}{4.7cm}[b]{r@{\extracolsep{\fill}}r@{\extracolsep{\fill}}r}
		\toprule
		Symbol&Old code&New code\\
		\midrule
		$1$&$1010$	&$1100$\\
		$2$&$01$		&$00$\\
		$3$&$00$		&$01$\\
		$4$&$111$	&$100$\\
		$5$&$1011$	&$1101$\\
		$6$&$110$	&$101$\\
		$7$&$1001$	&$1110$\\
		$8$&$1000$	&$1111$\\
		\bottomrule
	\end{tabular*}
	\label{tbl:mainExampleHuffmanCodingProductions}
	}
	\hfill
\end{table}
\begin{table}
	\footnotesize
	\hfill
	\subtable[{Huffman coding $C_3$ used to encode the names of the terminal symbols.}]{
	\centering
	\begin{tabular*}{4.7cm}[b]{r@{\extracolsep{\fill}}r@{\extracolsep{\fill}}r}
		\toprule
		Symbol	&Old code	&New code\\
		\midrule
		$3$		&$111$		&$010$\\
		$97$		&$00101$		&$11010$\\
		$98$		&$101$		&$011$\\
		$101$	&$00100$		&$11011$\\
		$104$	&$00111$		&$11100$\\
		$105$	&$1001$		&$1010$\\
		$107$	&$1101$		&$1011$\\
		$108$	&$110011$	&$111110$\\
		$110$	&$110010$	&$111111$\\
		$111$	&$01$		&$00$\\
		$114$	&$11000$		&$11101$\\
		$115$	&$1000$		&$1100$\\
		$116$	&$000$		&$100$\\
		$117$	&$00110$		&$11110$\\
		\bottomrule
	\end{tabular*}
	\label{tbl:mainExampleHuffmanCodingNames}
	}
	\hfill
	\subtable[{Super Huffman coding used to encode the code lengths of the base Huffman codings.}]{
	\centering
	\begin{tabular*}{4.7cm}[b]{r@{\extracolsep{\fill}}r@{\extracolsep{\fill}}r}
		\toprule
		Symbol	&Old code	&New code\\
		\midrule
		$0$		&$0$			&$0$\\
		$1$		&$110000$	&$111110$\\
		$2$		&$1101$		&$1110$\\
		$3$		&$101$		&$100$\\
		$4$		&$111$		&$101$\\
		$5$		&$100$		&$110$\\
		$6$		&$11001$		&$11110$\\
		$9$		&$110001$	&$111111$\\
		\bottomrule
	\end{tabular*}
	\label{tbl:mainExampleSuperHuffmanCoding}
	}
	\hfill
\end{table}
\begin{example}
	Consider the following sequence of integers:
	\begin{center}
		$122333\,444444\,555\,000000000$
	\end{center}
	Now, let us assume that we want to encode the above sequence using our run-length coding. Obviously, we have $n=5$. The above run of $4$'s with a length of 6 is represented by the sequence $4610$ since we have $n+1=6$ and $\mathsf{bin}_2(6-4)=10$. In contrast, the run of $0$'s with a length of $9$ leads to the sequence $7101$ because it holds that $n+2=7$ and that $\mathsf{bin}_3(9-4)=101$. All in all, we obtain the sequence $122333\,4610\,555\,7101$.
\end{example}
Surprisingly, our investigations evinced that an approach which dynamically adjusts the length of the bit strings used in the above encoding depending on the size of the input grammar does not lead to significantly better compression results.

\subsubsection{Example}\label{sec:mainExampleHuffman}

This example continues the encoding of the linear SLCF tree grammar $\mathcal{G}$ from Example \ref{ex:mainExampleEncoding} on page \pageref{ex:mainExampleEncoding}. The Tables \ref{tbl:mainExampleHuffmanCodingStartProduction}, \ref{tbl:mainExampleHuffmanCodingProductions} and \ref{tbl:mainExampleHuffmanCodingNames} list the three base Huffman codings, called $C_1$, $C_2$ and $C_3$ in the sequel, which are calculated by our implementation. The columns labeled \emph{Old code} show the initial Huffman codes while the columns labeled \emph{New code} list the newly assigned codes after the necessary reorganization described in Sect.~\ref{sec:serializingHuffmanTrees} on page \pageref{sec:serializingHuffmanTrees}.

While Fig.~\ref{fig:valueSequence} on page \pageref{fig:valueSequence} shows the second part of the output file as it is generated by a run of \trp the Fig.~\ref{fig:mainExampleFirstPartOutputFile} shows the first part of it. The latter stores the base Huffman codings $C_1$, $C_2$ and $C_3$ together with the corresponding super Huffman coding. For the sake of clarity the corresponding values are denoted by their integer representation instead of by their fixed-length or Huffman code. The Huffman coding $C_1$ from Table \ref{tbl:mainExampleHuffmanCodingStartProduction}, for instance, is given by the sequence of code lengths ranging from value 13 to value 22, where value 12 informs us about the length of this sequence. Analogously, the code lengths of the Huffman codings $C_2$ and $C_3$ are given by the values 24--32 and 34--60, respectively. The sequence of code lengths of the Huffman coding $C_3$ exhibits a longer run, namely, $94$ consecutive occurrences of the code length $0$. This run is encoded by the run indicator $9=n+3$ and the bit string $\mathsf{bin}_7(94-12)=1010010$, where $n=6$ is the maximal length of a code from $C_3$.

\begin{figure}[bt]
         \centering
		\pgfdeclarelayer{background}
		\pgfsetlayers{background,main}
		\begin{tikzpicture}[semithick]
			\pgfmathsetmacro{\cols}{16}
			\pgfmathsetmacro{\zero}{0}
			\pgfmathsetmacro{\hei}{0.5}
			
			\pgfmathdivide{\hei}{2}
			\pgfmathsetmacro{\halfhei}{\pgfmathresult}
			
			\pgfmathsetmacro{\dist}{1.8}
			\newcounter{curvaltwo}
			\setcounter{curvaltwo}{0}
		
			\foreach \content in {3,10,1,6,4,3,3,3,5,0,0,6,10,2,4,0,0,0,4,0,0,3,1,9,0,4,2,2,3,4,3,
			4,4,118,0,0,0,3,9,{\hfill 101\newline\hfill 0010},5,3,0,0,5,0,0,5,4,0,4,6,0,6,2,0,0,5,4,3,5} {
				\pgfmathmod{\value{curvaltwo}}{\cols}
				\pgfmathsetmacro\mod\pgfmathresult

					\pgfmathdivide{\value{curvaltwo}}{\cols}
					\pgfmathfloor{\pgfmathresult}
					\pgfmathsetmacro{\div}{\pgfmathresult}

				\pgfmathmultiply{\div}{\dist}
				\pgfmathsetmacro{\y}{\pgfmathresult}
				
				\begin{scope}[xshift=\mod cm,yshift=-\y cm]
				
					\draw (0.5,-0.5) node[inner sep=0pt, outer sep=0pt] (value\arabic{curvaltwo}southeast) {};
					\draw (0.5,0.5) node[inner sep=0pt, outer sep=0pt] (value\arabic{curvaltwo}northeast) {};
					\draw (-0.5,-0.5) node[inner sep=0pt, outer sep=0pt] (value\arabic{curvaltwo}southwest) {};
					\draw (-0.5,0.5) node[inner sep=0pt, outer sep=0pt] (value\arabic{curvaltwo}northwest) {};
					\draw (0,0) node[inner sep=0pt, outer sep=0pt] (value\arabic{curvaltwo}center) {};
					\draw (0,0.5) node[inner sep=0pt, outer sep=0pt] (value\arabic{curvaltwo}northcenter) {};
				
						\draw (-0.5,-0.5) -- (0.5,-0.5) -- (0.5,0.5) -- (-0.5,0.5) -- cycle;
					\draw (0,0) node[anchor=center,text centered,text width=0.8cm] {$\content$};
					\draw (-0.5,0.5) node[anchor=north west, gray] {\scriptsize\arabic{curvaltwo}};
				\end{scope}
				
				\stepcounter{curvaltwo}
			}
			
			\begin{scope}[pattern color=gray]
				\filldraw[draw=black,fill=lightgray,pattern=north west lines] (value1northwest) -- (value11northeast) -- ([xshift=0pt,yshift=\hei cm]value11northeast.center) -- ([xshift=0pt,yshift=\hei cm]value1northwest.center) -- (value1northwest);
				\draw ([yshift=\halfhei cm]value6northcenter) node[anchor=center, black,fill=white, inner sep=2pt] {\footnotesize super Huffman coding};
			
				\filldraw[draw=black,fill=lightgray,pattern=dots] (value12northwest) -- (value15northeast) -- ([xshift=0pt,yshift=\hei cm]value15northeast.center) -- ([xshift=0pt,yshift=\hei cm]value12northwest.center) -- (value12northwest);
				\filldraw[draw=black,fill=lightgray,pattern=dots] (value16northwest) -- (value22northeast) -- ([xshift=0pt,yshift=\hei cm]value22northeast.center) -- ([xshift=0pt,yshift=\hei cm]value16northwest.center) -- (value16northwest);
				\draw ([yshift=\halfhei cm]value19northcenter) node[anchor=center,black,fill=white,inner sep=2pt] {\footnotesize base Huffman coding $C_1$};
				
				\filldraw[draw=black,fill=lightgray,pattern=north east lines] (value23northwest) -- (value31northeast) -- ([xshift=0pt,yshift=\hei cm]value31northeast.center) -- ([xshift=0pt,yshift=\hei cm]value23northwest.center) -- (value23northwest);
				\filldraw[draw=black,fill=lightgray,pattern=north east lines] (value32northwest) -- (value32northeast) -- ([xshift=0pt,yshift=\hei cm]value32northeast.center) -- ([xshift=0pt,yshift=\hei cm]value32northwest.center) -- (value32northwest);
				\draw ([yshift=\halfhei cm]value27northcenter) node[anchor=center,black,fill=white,inner sep=2pt] {\footnotesize base Huffman coding $C_2$};
				
				\filldraw[draw=black,fill=lightgray,pattern=checkerboard light gray] (value33northwest) -- (value47northeast) -- ([xshift=0pt,yshift=\hei cm]value47northeast.center) -- ([xshift=0pt,yshift=\hei cm]value33northwest.center) -- (value33northwest);
				\filldraw[draw=black,fill=lightgray,pattern=checkerboard light gray] (value48northwest) -- (value60northeast) -- ([xshift=0pt,yshift=\hei cm]value60northeast.center) -- ([xshift=0pt,yshift=\hei cm]value48northwest.center) -- (value48northwest);
				\draw ([xshift=0cm,yshift=\halfhei cm]value40northwest) node[anchor=center, black,fill=white, inner sep=2pt] {\footnotesize base Huffman coding $C_3$};
			\end{scope}
			
			\begin{pgfonlayer}{background}
				\foreach \l/\r/\lc/\rc in {1/1/lightgray/lightgray,12/12/lightgray/lightgray,23/23/lightgray/lightgray,33/33/lightgray/lightgray} {
					\shade[left color=\lc,right color=\rc] (value\l northwest.center) rectangle (value\r southeast.center);
				}
			\end{pgfonlayer}
		\end{tikzpicture}
	\caption{Depiction of the part of the output file which contains the serialized four Huffman codings.}\label{fig:mainExampleFirstPartOutputFile}
\end{figure}

The super Huffman coding listed in Table \ref{tbl:mainExampleSuperHuffmanCoding} is written to the output file (values 2--11) using 3 bits per integer as it is stated by the value 0 of the output file. There need to be enumerated $10$ super code lengths since $10$ values --- the base code lengths $0,1,\ldots,6$ and the run indicators $7,8,9$ which are used by the base Huffman coding $C_3$ --- need to be encoded.

\section{Experimental Results}\label{ch:experimentalResults}

In the following, we compare the compression performance of our implementation of the Re-pair for Trees algorithm with existing algorithms. Furthermore, we will check the impact of the DAG representation of the input tree on the compression factors achieved and we will learn about the influences of small changes to the maximal rank allowed for a nonterminal.

\subsection{XML Documents Used}

The set of XML documents we used for investigating the performance of \trp consists of 23 files with different characteristics (\cf Table \ref{tbl:characteristicsXmlDocuments}). Most of them were used in past papers evaluating various XML compressors and therefore may be familiar to the reader. The original files can be obtained from the sources listed in Table \ref{tbl:sourcesOfXmlDocuments}. In all cases character data, attributes, comments, namespace information were removed from the XML files, \ie, the XML documents consist only of start tags, end tags and empty element tags. We do so, because, at this time, \trp ignores this information and solely concentrates on the XML document tree.

\begin{table}[t]
        \centering
		\begin{tabular}{lrrrrrr}\toprule
		XML document						&File size (kb)&\#\,Edges	&Depth	&\#\,Element types	&Source\\
		\midrule
		\texttt{1998statistics}			&349			&28\,305		&5		&46				&1\\
		\texttt{catalog-01}				&4\,219		&225\,193		&7		&50				&9\\
		\texttt{catalog-02}				&44\,656		&2\,390\,230	&7		&53				&9\\
		\texttt{dictionary-01}				&1\,737		&277\,071		&7		&24				&9\\
		\texttt{dictionary-02}				&17\,128		&2\,731\,763	&7		&24				&9\\
		\texttt{dblp}						&117\,822		&10\,802\,123	&5		&35				&2\\
		\texttt{EnWikiNew}				&4\,843		&404\,651		&4		&20				&3\\
		\texttt{EnWikiQuote}				&3\,134		&262\,954		&4		&20				&3\\
		\texttt{EnWikiSource}			&13\,457		&1\,133\,534	&4		&20				&3\\
		\texttt{EnWikiVersity}			&5\,887		&495\,838		&4		&20				&3\\
		\texttt{EnWikTionary}			&99\,201		&8\,385\,133	&4		&20				&3\\
		\texttt{EXI-Array}				&5\,347		&226\,522		&9		&47				&5\\
		\texttt{EXI-factbook}			&1\,214		&55\,452		&4		&199				&5\\
		\texttt{EXI-Invoice}				&266			&15\,074		&6		&52				&5\\
		\texttt{EXI-Telecomp}			&3\,700		&177\,633		&6		&39				&5\\
		\texttt{EXI-weblog}				&1\,104		&93\,434		&2		&12				&5\\
		\texttt{JST\_gene.chr1}			&4\,202		&216\,400		&6		&26				&8\\
		\texttt{JST\_snp.chr1}				&13\,795		&655\,945		&7		&42				&8\\
		\texttt{medline02n0328}			&51\,751		&2\,866\,079	&6		&78				&6\\
		\texttt{NCBI\_gene.chr1}			&6\,862		&360\,349		&6		&50				&8\\
		\texttt{NCBI\_snp.chr1}			&63\,941		&3\,642\,224	&3		&15				&8\\
		\texttt{sprot39.dat}				&111\,175		&10\,903\,567	&5		&48				&7\\
		\texttt{treebank}				&19\,551		&2\,447\,726	&36		&251				&4\\
		\bottomrule
		\end{tabular}
	\caption{Characteristics of the XML documents used in our tests. The values in the "Source"-column match the source IDs in Table \ref{tbl:sourcesOfXmlDocuments}. The depth of an XML document tree specifies the length (number of edges) of the longest path from the root of the tree to a leaf.}
	\label{tbl:characteristicsXmlDocuments}
\end{table}

\begin{table}[tb]
	\footnotesize\centering
	\begin{tabular}{cl}
		\toprule
		ID&Source\\
		\midrule
		1&\url{http://www.cafeconleche.org/examples}\\
		2&\url{http://dblp.uni-trier.de/xml}\\
		3&\url{http://download.wikipedia.org/backup-index.html}\\
		4&\url{http://www.cs.washington.edu/research/xmldatasets}\\
		5&\url{http://www.w3.org/XML/EXI}\\
		6&\url{http://www.ncbi.nlm.nih.gov/pubmed}\\
		7&\url{http://expasy.org/sprot}\\
		8&\url{http://snp.ims.u-tokyo.ac.jp}\\
		9&\url{http://softbase.uwaterloo.ca/~ddbms/projects/xbench}\\
		\bottomrule
	\end{tabular}
	\caption{Sources of the XML documents from Table \ref{tbl:characteristicsXmlDocuments}.}
	\label{tbl:sourcesOfXmlDocuments}
\end{table}

\subsection{Algorithms Used in Comparison}

Basically, we compare our implementation of Re-pair for Trees with two
other compression algorithms based on linear SLCF tree grammars,
namely, BPLEX \cite{Busatto08efficient} and Extended-Repair
\cite{Krislin08repair,Boettcher10clustering}. The former is a
sliding-window based linear time approximation algorithm. It searches
bottom-up in a fixed window for repeating tree patterns. The size of
the sliding window, the maximal pattern size and the maximal rank of a
nonterminal can be specified as input parameters. 
One of the main drawbacks of BPLEX is that there exists only a slowly running implementation of it.

Extended-Repair (which we sometimes call E-Repair in the
sequel) is an algorithm developed by a group from 
the University of Paderborn, Germany 
\cite{Krislin08repair,Boettcher10clustering}. This algorithm is, just
like our Re-pair for Trees algorithm, based on the Re-pair algorithm
introduced in \cite{larsson2000off}. However, it was independently
developed and exhibits some fundamental differences to our
algorithm. One of the main differences is that the Extended-Repair
algorithm at first generates a DAG of the input tree and then
processes each part of it individually, \ie, it generates multiple
grammars which are combined in the end. The individual parts of the
input tree are called "repair packets". The maximal size of each
packet can be specified by an input parameter (default is 20\,000
edges). The author of \cite{Krislin08repair} points out that this
packet-based behavior may have a negative impact on the compression
performance of the Extended-Repair algorithm. Our own investigations
concerning a \trp version running on the DAG of the input tree instead 
of on the whole tree support this point of view. 

In \cite{Krislin08repair} it is shown that Extended-Repair achieves a much better
compression ratio on the XML document \texttt{NCBI\_snp.chr1},
when the input tree is not broken down into packets (this can 
be achieved by choosing the maximum packet size large enough). 
However, our experiments show that at
the same time the memory requirements and the runtime of the
Extended-Repair algorithm rise drastically. Note that, regarding our
algorithm, the DAG representation is merely used to save memory
resources and is almost completely 
transparent to the overlying \tp replacement process 
(\cf Sect.~\ref{sec:impactDag} on page \pageref{sec:impactDag}). 

\subsection{Testing Environment}

Our experiments were done on a computer with an Intel\textregistered\space Core\texttrademark\space2 Duo CPU T9400 processor, four gigabytes of RAM and the Linux operating system. Every algorithm was executed on a single processor core, \ie, no algorithm was able to make use of multiprocessing. \trp and BPLEX were compiled with the \texttt{gcc}-compiler using the \texttt{-O3} (compile time optimizations) and \texttt{-m32} (\ie, we generated them as 32bit-applications) switches. We were not able to compile the \texttt{succ}-tool of the BPLEX distribution with compile time optimizations (\ie, using the \texttt{-O3} switch). This tool is used to apply a succinct coding to a grammar generated by the BPLEX algorithm.  However, this should not have a great influence on the runtime measured for BPLEX since the \texttt{succ}-tool usually executes quite fast compared to the runtime of the actual BPLEX algorithm. In contrast, Extended-Repair is an application written in Java\texttrademark\space for which we only had the bytecode at hand, \ie, we did not have access to the source code of it. We executed Extended-Repair using the Java SE Runtime Environment\texttrademark\space in version \texttt{1.6.0\_15}.

During the execution of the algorithms we always measured their memory usage. We accomplished this by constantly polling the \texttt{VmRSS}-value which is printed out by executing the command \texttt{cat /proc/<pid>/status}, where \texttt{<pid>} is the process ID assigned to the algorithm. In the first second of the execution of an algorithm this value was checked every ten milliseconds and after that the frequency was slowly reduced to one second.

Every time we executed BPLEX we used its default input parameters, namely, window size: 20\,000, maximal pattern size: 20, maximal rank: 10. In order to be able to test BPLEX together with every file of our set of test XML documents we needed to explicitly allow large stack sizes using the standard tool \texttt{ulimit}.

\begin{table}[tb]
	\centering\small
	\begin{tabular}{lrrrrr}
		\toprule
		&\trp&BPLEX&E-Repair&mDAG&bin. mDAG\\
		\midrule
		Edges (\%)&2.9&3.4&4.1&12.8&18.3\\
		\#\,NTs&4\,753&13\,660&6\,522&2\,075&5\,320\\
		\midrule
		Time (sec)&10&322&63&-&-\\
		Mem (MB)&47&536&401&-&-\\
		File size (\%)&0.46&0.71&0.61&-&-\\
		\bottomrule
	\end{tabular}
	\caption{Average values of the characteristics of the generated grammars and of the corresponding runs of the algorithms.}
	\label{tbl:avgValuesOfCharacteristicsOptimizeEdges}
\end{table}

\subsection{Comparison of the Generated Grammars}

In this section, we compare the final grammars generated by the algorithms \trp, BPLEX and Extended-Repair. All algorithms were instructed to minimize the number of edges of the generated grammar. For \trp, we achieved this behavior by specifying the \texttt{-optimize edges} input parameter. Regarding Extended-Repair, we used the supplied \texttt{ConfEdges.xml} configuration file which is supposed to make Extended-Repair minimize the number of edges. The BPLEX algorithm was executed with its default input parameter values and no changes were made to the generated grammar (besides pruning nonterminals which are referenced only once by using the supplied \texttt{gprint} tool).

Table \ref{tbl:avgValuesOfCharacteristicsOptimizeEdges} shows the average values of the essential characteristics of the final grammars generated by the three competing algorithms. The first row shows the average compression factors in terms of the number of edges in percent. The edge compression factor is computed as follows: if $t\in T(\mathcal{F})$ is the binary representation of the input tree and $\mathcal{G}$ is the final grammar, we obtain the edge compression factor by computing $\nicefrac{|\mathcal{G}|}{|t|}\cdot 100$. The second row shows the average number of nonterminals of the final grammars. For the sake of completeness, the average runtimes (in seconds), the average memory usages (in megabytes) and the average file size compression factors are also listed. The compression factor in terms of file size specifies the ratio between the size of the input file and the file size of the succinct coding of the final grammar in percent.

We also added two columns to Table \ref{tbl:avgValuesOfCharacteristicsOptimizeEdges} showing the average number of edges and the average number of nonterminals of the minimal DAGs of the input trees (mDAG) and the minimal DAGs of the binary representations of the input trees (bin. mDAG).

As it can be seen, on average, \trp generates the smallest linear SLCF tree grammars (in terms of the number of edges) compared to the other two algorithms. At the same time, its grammars exhibit a small number of nonterminals. It outperforms BPLEX and Extended-Repair in terms of runtime and memory usage. The speed and moderate requirements on main memory are a result of the transparent DAG representation of the input tree and the many optimizations we made to the source code of \trp during our investigations.

\begin{table}[tb]
	\centering\small
	\begin{tabular}{lrrrrrr}
		\toprule
		&\trp&BPLEX&E-Repair&XMill&gzip&bzip2\\
		\midrule
		File size (\%)&0.45&0.57&0.61&0.47&1.36&0.58\\
		\midrule
		Time (sec)&10&329&167&119&$<$\,1&16\\
		Mem (MB)&47&536&399&7&-&7\\
		Edges (\%)&3.0&3.9&4.1&-&-&-\\
		\#\,NTs&2\,642&2\,796&7\,003&-&-&-\\
		\bottomrule
	\end{tabular}
	\caption{Average values of the characteristics of the runs of the three algorithms when making a small size of the output file top priority.}\label{tbl:avgValuesOfCharacteristicsOptimizeFileSize}
\end{table}

Figure \ref{fig:diagramEdgesCompression} on page \pageref{fig:diagramEdgesCompression} gives an impression on how each of the three algorithms performs on the individual XML documents in terms of the size of the final grammar in edges. For each file, the algorithm which generates the largest grammar is set to 100\%. In Appendix \ref{sec:detailedResultsOptimizationEdges} on page \pageref{sec:detailedResultsOptimizationEdges} there is a detailed table listing all relevant characteristics of the runs of the algorithms on the set of test XML documents.

\subsection{Comparison of Output File Sizes}

In this section, we concentrate on the sizes of the files generated by the runs of the algorithms on our set of test XML documents. In fact, we execute each algorithm in a mode in which the size of the resulting file is made a top priority. For \trp, we achieve this by specifying the input parameter \texttt{-optimize filesize} and for Extended-Repair, we get such a behavior by using the supplied \texttt{ConfSize.xml} configuration file and the \texttt{-s 4} switch. The latter chooses a certain succinct coding of the Extended-Repair distribution which is supposed to generate very small representations of the generated grammar. Regarding BPLEX, we first apply the supplied \texttt{gprint}-tool using the parameters \texttt{-\xspace-prune} and \texttt{-\xspace-threshold 14}. After that we use the \texttt{succ}-tool of the BPLEX distribution together with the parameter \texttt{-\xspace-type 68} to generate a Huffman coding-based succinct coding of the corresponding grammar. In \cite{Maneth08xml} it is stated that this approach leads to the best compression performance of BPLEX in general (in terms of file size).

In addition to the above three algorithms, we also consider the compression results produced by gzip, bzip2\footnote{For more information about the gzip algorithm, see \url{http://www.gzip.org}. For bzip2, see \url{http://www.bzip.org}.} and XMill 0.8 \cite{liefke2000xmill}. We include them in our comparison to make it easier to get a handle for common compression rates and runtimes. The first two algorithms are widely used general purpose file compressors which, of course, produce a non-queryable compressed representation of the input file. In contrast, XMill is a compressor specialized in compressing the structure and, in particular, the character data of XML documents. In fact, it mainly concentrates on how to group the character data of an XML document in such a way that it can be efficiently compressed by general purpose compressors like gzip. Since its implementation does not exhibit a special "only consider the structure of the XML document" mode, it may be unfair to directly compare its compression results with those of \trp, BPLEX or Extended-Repair. However, we included its compression results, which we obtained using its default input parameters, because we were interested in its performance in this setting.

Table \ref{tbl:avgValuesOfCharacteristicsOptimizeFileSize} shows the average sizes of the output files generated by the six algorithms mentioned above. For the sake of completeness, the average runtime, the average memory usage, the average number of edges and the average number of nonterminals are also listed. Again, \trp outperforms BPLEX and Extended-Repair regarding all considered characteristics. Surprisingly, its queryable output files are even smaller than the non-queryable ones produced by the highly optimized gzip and bzip2 algorithms. However, gzip (but interestingly not bzip2) runs much faster than \trp on our test data.

Figure \ref{fig:diagramFileSizeCompression} gives an impression on how each of the six algorithms performs on the individual XML documents in terms of the size of the generated output file. For each file, the algorithm which generates the biggest output file is set to 100\%. In Appendix \ref{sec:detailedResultsOptimizationFileSize} on page \pageref{sec:detailedResultsOptimizationFileSize} there is a detailed table listing all relevant characteristics of the runs of the algorithms on our set of test XML documents.

\begin{figure}[p]
	\subfigure[{Comparison of the number of edges of the final grammars.}]{
		\centering
		\includegraphics{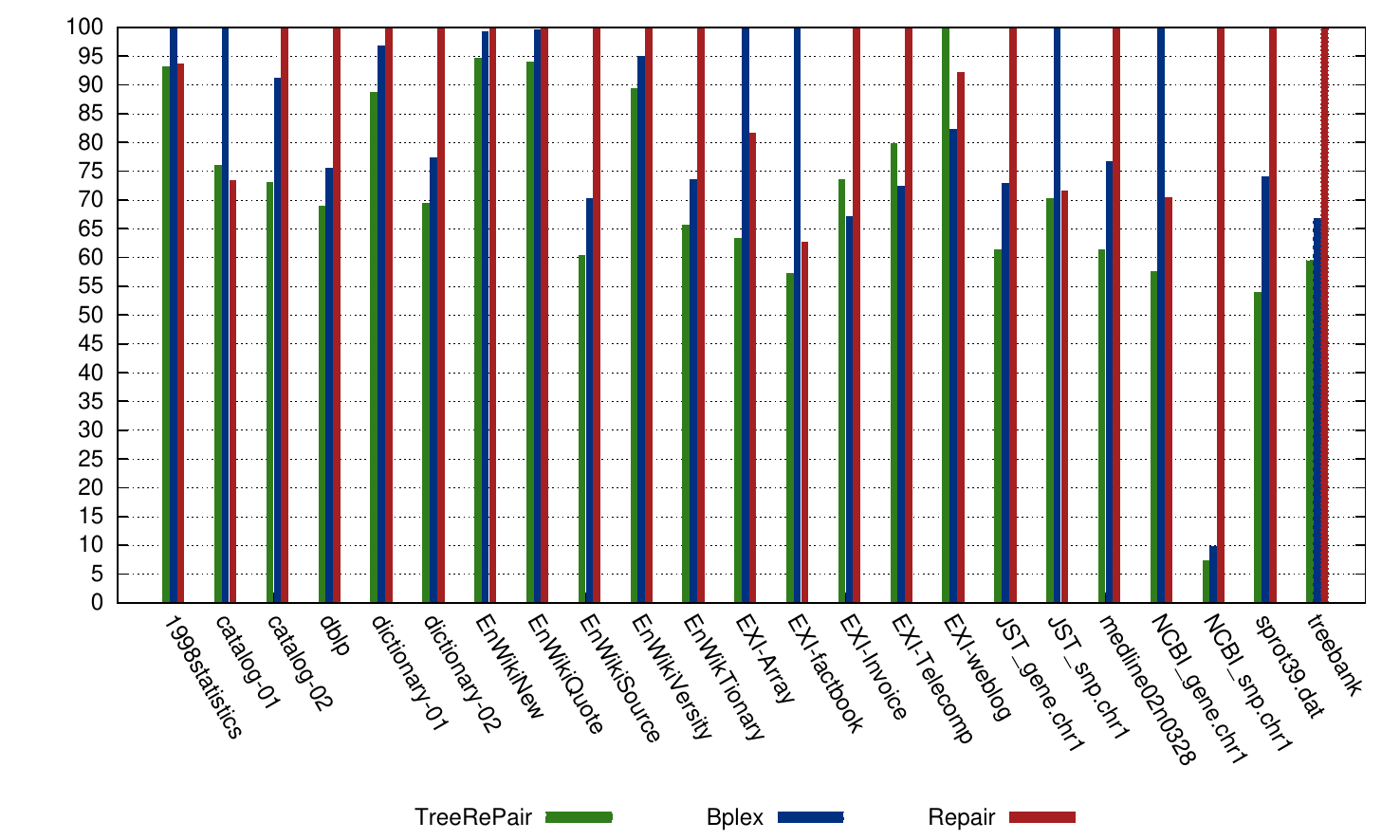}
		\label{fig:diagramEdgesCompression}
	}
	\subfigure[{Comparison of the sizes of the output files.}]{
		\centering
		\includegraphics{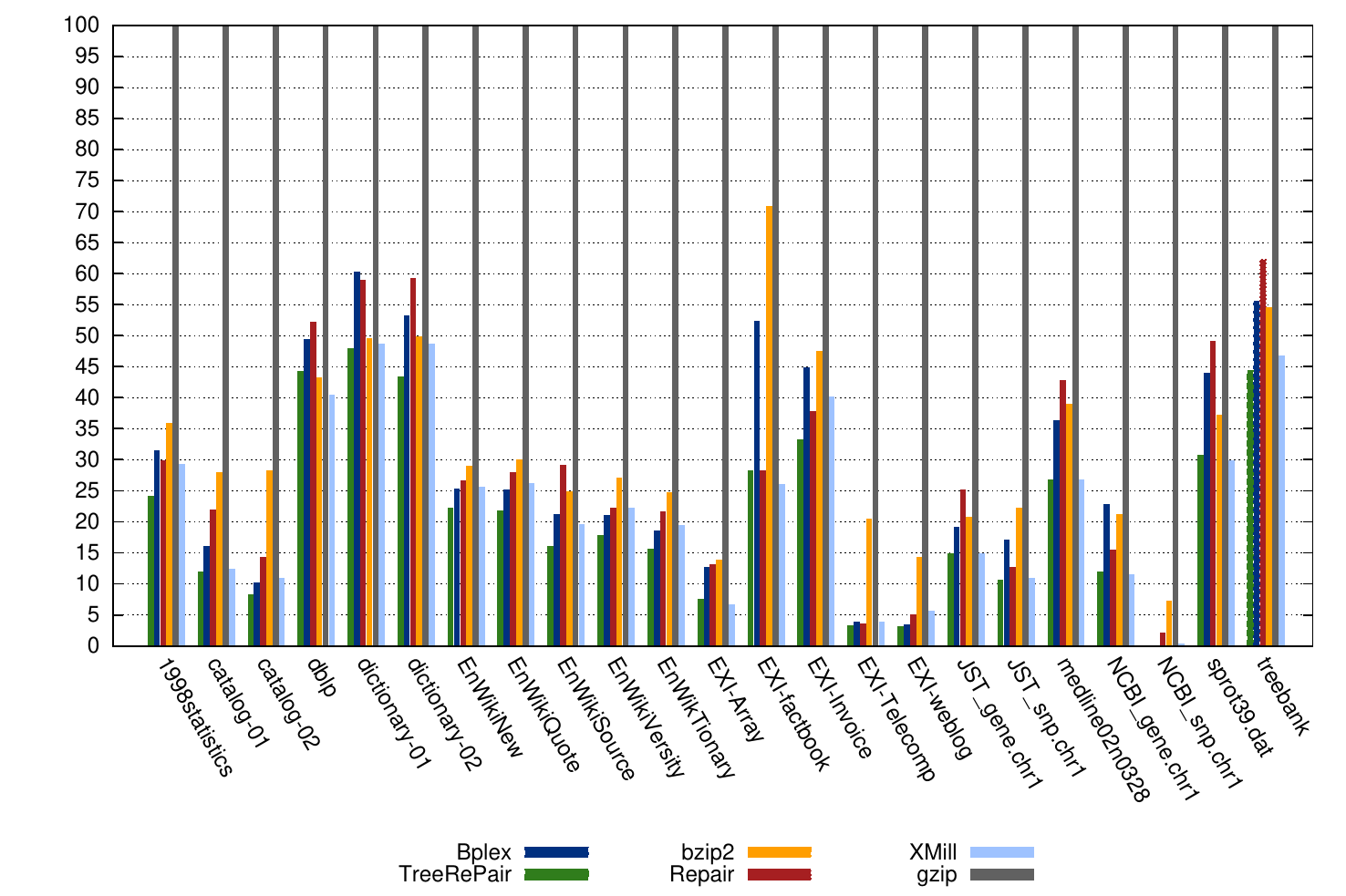}
		\label{fig:diagramFileSizeCompression}
	}
\end{figure}

\subsection{Results without DAG Representation}\label{sec:resultsWithoutDag}

Table \ref{tbl:avgValuesOfCharacteristicsNoDag} shows a comparison between the compression results of \trp when using and when not using, respectively, the DAG representation described in Sect.~\ref{sec:representingTreeInMemory} on page \pageref{sec:representingTreeInMemory}. The left column shows the values obtained when executing \trp with its default parameters in edge optimization mode, \ie, we are only using the \texttt{-optimize edges} switch since our algorithm uses the DAG representation by default. In contrast, the right column is a result of running \trp with the \texttt{-no\_dag} and \texttt{-optimize edges} switches.  Again, in Appendix \ref{sec:detailedResultsNoDag} on page \pageref{sec:detailedResultsNoDag}, there is a detailed table listing all relevant characteristics of the runs of the two \trp configurations on each test XML document.

\begin{table}[tb]
	\centering\small
	\begin{tabular}{lrr}
		\toprule
		&with DAG&without DAG\\
		\midrule
		Edges (\%)&2.86&2.84\\
		\#\,NTs&4\,753&4\,620\\
		File size (\%)&0.463&0.459\\
		Time (sec)&9.8&11.2\\
		Mem (MB)&47&188\\
		\bottomrule
	\end{tabular}
	\caption{Average values of the characteristics of the runs of \trp with and without the DAG representation of the input tree.}\label{tbl:avgValuesOfCharacteristicsNoDag}
\end{table}

Regarding the differences between the compression results of \trp and the ones of the competing algorithms, it can be said that the DAG representation only has a minor impact on the compression performance of our algorithm. However, we can state that it drastically reduces the memory demands of \trp\space--- it slashes the memory consumption by a factor of 4. Interestingly, even without the DAG representation, \trp uses only half as much main memory as Extended-Repair does (\cf Table~\ref{tbl:avgValuesOfCharacteristicsOptimizeEdges}). Furthermore, the DAG representation leads to a faster compression speed since it saves repetitive recalculations concerning equal subtrees.

\subsection{Results with Different Maximal Ranks}\label{sec:resultsWithDifferentMaximalRanks}

\begin{table}[tb]
	\centering\small
	\begin{tabular}{lrrrrrrr}
		\toprule
		Max. rank&0&1&2&3&4&5&6\\
		\midrule
		Edges (\%)&55.02&3.29&2.92&2.89&2.86&2.89&2.89\\
		\#\,NTs&1\,265&5\,539&4\,712&4\,916&4\,753&4\,956&4\,958\\
		File size (\%)&2.12&0.51&0.47&0.47&0.46&0.47&0.46\\
		Time (sec)&7.0&8.4&9.3&9.5&9.6&9.8&9.8\\
		Mem (MB)&44&44&45&47&47&47&47\\
		\bottomrule
	\end{tabular}
	\caption{Average values of the characteristics of the runs of \trp with different maximal ranks allowed for a nonterminal.}\label{tbl:avgValuesOfCharacteristicsRanks}
\end{table}

We executed TreeRePair using the \texttt{-optimize edges} (\ie, we
enabled the edge optimization mode) and the \texttt{-max\_rank}
switches. Each time, we specified a different maximal rank for a
nonterminal in order to get information concerning its influence on 
the compression performance. Table~\ref{tbl:avgValuesOfCharacteristicsRanks} shows that, regarding our
set of test XML documents, a maximal rank 
of 4 leads to the best compression results on average. 

At the same time, we can see that even when restricting the maximal
rank to 1 \trp performs better than BPLEX and Extended-Repair (\cf
Table~\ref{tbl:avgValuesOfCharacteristicsOptimizeEdges}). The fact
that large maximal ranks can lead to a worse compression ratio can be
explained by the trees from Sect.~\ref{sec:limitingTheMaximalRank} on
page \pageref{sec:limitingTheMaximalRank}. Note that the trees from
this section are basically long lists. Although this is not the case
for our test trees, their shape is nevertheless similar to a list
structure. In any case, its quite distinct from the shape of a full
binary tree, where an unlimited maximal rank leads to 
the best compression ratio (\cf Sect.~\ref{sec:observations} on page \pageref{sec:observations}).


\newcommand{\etalchar}[1]{$^{#1}$}

\newpage

\appendix

\section{Detailed Test Results}\label{ch:detailedResults}

\subsection{Optimization of Total Number of Edges}\label{sec:detailedResultsOptimizationEdges}

\begin{longtable}{lrrrrr}
			\toprule
\textbf{Algorithm}&\textbf{Edges}&\textbf{File size}&\textbf{\#NTs}&\textbf{Time}&\textbf{Mem (MB)}\\\midrule
\endhead
			\toprule
\textbf{Algorithm}&\textbf{Edges}&\textbf{File size}&\textbf{\#NTs}&\textbf{Time}&\textbf{Mem (MB)}\\
\endfirsthead
			\midrule\multicolumn{6}{c}{\emph{1998statistics}}\\*
			TreeRePair&1.68\%&0.20\%&54&100ms&1\\*
			BPLEX&1.80\%&0.34\%&168&1.813s&295\\*
			E-Repair&1.69\%&0.24\%&37&7.518s&114\\*
			bin. mDAG&8.49\%&-&31&-&-\\*
			mDAG&4.87\%&-&15&-&-\\
			\midrule\multicolumn{6}{c}{\emph{catalog-01}}\\*
			TreeRePair&1.69\%&0.10\%&400&887ms&2\\*
			BPLEX&2.22\%&0.22\%&1251&6.548s&315\\*
			E-Repair&1.63\%&0.12\%&291&9.975s&279\\*
			bin. mDAG&3.10\%&-&520&-&-\\*
			mDAG&3.80\%&-&506&-&-\\
			\midrule\multicolumn{6}{c}{\emph{catalog-02}}\\*
			TreeRePair&1.11\%&0.07\%&965&9.409s&10\\*
			BPLEX&1.38\%&0.11\%&3045&30s&512\\*
			E-Repair&1.52\%&0.11\%&1499&42s&511\\*
			bin. mDAG&2.22\%&-&805&-&-\\*
			mDAG&1.39\%&-&792&-&-\\
			\midrule\multicolumn{6}{c}{\emph{dblp}}\\*
			TreeRePair&3.89\%&0.59\%&25250&43s&227\\*
			BPLEX&4.27\%&0.73\%&38712&57m 42s&1644\\*
			E-Repair&5.65\%&0.68\%&30430&4m 34s&510\\*
			bin. mDAG&19.36\%&-&6592&-&-\\*
			mDAG&11.11\%&-&3378&-&-\\
			\midrule\multicolumn{6}{c}{\emph{dictionary-01}}\\*
			TreeRePair&7.72\%&1.54\%&1676&1.010s&9\\*
			BPLEX&8.43\%&2.37\%&3994&44s&323\\*
			E-Repair&8.71\%&1.83\%&1248&16s&433\\*
			bin. mDAG&27.99\%&-&2058&-&-\\*
			mDAG&21.07\%&-&448&-&-\\
			\midrule\multicolumn{6}{c}{\emph{dictionary-02}}\\*
			TreeRePair&5.92\%&1.38\%&9757&11s&69\\*
			BPLEX&6.58\%&1.95\%&23209&6m 12s&587\\*
			E-Repair&8.52\%&1.83\%&11672&1m 40s&494\\*
			bin. mDAG&24.93\%&-&16281&-&-\\*
			mDAG&19.96\%&-&2414&-&-\\
			\midrule\multicolumn{6}{c}{\emph{EnWikiNew}}\\*
			TreeRePair&2.29\%&0.21\%&667&1.585s&8\\*
			BPLEX&2.40\%&0.30\%&1369&35s&337\\*
			E-Repair&2.42\%&0.24\%&476&12s&347\\*
			bin. mDAG&17.31\%&-&23&-&-\\*
			mDAG&8.67\%&-&29&-&-\\
			\midrule\multicolumn{6}{c}{\emph{EnWikiQuote}}\\*
			TreeRePair&2.42\%&0.21\%&452&1.158s&7\\*
			BPLEX&2.56\%&0.31\%&985&25s&321\\*
			E-Repair&2.58\%&0.26\%&323&9.924s&290\\*
			bin. mDAG&18.14\%&-&19&-&-\\*
			mDAG&9.09\%&-&25&-&-\\
			\midrule\multicolumn{6}{c}{\emph{EnWikiSource}}\\*
			TreeRePair&1.10\%&0.10\%&861&4.927s&26\\*
			BPLEX&1.28\%&0.16\%&1895&1m 9s&418\\*
			E-Repair&1.82\%&0.18\%&1106&23s&500\\*
			bin. mDAG&17.52\%&-&19&-&-\\*
			mDAG&8.77\%&-&24&-&-\\
			\midrule\multicolumn{6}{c}{\emph{EnWikiVersity}}\\*
			TreeRePair&1.44\%&0.13\%&525&2.107s&12\\*
			BPLEX&1.53\%&0.18\%&1043&34s&347\\*
			E-Repair&1.61\%&0.15\%&423&12s&437\\*
			bin. mDAG&17.60\%&-&19&-&-\\*
			mDAG&8.81\%&-&24&-&-\\
			\midrule\multicolumn{6}{c}{\emph{EnWikTionary}}\\*
			TreeRePair&0.97\%&0.11\%&4535&36s&183\\*
			BPLEX&1.09\%&0.14\%&6402&8m 58s&1287\\*
			E-Repair&1.48\%&0.15\%&6315&1m 33s&540\\*
			bin. mDAG&17.32\%&-&26&-&-\\*
			mDAG&8.66\%&-&30&-&-\\
			\midrule\multicolumn{6}{c}{\emph{EXI-Array}}\\*
			TreeRePair&0.41\%&0.03\%&123&1.281s&14\\*
			BPLEX&0.65\%&0.06\%&383&42s&322\\*
			E-Repair&0.53\%&0.05\%&142&8.017s&320\\*
			bin. mDAG&56.51\%&-&8&-&-\\*
			mDAG&42.20\%&-&13&-&-\\
			\midrule\multicolumn{6}{c}{\emph{EXI-factbook}}\\*
			TreeRePair&2.35\%&0.31\%&145&271ms&2\\*
			BPLEX&4.11\%&0.77\%&1423&5.138s&298\\*
			E-Repair&2.58\%&0.31\%&146&11s&408\\*
			bin. mDAG&9.16\%&-&236&-&-\\*
			mDAG&8.07\%&-&293&-&-\\
			\midrule\multicolumn{6}{c}{\emph{EXI-Invoice}}\\*
			TreeRePair&0.68\%&0.21\%&14&74ms&1\\*
			BPLEX&0.62\%&0.30\%&40&1.483s&293\\*
			E-Repair&0.93\%&0.24\%&20&4.689s&119\\*
			bin. mDAG&13.74\%&-&6&-&-\\*
			mDAG&7.12\%&-&15&-&-\\
			\midrule\multicolumn{6}{c}{\emph{EXI-Telecomp}}\\*
			TreeRePair&0.07\%&0.01\%&21&780ms&3\\*
			BPLEX&0.06\%&0.02\%&47&9.684s&310\\*
			E-Repair&0.08\%&0.02\%&21&11s&452\\*
			bin. mDAG&11.15\%&-&10&-&-\\*
			mDAG&5.59\%&-&15&-&-\\
			\midrule\multicolumn{6}{c}{\emph{EXI-weblog}}\\*
			TreeRePair&0.06\%&0.01\%&13&324ms&3\\*
			BPLEX&0.04\%&0.01\%&24&9.097s&303\\*
			E-Repair&0.05\%&0.02\%&11&7.868s&279\\*
			bin. mDAG&18.19\%&-&2&-&-\\*
			mDAG&9.10\%&-&2&-&-\\
			\midrule\multicolumn{6}{c}{\emph{JST\_gene.chr1}}\\*
			TreeRePair&1.84\%&0.10\%&354&874ms&3\\*
			BPLEX&2.19\%&0.19\%&1113&11s&315\\*
			E-Repair&2.99\%&0.17\%&126&8.006s&233\\*
			bin. mDAG&6.75\%&-&114&-&-\\*
			mDAG&4.24\%&-&76&-&-\\
			\midrule\multicolumn{6}{c}{\emph{JST\_snp.chr1}}\\*
			TreeRePair&1.51\%&0.09\%&856&3.150s&8\\*
			BPLEX&2.15\%&0.21\%&4193&31s&360\\*
			E-Repair&1.54\%&0.10\%&634&15s&445\\*
			bin. mDAG&6.20\%&-&282&-&-\\*
			mDAG&3.59\%&-&242&-&-\\
			\midrule\multicolumn{6}{c}{\emph{medline02n0328}}\\*
			TreeRePair&4.13\%&0.35\%&9064&16s&79\\*
			BPLEX&5.17\%&0.62\%&33976&5m 52s&574\\*
			E-Repair&6.73\%&0.54\%&13010&1m 32s&479\\*
			bin. mDAG&25.84\%&-&20013&-&-\\*
			mDAG&22.80\%&-&3960&-&-\\
			\midrule\multicolumn{6}{c}{\emph{NCBI\_gene.chr1}}\\*
			TreeRePair&1.37\%&0.09\%&504&1.374s&4\\*
			BPLEX&2.38\%&0.28\%&3631&14s&327\\*
			E-Repair&1.68\%&0.11\%&328&10s&308\\*
			bin. mDAG&3.98\%&-&605&-&-\\*
			mDAG&4.45\%&-&436&-&-\\
			\midrule\multicolumn{6}{c}{\emph{NCBI\_snp.chr1}}\\*
			TreeRePair&$<$\,0.01\%&$<$\,0.01\%&17&15s&80\\*
			BPLEX&$<$\,0.01\%&$<$\,0.01\%&23&2m 6s&770\\*
			E-Repair&0.03\%&0.01\%&291&37s&504\\*
			bin. mDAG&22.22\%&-&2&-&-\\*
			mDAG&11.11\%&-&2&-&-\\
			\midrule\multicolumn{6}{c}{\emph{sprot39.dat}}\\*
			TreeRePair&2.30\%&0.38\%&20224&43s&178\\*
			BPLEX&3.16\%&0.79\%&111167&14m 41s&1446\\*
			E-Repair&4.27\%&0.59\%&33102&3m 48s&499\\*
			bin. mDAG&13.18\%&-&31116&-&-\\*
			mDAG&16.07\%&-&10243&-&-\\
			\midrule\multicolumn{6}{c}{\emph{treebank}}\\*
			TreeRePair&20.72\%&4.41\%&32857&22s&164\\*
			BPLEX&23.29\%&6.16\%&76109&21m 27s&645\\*
			E-Repair&34.85\%&6.03\%&48358&6m 50s&526\\*
			bin. mDAG&59.42\%&-&43586&-&-\\*
			mDAG&53.75\%&-&24746&-&-\\
			\bottomrule
		\end{longtable}

\newpage

\subsection{Optimization of File Size}\label{sec:detailedResultsOptimizationFileSize}

\begin{longtable}{lrrrrr}
			\toprule
\textbf{Algorithm}&\textbf{Edges}&\textbf{File size}&\textbf{\#NTs}&\textbf{Time}&\textbf{Mem (MB)}\\\midrule
\endhead
			\toprule
\textbf{Algorithm}&\textbf{Edges}&\textbf{File size}&\textbf{\#NTs}&\textbf{Time}&\textbf{Mem (MB)}\\
\endfirsthead
			\midrule\multicolumn{6}{c}{\emph{1998statistics}}\\*
			TreeRePair&1.77\%&0.20\%&35&109ms&1\\*
			BPLEX&2.19\%&0.25\%&27&2.018s&295\\*
			E-Repair&1.68\%&0.24\%&37&4.578s&108\\*
			bzip2&-&0.29\%&-&229ms&4\\*
			gzip&-&0.81\%&-&8ms&-\\*
			XMill&-&0.24\%&-&2.728s&2\\
			\midrule\multicolumn{6}{c}{\emph{catalog-01}}\\*
			TreeRePair&1.76\%&0.10\%&279&898ms&2\\*
			BPLEX&2.23\%&0.14\%&342&6.834s&315\\*
			E-Repair&2.77\%&0.19\%&236&11s&349\\*
			bzip2&-&0.24\%&-&2.701s&8\\*
			gzip&-&0.85\%&-&51ms&-\\*
			XMill&-&0.11\%&-&12s&2\\
			\midrule\multicolumn{6}{c}{\emph{catalog-02}}\\*
			TreeRePair&1.12\%&0.07\%&770&10s&10\\*
			BPLEX&1.27\%&0.08\%&948&32s&512\\*
			E-Repair&1.49\%&0.12\%&1692&47s&521\\*
			bzip2&-&0.23\%&-&28s&8\\*
			gzip&-&0.81\%&-&450ms&-\\*
			XMill&-&0.09\%&-&1m 58s&12\\
			\midrule\multicolumn{6}{c}{\emph{dblp}}\\*
			TreeRePair&4.03\%&0.58\%&14533&43s&227\\*
			BPLEX&4.52\%&0.65\%&11693&61m 15s&1644\\*
			E-Repair&5.52\%&0.68\%&35125&42m 48s&516\\*
			bzip2&-&0.56\%&-&1m 11s&8\\*
			gzip&-&1.30\%&-&1.230s&-\\*
			XMill&-&0.53\%&-&11m 36s&15\\
			\midrule\multicolumn{6}{c}{\emph{dictionary-01}}\\*
			TreeRePair&8.08\%&1.47\%&930&1.117s&9\\*
			BPLEX&9.67\%&1.85\%&1044&46s&323\\*
			E-Repair&8.51\%&1.81\%&1428&19s&462\\*
			bzip2&-&1.52\%&-&1.313s&7\\*
			gzip&-&3.07\%&-&39ms&-\\*
			XMill&-&1.49\%&-&17s&2\\
			\midrule\multicolumn{6}{c}{\emph{dictionary-02}}\\*
			TreeRePair&6.15\%&1.32\%&5024&11s&69\\*
			BPLEX&7.56\%&1.63\%&5424&6m 12s&587\\*
			E-Repair&8.30\%&1.81\%&13698&1m 57s&475\\*
			bzip2&-&1.52\%&-&15s&7\\*
			gzip&-&3.05\%&-&279ms&-\\*
			XMill&-&1.49\%&-&2m 41s&13\\
			\midrule\multicolumn{6}{c}{\emph{EnWikiNew}}\\*
			TreeRePair&2.38\%&0.20\%&390&1.721s&8\\*
			BPLEX&2.63\%&0.23\%&335&35s&337\\*
			E-Repair&2.42\%&0.24\%&476&12s&369\\*
			bzip2&-&0.26\%&-&2.999s&8\\*
			gzip&-&0.90\%&-&57ms&-\\*
			XMill&-&0.23\%&-&23s&2\\
			\midrule\multicolumn{6}{c}{\emph{EnWikiQuote}}\\*
			TreeRePair&2.51\%&0.20\%&274&1.195s&7\\*
			BPLEX&2.81\%&0.23\%&236&25s&321\\*
			E-Repair&2.58\%&0.26\%&323&10s&268\\*
			bzip2&-&0.28\%&-&2.013s&8\\*
			gzip&-&0.93\%&-&36ms&-\\*
			XMill&-&0.24\%&-&15s&2\\
			\midrule\multicolumn{6}{c}{\emph{EnWikiSource}}\\*
			TreeRePair&1.14\%&0.10\%&515&5.025s&26\\*
			BPLEX&1.40\%&0.13\%&535&1m 10s&418\\*
			E-Repair&1.82\%&0.18\%&1127&23s&488\\*
			bzip2&-&0.16\%&-&8.742s&8\\*
			gzip&-&0.63\%&-&131ms&-\\*
			XMill&-&0.12\%&-&1m 4s&9\\
			\midrule\multicolumn{6}{c}{\emph{EnWikiVersity}}\\*
			TreeRePair&1.50\%&0.12\%&303&2.244s&12\\*
			BPLEX&1.70\%&0.15\%&287&36s&347\\*
			E-Repair&1.61\%&0.15\%&423&13s&415\\*
			bzip2&-&0.19\%&-&3.698s&8\\*
			gzip&-&0.69\%&-&59ms&-\\*
			XMill&-&0.15\%&-&28s&2\\
			\midrule\multicolumn{6}{c}{\emph{EnWikTionary}}\\*
			TreeRePair&1.00\%&0.11\%&2575&37s&183\\*
			BPLEX&1.15\%&0.13\%&2062&9m 13s&1287\\*
			E-Repair&1.48\%&0.15\%&6314&1m 40s&526\\*
			bzip2&-&0.17\%&-&57s&8\\*
			gzip&-&0.68\%&-&938ms&-\\*
			XMill&-&0.13\%&-&7m 25s&15\\
			\midrule\multicolumn{6}{c}{\emph{EXI-Array}}\\*
			TreeRePair&0.44\%&0.03\%&75&1.393s&14\\*
			BPLEX&0.77\%&0.05\%&124&43s&322\\*
			E-Repair&0.51\%&0.05\%&155&7.833s&312\\*
			bzip2&-&0.05\%&-&3.250s&8\\*
			gzip&-&0.37\%&-&67ms&-\\*
			XMill&-&0.03\%&-&10s&6\\
			\midrule\multicolumn{6}{c}{\emph{EXI-factbook}}\\*
			TreeRePair&2.51\%&0.31\%&99&356ms&2\\*
			BPLEX&6.44\%&0.58\%&170&5.333s&298\\*
			E-Repair&2.59\%&0.31\%&151&12s&438\\*
			bzip2&-&0.78\%&-&854ms&8\\*
			gzip&-&1.10\%&-&17ms&-\\*
			XMill&-&0.29\%&-&5.248s&1\\
			\midrule\multicolumn{6}{c}{\emph{EXI-Invoice}}\\*
			TreeRePair&0.72\%&0.21\%&11&147ms&2\\*
			BPLEX&0.78\%&0.28\%&8&1.406s&293\\*
			E-Repair&0.91\%&0.24\%&21&4.320s&113\\*
			bzip2&-&0.30\%&-&191ms&3\\*
			gzip&-&0.64\%&-&7ms&-\\*
			XMill&-&0.26\%&-&1.256s&2\\
			\midrule\multicolumn{6}{c}{\emph{EXI-Telecomp}}\\*
			TreeRePair&0.08\%&0.01\%&12&829ms&3\\*
			BPLEX&0.07\%&0.02\%&15&9.548s&310\\*
			E-Repair&0.08\%&0.02\%&24&13s&450\\*
			bzip2&-&0.09\%&-&2.363s&8\\*
			gzip&-&0.45\%&-&36ms&-\\*
			XMill&-&0.02\%&-&11s&2\\
			\midrule\multicolumn{6}{c}{\emph{EXI-weblog}}\\*
			TreeRePair&0.06\%&0.01\%&9&400ms&3\\*
			BPLEX&0.05\%&0.01\%&12&9.004s&303\\*
			E-Repair&0.05\%&0.02\%&12&7.942s&288\\*
			bzip2&-&0.06\%&-&720ms&8\\*
			gzip&-&0.40\%&-&14ms&-\\*
			XMill&-&0.02\%&-&8.342s&2\\
			\midrule\multicolumn{6}{c}{\emph{JST\_gene.chr1}}\\*
			TreeRePair&1.91\%&0.10\%&227&906ms&3\\*
			BPLEX&2.42\%&0.13\%&211&11s&315\\*
			E-Repair&2.99\%&0.17\%&128&9.947s&211\\*
			bzip2&-&0.14\%&-&2.599s&8\\*
			gzip&-&0.67\%&-&43ms&-\\*
			XMill&-&0.10\%&-&14s&2\\
			\midrule\multicolumn{6}{c}{\emph{JST\_snp.chr1}}\\*
			TreeRePair&1.58\%&0.08\%&537&3.213s&8\\*
			BPLEX&2.45\%&0.14\%&569&32s&360\\*
			E-Repair&1.51\%&0.10\%&673&15s&453\\*
			bzip2&-&0.18\%&-&9.251s&8\\*
			gzip&-&0.79\%&-&149ms&-\\*
			XMill&-&0.09\%&-&40s&8\\
			\midrule\multicolumn{6}{c}{\emph{medline02n0328}}\\*
			TreeRePair&4.32\%&0.34\%&4923&16s&79\\*
			BPLEX&6.47\%&0.46\%&6717&5m 45s&574\\*
			E-Repair&6.71\%&0.54\%&13243&1m 38s&477\\*
			bzip2&-&0.49\%&-&31s&7\\*
			gzip&-&1.26\%&-&544ms&-\\*
			XMill&-&0.34\%&-&2m 13s&13\\
			\midrule\multicolumn{6}{c}{\emph{NCBI\_gene.chr1}}\\*
			TreeRePair&1.43\%&0.09\%&354&1.442s&4\\*
			BPLEX&3.00\%&0.16\%&464&14s&327\\*
			E-Repair&1.66\%&0.11\%&342&10s&265\\*
			bzip2&-&0.15\%&-&4.110s&8\\*
			gzip&-&0.71\%&-&65ms&-\\*
			XMill&-&0.08\%&-&21s&8\\
			\midrule\multicolumn{6}{c}{\emph{NCBI\_snp.chr1}}\\*
			TreeRePair&$<$\,0.01\%&$<$\,0.01\%&11&15s&80\\*
			BPLEX&$<$\,0.01\%&$<$\,0.01\%&15&2m 6s&770\\*
			E-Repair&0.03\%&0.01\%&292&33s&465\\*
			bzip2&-&0.03\%&-&40s&8\\*
			gzip&-&0.39\%&-&578ms&-\\*
			XMill&-&0.00\%&-&3m 45s&14\\
			\midrule\multicolumn{6}{c}{\emph{sprot39.dat}}\\*
			TreeRePair&2.41\%&0.37\%&11699&43s&178\\*
			BPLEX&4.33\%&0.53\%&11783&13m 43s&1446\\*
			E-Repair&4.25\%&0.59\%&33700&3m 59s&497\\*
			bzip2&-&0.45\%&-&1m 11s&8\\*
			gzip&-&1.20\%&-&1.122s&-\\*
			XMill&-&0.36\%&-&9m 52s&15\\
			\midrule\multicolumn{6}{c}{\emph{treebank}}\\*
			TreeRePair&21.59\%&4.28\%&17186&22s&164\\*
			BPLEX&26.21\%&5.37\%&21302&21m 36s&646\\*
			E-Repair&34.53\%&6.01\%&51470&7m 44s&514\\*
			bzip2&-&5.26\%&-&6.407s&7\\*
			gzip&-&9.65\%&-&843ms&-\\*
			XMill&-&4.51\%&-&1m 36s&12\\
			\bottomrule
		\end{longtable}

\newpage

\subsection{Without Using DAG Representation}\label{sec:detailedResultsNoDag}

\begin{longtable}{lrrrrr}
			\toprule
\textbf{Algorithm}&\textbf{Edges}&\textbf{File size}&\textbf{\#NTs}&\textbf{Time}&\textbf{Mem (MB)}\\\midrule
\endhead
			\toprule
\textbf{Algorithm}&\textbf{Edges}&\textbf{File size}&\textbf{\#NTs}&\textbf{Time}&\textbf{Mem (MB)}\\
\endfirsthead
			\midrule\multicolumn{6}{c}{\emph{1998statistics}}\\*
			Without DAG&1.62\%&0.20\%&53&121ms&4\\*
			With DAG&1.68\%&0.20\%&54&214ms&1\\
			\midrule\multicolumn{6}{c}{\emph{catalog-01}}\\*
			Without DAG&1.69\%&0.10\%&400&1.381s&20\\*
			With DAG&1.69\%&0.10\%&400&1.022s&3\\
			\midrule\multicolumn{6}{c}{\emph{catalog-02}}\\*
			Without DAG&1.11\%&0.07\%&967&15s&199\\*
			With DAG&1.11\%&0.07\%&965&9.584s&10\\
			\midrule\multicolumn{6}{c}{\emph{dblp}}\\*
			Without DAG&3.89\%&0.59\%&25039&55s&1015\\*
			With DAG&3.89\%&0.59\%&25250&44s&227\\
			\midrule\multicolumn{6}{c}{\emph{dictionary-01}}\\*
			Without DAG&7.63\%&1.51\%&1622&1.238s&25\\*
			With DAG&7.72\%&1.54\%&1676&1.044s&9\\
			\midrule\multicolumn{6}{c}{\emph{dictionary-02}}\\*
			Without DAG&5.88\%&1.36\%&9390&12s&238\\*
			With DAG&5.92\%&1.38\%&9757&11s&69\\
			\midrule\multicolumn{6}{c}{\emph{EnWikiNew}}\\*
			Without DAG&2.28\%&0.21\%&656&2.042s&37\\*
			With DAG&2.29\%&0.21\%&667&1.732s&8\\
			\midrule\multicolumn{6}{c}{\emph{EnWikiQuote}}\\*
			Without DAG&2.41\%&0.21\%&458&1.320s&24\\*
			With DAG&2.42\%&0.21\%&452&1.223s&7\\
			\midrule\multicolumn{6}{c}{\emph{EnWikiSource}}\\*
			Without DAG&1.09\%&0.10\%&863&5.652s&101\\*
			With DAG&1.10\%&0.10\%&861&5.087s&26\\
			\midrule\multicolumn{6}{c}{\emph{EnWikiVersity}}\\*
			Without DAG&1.43\%&0.13\%&522&2.472s&45\\*
			With DAG&1.44\%&0.13\%&525&2.229s&12\\
			\midrule\multicolumn{6}{c}{\emph{EnWikTionary}}\\*
			Without DAG&0.97\%&0.11\%&4539&42s&743\\*
			With DAG&0.97\%&0.11\%&4535&38s&183\\
			\midrule\multicolumn{6}{c}{\emph{EXI-Array}}\\*
			Without DAG&0.40\%&0.03\%&122&1.378s&21\\*
			With DAG&0.41\%&0.03\%&123&1.394s&14\\
			\midrule\multicolumn{6}{c}{\emph{EXI-factbook}}\\*
			Without DAG&2.34\%&0.31\%&144&331ms&6\\*
			With DAG&2.35\%&0.31\%&145&330ms&2\\
			\midrule\multicolumn{6}{c}{\emph{EXI-Invoice}}\\*
			Without DAG&0.61\%&0.21\%&12&85ms&3\\*
			With DAG&0.68\%&0.21\%&14&124ms&1\\
			\midrule\multicolumn{6}{c}{\emph{EXI-Telecomp}}\\*
			Without DAG&0.06\%&0.01\%&17&1.132s&17\\*
			With DAG&0.07\%&0.01\%&21&850ms&3\\
			\midrule\multicolumn{6}{c}{\emph{EXI-weblog}}\\*
			Without DAG&0.05\%&0.01\%&10&607ms&10\\*
			With DAG&0.06\%&0.01\%&13&400ms&3\\
			\midrule\multicolumn{6}{c}{\emph{JST\_gene.chr1}}\\*
			Without DAG&1.73\%&0.09\%&299&1.365s&21\\*
			With DAG&1.84\%&0.10\%&354&910ms&3\\
			\midrule\multicolumn{6}{c}{\emph{JST\_snp.chr1}}\\*
			Without DAG&1.50\%&0.09\%&841&4.187s&59\\*
			With DAG&1.51\%&0.09\%&856&3.287s&8\\
			\midrule\multicolumn{6}{c}{\emph{medline02n0328}}\\*
			Without DAG&4.11\%&0.34\%&8524&17s&235\\*
			With DAG&4.13\%&0.35\%&9064&17s&79\\
			\midrule\multicolumn{6}{c}{\emph{NCBI\_gene.chr1}}\\*
			Without DAG&1.37\%&0.09\%&486&1.959s&32\\*
			With DAG&1.37\%&0.09\%&504&1.498s&4\\
			\midrule\multicolumn{6}{c}{\emph{NCBI\_snp.chr1}}\\*
			Without DAG&$<$\,0.01\%&$<$\,0.01\%&13&18s&337\\*
			With DAG&$<$\,0.01\%&$<$\,0.01\%&17&15s&80\\
			\midrule\multicolumn{6}{c}{\emph{sprot39.dat}}\\*
			Without DAG&2.31\%&0.37\%&18516&55s&936\\*
			With DAG&2.30\%&0.38\%&20224&44s&178\\
			\midrule\multicolumn{6}{c}{\emph{treebank}}\\*
			Without DAG&20.71\%&4.41\%&32786&14s&215\\*
			With DAG&20.72\%&4.41\%&32857&22s&164\\
			\bottomrule
		\end{longtable}

\end{document}